\documentclass[11pt]{article}

\usepackage[margin=1in]{geometry}
\usepackage{amsthm,amsmath,amsfonts,amssymb}
\usepackage[]{natbib}
\usepackage{graphicx}
\usepackage{lmodern}
\usepackage[T1]{fontenc}
\usepackage[dvipsnames]{xcolor}
\usepackage{multirow}
\usepackage{lscape}
\usepackage{enumitem}
\usepackage{array,mathrsfs}
\usepackage{algorithm,algpseudocode}
\usepackage{setspace}
\usepackage[colorlinks,citecolor=blue,linkcolor=blue,urlcolor=blue]{hyperref}

\theoremstyle{plain}
\newtheorem{Th}{Theorem}
\newtheorem{Pro}{Proposition}

\newtheorem{Lem}{Lemma}

\theoremstyle{definition}
\newtheorem{Rem}{Remark}

\def\calT{{\cal T}}

\def\S{{\bf S}}
\def\Z{{\bf Z}}
\def\a{{\bf a}}
\def\bfe{{\bf e}}

\def\z{{\bf z}}
\def\0{{\bf 0}}

\def\d{{\lambda}}
\def\dd{{\rm d}}
\def\wt{\widetilde}
\providecommand{\widecheck}[1]{%
  \reflectbox{\ensuremath{\widehat{\reflectbox{\ensuremath{#1}}}}}}
\def\wc{\widecheck}
\def\wh{\widehat}
\def\sumi{\sum_{i=1}^n}

\def\log{\hbox{log}}

\def\var{\hbox{var}}

\def\trace{\hbox{trace}}
\def\pr{\hbox{pr}}
\def\Normal{\hbox{Normal}}
\def\TruncNormal{\hbox{TruncNormal}}

\def\trans{^{\rm T}}

\def\n{\nonumber}

\def\bse{\begin{eqnarray*}}
\def\ese{\end{eqnarray*}}
\def\be{\begin{eqnarray}}
\def\ee{\end{eqnarray}}

\def\eff{_{\rm eff}}

\def\SCP{_{\rm SCP}}
\def\FCP{_{\rm FCP}}
\def\JK{_{\rm JK+}}

\def\ba{{\mathbf a}}
\def\bA{{\mathbf A}}
\def\bB{{\mathbf B}}

\def\bD{{\mathbf D}}

\def\b1e{{\mathbf e}}

\def\bg{{\boldsymbol\gamma}}

\def\bh{{\mathbf h}}

\def\bo{{\mathbf o}}
\def\bO{{\mathbf O}}

\def\bv{{\boldsymbol{v}}}

\def\bz{{\mathbf z}}
\def\bZ{{\mathbf Z}}

\def\bb{{\boldsymbol\beta}}

\def\bpsi{{\boldsymbol{\psi}}}

\def\bSigma{{\boldsymbol {\Sigma}}}

\newcommand{\bmu}{{\boldsymbol\mu}}

\newcommand{\balpha}{{\boldsymbol\alpha}}
\newcommand{\bPhi}{{\boldsymbol\Phi}}

\newcommand{\brho}{{\boldsymbol\rho}}

\newcommand{\bxi}{\mbox{\boldmath $\xi$}}

\def\red{\color{red}}
\newcommand{\indep}{\rotatebox[origin=c]{90}{$\models$}}

\begin{document}

\def\spacingset#1{\renewcommand{\baselinestretch}%
{#1}\small\normalsize} \spacingset{1}

\title{\bf PRESCCO: Efficient Prediction Intervals under a
Right-Censored Covariate}

\author{
Kihyun Han\thanks{Department of Statistics, Pennsylvania State
University. E-mail: \texttt{kqh5716@psu.edu}}
\and
Yanyuan Ma\thanks{Department of Statistics, Pennsylvania State
University. E-mail: \texttt{yanyuanma@gmail.com}}
\and
Karen Marder\thanks{Department of Neurology, Columbia University
Medical Center. E-mail: \texttt{ksm1@columbia.edu}}
\and
Tanya P. Garcia\thanks{Department of Biostatistics, Gillings School of
Global Public Health, University of North Carolina at Chapel Hill.
E-mail: \texttt{tpgarcia@email.unc.edu}}
}

%\date{\today}
\date{}

\maketitle

\begin{abstract}
In neurodegenerative diseases, atypical patients can reveal
disease-modifying biology. A patient is atypical when their outcome
falls outside a prediction interval for patients at the same point in
the disease, measured by the time to a disease event. Most patients
leave the study before the event, so their time to it is
right-censored, and no method computes a prediction interval with a
right-censored covariate. We develop the first methods. We adapt three
conformal prediction methods and show that their estimated
half-length, half the interval's width, varies across studies,
changing which patients are identified as atypical. We then develop
the PRESCCO method. We prove its estimator of the half-length is
doubly robust, consistent when one of its time-to-event and censoring
models is misspecified, and semiparametrically efficient when both are
correct. At a $90\%$ prediction level, its prediction interval
identifies close to $10\%$ of patients as atypical, across studies
and, with high probability, in a single study. In simulations, its
estimated half-length is four to thirty times less variable than under
conformal prediction. In Huntington disease data with $77.2\%$
right-censoring, the PRESCCO method identifies close to the targeted
$10\%$ of patients as atypical, whereas conformal prediction identifies
too many or too few.
\end{abstract}

\noindent%

\spacingset{1.3}

%%%%%%%%%%%%%%%%%%%%%%%%%%%%%%%%%%%%%%%%
\section{Introduction}
%%%%%%%%%%%%%%%%%%%%%%%%%%%%%%%%%%%%%%%%

Several promising leads for slowing neurodegenerative diseases have
come from atypical patients, those who fare far better or far worse
than others with the disease \citep{gemhd2015, marino2024antibodies}.
A trait that sets these patients apart, such as a gene variant, can
reveal biology that slows or speeds the disease and point to targets
for therapy. But first, these patients must be found. Researchers find
them by comparing an outcome that measures how a patient fares, such
as a cognitive test score, with a prediction
interval, the range of outcomes common among patients followed in an
observational study. Patients whose outcomes fall outside that
prediction interval, on either side, are atypical
\citep{braisch2017extreme, braisch2018sdmt}. Yet outcomes worsen as 
the disease advances, so an outcome that is common early in the disease
course may be uncommon later. To find atypical patients, the
prediction interval must therefore depend on where each patient is in
the disease course, and for most patients, that is not known.

Age might seem to indicate where a patient is in the disease course,
but two patients of the same age can be at very different points. Two
patients three years from a common disease event, such as the point at
which Huntington disease becomes clinically detectable or a diagnosis
of dementia in Alzheimer disease, are instead at a similar point
whatever their ages. The time to the event therefore indicates where a
patient is in the disease course
\citep{dempsey2018survival, kongetal2018, lee2024robust}, and the
prediction interval is computed as a function of this time, the
time-to-event covariate. Computing that prediction interval, however,
requires knowing the time to the event for every patient the
observational study followed. Neurodegenerative diseases advance over
years, and most patients leave the observational study before the
event. Their time to the event is right-censored, known only to exceed
how long they were followed. In this right-censored covariate setting,
unlike for a right-censored outcome \citep{candes2023conformalized},
no existing method computes a prediction interval.

This paper develops the first methods for such a prediction interval.
A prediction interval is determined by two quantities, the center and
the target half-length. The center is the outcome expected of patients
at the same point in the disease course. The target half-length is
half the width of the prediction interval, so it sets how far an
outcome can fall from the center before a patient is identified as
atypical. Existing methods for a right-censored covariate can be used
to estimate the center \citep{lee2024robust, vazquez2026estimators},
but not the target half-length. We first adapt conformal prediction,
the standard framework for prediction intervals
\citep{vovk2005algorithmic, lei2018distribution}, to estimate the
target half-length in the right-censored covariate setting. We show
that the resulting estimator of the target half-length has high
variance. The estimated half-length changes with which patients an
observational study happens to follow, and so do the patients
identified as atypical. In the one observational study a researcher
has, the estimated half-length can therefore be too short or too
long. Too short, and patients with common outcomes are identified as
atypical, and what sets them apart reveals nothing about the disease;
too long, and atypical patients are missed, along with the biology
they could reveal.

To avoid both errors, we develop the PRESCCO method (PREdiction with
Semiparametric efficiency under a right-Censored COvariate). The
PRESCCO method estimates the target half-length by modeling how long
patients take to reach the event and how long they spend in the
observational study. We prove that its estimator of the target half-length
is doubly robust,
remaining consistent even when one of the two models is misspecified,
and semiparametrically efficient, having the smallest possible
variance across observational studies when both models are correct.
We also prove that this efficiency costs no coverage. A researcher
chooses the prediction level, say $90\%$, so that a patient is
identified as atypical only when their outcome is among the $10\%$
farthest from the center. Most prediction intervals guarantee this
coverage, containing close to $90\%$ of outcomes, only on average
across observational studies. The prediction interval from the PRESCCO
method guarantees it both on average and, with high probability, in
the one observational study a researcher has.  A patient is then
identified as atypical because their outcome differs from the outcomes
of others at the same point in the disease course, not because of
which patients the observational study happened to follow.

%%%%%%%%%%%%%%%%%%%%%%%%%
\section{Prediction intervals under a right-censored covariate}
\label{sec:setup}
%%%%%%%%%%%%%%%%%%%%%

The prediction interval is computed from the data of an observational
study that enrolls $n$ patients. At study entry, the observational
study records each patient's outcome $Y\in\mathbb{R}$, such as a
cognitive test score, and fully observed covariates
$\bZ\in\mathbb{R}^{d_{\Z}}$, where $d_{\Z}\in\mathbb{N}\cup\{0\}$,
such as genetic risk, sex, and education. The observational study then
follows each patient until they reach the event or leave the
observational study, which ends their follow-up. The time a patient
takes to reach the event is the time-to-event covariate
$X\in\mathbb{R}$, and the time they spend in the observational study is
the censoring time $C\in\mathbb{R}$, both measured from study entry.
The observational study records $W\equiv\min(X,C)$, whichever time
comes first, and $\Delta\equiv I(X\le C)$, equal to one when the
patient reached the event while still in the observational study. A
patient with $\Delta=1$ has $X=W$. A patient with $\Delta=0$ left the
observational study before reaching the event, and their time to the
event is right-censored, known only to exceed $W$. The full data are
$\bD=(Y,X,C,\bZ)$, and the observed data are $\bO=(Y,W,\Delta,\bZ)$.

We assume that the $n$ patients are drawn independently from the
population, so $\bO_1,\ldots,\bO_n$ are independent and identically
distributed. The population consists of the patients the
observational study could have enrolled and followed in the same way.
Four models describe the population. The outcome model
$f_{Y|X,\bZ}$, parametrized by $\bb\in\mathbb{R}^{d_\bb}$, describes
how a patient's outcome relates to the time they take to reach the
event. The time-to-event model $f_{X|\bZ}$ describes how long patients
take to reach the event, and the censoring model $f_{C|X,Y,\bZ}$
describes how long they spend in the observational study. The fully
observed covariate model $f_{\bZ}$ describes how the fully observed
covariates are distributed. We assume that, among patients with the
same fully observed covariates, how long a patient spends in the
observational study carries no information about their time to the
event or their outcome, that is, $C\indep(X,Y)\mid\bZ$, or
equivalently $f_{C|X,Y,\bZ}=f_{C|\bZ}$. This noninformative censoring
mechanism is standard in the right-censored covariate setting
\citep{lee2024robust, vazquez2026estimators}, and we write the
censoring model as $f_{C|\bZ}$ hereafter.

Consider now a new patient, whose outcome is compared with the
prediction interval to identify whether they are atypical. The new
patient's data $\bO_0=(Y_0,W_0,\Delta_0,\bZ_0)$ are distributed as
$\bO_1,\ldots,\bO_n$ and independent of them. Like the $n$ patients,
the new patient has $Y_0$ and $\bZ_0$ recorded at study entry, and
$W_0$ and $\Delta_0$ known only after follow-up ends. Unlike most
prediction intervals, which are computed before the outcome is
observed, the prediction interval for $Y_0$ is retrospective,
computed after follow-up ends. The reason is that whether $Y_0$ is
atypical depends on where the new patient was in the disease course
when $Y_0$ was recorded, and only $W_0$ and $\Delta_0$ indicate that
point. The prediction interval gives the range of outcomes common
among patients at that point, and the new patient is identified as
atypical when $Y_0$ falls outside it, on either side. The prediction
interval takes the form $m\pm\zeta$ and is determined by two
quantities, the center $m$, computed from $(W_0,\Delta_0,\bZ_0)$, and
the target half-length $\zeta$.

The center is the outcome expected at the new patient's point in the
disease course. If $X_0$ were always observed, the center would be
$m_0(X_0,\bZ_0,\bb)\equiv E(Y_0\mid X_0,\bZ_0,\bb)$. When
$\Delta_0=0$, however, $X_0$ is right-censored, and the first center
is instead the conditional mean of $Y_0$ given the observed
$(W_0,\Delta_0,\bZ_0)$, $m_1(W_0,\Delta_0,\bZ_0,\bb)\equiv E(Y_0\mid
W_0,\Delta_0,\bZ_0,\bb,f_{X|\bZ})$. Under the noninformative censoring
mechanism,
\bse
&&E(Y_0\mid W_0,\Delta_0,\bZ_0,\bb,f_{X|\bZ})\\
&=&\Delta_0 m_0(X_0,\bZ_0,\bb)+
(1-\Delta_0)\frac{E\{I(X_0 > C_0)m_0(X_0,\bZ_0,\bb)\mid C_0,\bZ_0,f_{X|\bZ}\}}{E\{I(X_0 > C_0)\mid C_0,\bZ_0,f_{X|\bZ}\}}.
\ese
When $\Delta_0=1$, $m_1$ recovers $m_0$ at $X_0=W_0$. When
$\Delta_0=0$, the new patient could be any time beyond $W_0$ from the
event, and $m_1$ averages $m_0$ over those times, weighting each by how
likely it is under $f_{X|\bZ}$. For example, a new patient with $W_0$
of two years and $\Delta_0=0$ could be three years from the event, or
four, or ten, and $m_1$ averages over all of these times. The center
$m_1$ is therefore the outcome expected at the new patient's point in
the disease course, whether or not their time to the event is
right-censored.

Computing $m_1$, however, requires the time-to-event model
$f_{X|\bZ}$, which is unknown in practice. The second center replaces
$f_{X|\bZ}$ with a working model $f_{X|\bZ}^*$ that a researcher
specifies, giving $m_1^*(W_0,\Delta_0,\bZ_0,\bb)\equiv E(Y_0\mid
W_0,\Delta_0,\bZ_0,\bb,f_{X|\bZ}^*)$. The third center uses no
time-to-event model and treats $W_0$ as though it were $X_0$, giving
$m_2(W_0,\Delta_0,\bZ_0,\bb)\equiv m_0(W_0,\bZ_0,\bb)$. When
$\Delta_0=1$, $m_2$ equals $m_1$. When $\Delta_0=0$, $m_2$ treats the
new patient as though they were exactly $W_0$ from the event. For
example, a new patient with $W_0$ of two years and $\Delta_0=0$ was
more than two years from the event, yet $m_2$ is the outcome expected
of patients two years from the event. We include $m_2$ for comparison,
to show what the time-to-event model contributes to the center. Other
centers are possible, but we restrict attention to these three. Each
depends on $\bb$, and we write the center as
$m(W_0,\Delta_0,\bZ_0,\bb)$ throughout.

The target half-length is how far $Y_0$ can fall from the center, on
either side, before the new patient is identified as atypical.
Whichever center is used, how far $Y_0$ falls from the center is the
new patient's residual $r(\bO_0,\bb)$, where
$r(\bo,\bb)\equiv|y-m(w,\delta,\z,\bb)|$ is the residual function. A
researcher sets how far $Y_0$ can fall by choosing the prediction
level $1-\alpha$, such as $0.9$. The target half-length is defined so
that the prediction interval contains $Y_0$ with probability
$1-\alpha$, or equivalently, so that the residual is at most $\zeta$
with probability $1-\alpha$,
\bse
\pr\{r(\bO_0,\bb)\le \zeta\}=1-\alpha.
\ese
We assume that the residual is a continuous random variable, and the
target half-length is then the $(1-\alpha)$ quantile of the
distribution of the residual over the population. At $1-\alpha=0.9$,
for example, $10\%$ of new patients have a residual larger than
$\zeta$ and are identified as atypical. The prediction interval that
identifies these atypical patients thus depends 
on the population through $\bb$ and $\zeta$,
and what separates the methods that follow is how they estimate
$\zeta$.

%%%%%%%%%%%%%%%%%%%%%%%%%%%%%%%%%%%%%%%%%%%%%%%%%%%%%%%%%%%%%%%%%%%%%%%%%%%%%%%%%
\section{Conformal prediction methods extended to the right-censored covariate setting}\label{sec:conformal}
%%%%%%%%%%%%%%%%%%%%%%%%%%%%%%%%%%%%%%%%%%%%%%%%%%%%%%%%%%%%%%%%%%%%%%%%%%%%%%%%%

One estimator of $\zeta$ is the $(1-\alpha)$ sample quantile of the
residuals, one for each patient the observational study followed. This
estimator is the basis of conformal prediction
\citep{vovk2005algorithmic, lei2018distribution}, and apart from the
center, it uses none of the four models that describe the population.
It requires only that $\bO_1,\ldots,\bO_n,\bO_0$ be exchangeable, and
right-censoring leaves that requirement intact: the observed data are
independent and identically distributed, and the residual $r(\bO,\bb)$ is computed from them alone, without
any patient's time to the event. We adapt three widely used conformal
prediction methods to the right-censored covariate setting, split
conformal prediction, full conformal prediction, and jackknife+, which
differ in how they compute the residuals.

The three methods produce different prediction intervals, which we
compare with three criteria. The first is the expected coverage rate
(also known as ``marginal coverage''), 
a standard criterion for analyzing coverage
\citep{cox1975prediction, tian2022methods}.
Repeating the observational study with different patients from the
population would give a different prediction interval, and the
expected coverage rate is the probability that the prediction interval
contains $Y_0$, averaged over such repetitions. The second is the
coverage rate (also known as ``training-conditional coverage''), the
probability that the prediction interval contains $Y_0$ given the data
$\bO_1,\ldots,\bO_n$ from the one observational study a researcher has
\citep{vovk2012conditional}. One minus the coverage rate is the
proportion of new patients that the prediction interval from that
observational study identifies as atypical. The third is how much the
estimated half-length varies across repetitions of the observational
study. If the estimated half-length varies widely, the same new
patient can be identified as atypical in one observational study and
not in another.

%%%%%%%%%%%%%%%%%%%%%%%%%%%%%%%%%%%%%%%%%%%%
\subsection{Split conformal prediction}\label{sec:split-conformal-prediction}
%%%%%%%%%%%%%%%%%%%%%%%%%%%%%%%%%%%%%%%%%%%%

Split conformal prediction divides $\bO_1,\ldots,\bO_n$ into a training
set $\{\bO_1,\ldots,\bO_{n_1}\}$ and a calibration set
$\{\bO_{n_1+1},\ldots,\bO_n\}$. The training set is used to estimate
$\bb$ with an existing estimator under a right-censored covariate
\citep{lee2024robust, vazquez2026estimators}, giving $\wh\bb$, and
the residuals $r(\bO_i,\wh\bb)$ are then computed for the calibration
set. The calibration set takes no part in $\wh\bb$, and the
calibration residuals are therefore exchangeable with the new
patient's residual, conditional on $\wh\bb$. The estimator of the
target half-length is
\be
\wh\zeta\SCP \equiv
Q_{1-\alpha}\{r(\bO_{n_1+1},\wh\bb),\ldots,r(\bO_n,\wh\bb),\infty\},
\label{eq:e16}
\ee
where $Q_{1-\alpha}(\cdot)$ is the $(1-\alpha)$ sample quantile of the
elements of the set, and $\infty$ stands in for the new patient's
residual, which is not used in computing the prediction interval. The prediction interval for the new patient is
$m(W_0,\Delta_0,\bZ_0,\wh\bb)\pm\wh\zeta\SCP$.

Exchangeability conditional on $\wh\bb$ is all that split conformal
prediction requires to bound the expected coverage rate
\citep{lei2018distribution}. If $r(\bO_0,\wh\bb)$ has a continuous
distribution conditional on $\wh\bb$, then
$1-\alpha\le\pr\{r(\bO_0,\wh\bb)\le \wh\zeta\SCP\mid
\wh\bb\}\le1-\alpha+(n-n_1+1)^{-1}$, where the probability is taken
over the new patient's data and the calibration set. The bounds hold
at every value of $\wh\bb$, and hence for the expected coverage rate.
The guarantee is finite-sample, holding at every $n$ and under any
distribution for the population.
The coverage rate also
falls below the prediction level by at most a term of order
$(n-n_1)^{-1/2}$ with high probability \citep{vovk2012conditional}.

The third criterion, how much $\wh\zeta\SCP$ varies across repetitions
of the observational study, is not covered by either guarantee. As a
sample quantile of the calibration residuals, $\wh\zeta\SCP$ has the
standard asymptotic distribution for sample quantiles, conditional on
$\wh\bb$.
\begin{Th}\label{th:scp}
Let $\zeta\SCP$ satisfy
$\pr\{r(\bO_0,\wh\bb)\le \zeta\SCP\mid \wh\bb\}=1-\alpha$.
For $\wh\zeta\SCP$ defined in \eqref{eq:e16},
\bse
(n-n_1)^{1/2}(\wh\zeta\SCP-\zeta\SCP)\mid \wh\bb
\stackrel{d}{\to}
\Normal\left(
0,
\alpha(1-\alpha)
E[\d\{\zeta\SCP-r(\bO,\wh\bb)\}\mid \wh\bb]^{-2}
\right),
\ese
where $\d$ is the Dirac delta function, so that
$E[\d\{\zeta\SCP-r(\bO,\wh\bb)\}\mid \wh\bb]$ is the conditional density
of $r(\bO,\wh\bb)$ at $\zeta\SCP$.
\end{Th}
The proof of Theorem~\ref{th:scp} is in Section~\ref{sec:scp}. The
asymptotic variance in Theorem~\ref{th:scp} is large for two reasons.
The first is the scaling $(n-n_1)^{1/2}$ rather than $n^{1/2}$: only
the calibration set contributes residuals, and the training set enters
$\wh\zeta\SCP$ only through $\wh\bb$. The second is the density in the
denominator, the density of the residual at $\zeta\SCP$. The target
half-length lies in the tail of the distribution of the residual,
where the density is small and few residuals fall, and the asymptotic
variance grows as the density shrinks.

%%%%%%%%%%%%%%%%%%%%%%%%%%%%%%%%%%%%%%%%%%%%
\subsection{Full conformal prediction}\label{sec:full-conformal-prediction}
%%%%%%%%%%%%%%%%%%%%%%%%%%%%%%%%%%%%%%%%%%%%

Full conformal prediction does not divide $\bO_1,\ldots,\bO_n$, and
every patient the observational study followed contributes a residual.
It instead considers in turn each candidate value $y$ that $Y_0$ might
take, treating the new patient's data as though their outcome were
$y$, written $\bO_0(y)\equiv(y,W_0,\Delta_0,\bZ_0)$. The augmented set
$\{\bO_0(y),\bO_1,\ldots,\bO_n\}$ is used to estimate $\bb$, giving
$\wh\bb(y)$, and the residuals are computed for all $n+1$ elements of
the augmented set. All $n+1$ enter $\wh\bb(y)$ alike, and the $n+1$
residuals are therefore exchangeable, conditional on $\wh\bb(y)$. The
estimator of the target half-length is $\wh\zeta\FCP(y)\equiv
Q_{1-\alpha}[r\{\bO_0(y),\wh\bb(y)\},r\{\bO_1,\wh\bb(y)\},\ldots,
r\{\bO_n,\wh\bb(y)\}]$, and the prediction interval for the new patient
is $[y:r\{\bO_0(y),\wh\bb(y)\}\le\wh\zeta\FCP(y)]$, the candidate values
$y$ whose residual does not exceed $\wh\zeta\FCP(y)$. Unlike the
prediction interval from split conformal prediction, this prediction
interval does not take the form $m\pm\wh\zeta\FCP$, since its center
and its estimated half-length both change with $y$.

Like split conformal prediction, full conformal prediction requires
only exchangeability, here conditional on $\wh\bb(y)$, to bound the
expected coverage rate \citep{vovk2005algorithmic}. If the residuals
again have a continuous distribution, then
$1-\alpha\le\pr[r\{\bO_0,\wh\bb(Y_0)\}\le\wh\zeta\FCP(Y_0)]
\le1-\alpha+(n+1)^{-1}$, where the probability is taken over the new
patient's data and $\bO_1,\ldots,\bO_n$. The bounds are evaluated at
the candidate value equal to the new patient's own outcome, for which
$\bO_0(Y_0)=\bO_0$, and we write $\wh\bb(\bO_0)$ and
$\wh\zeta\FCP(\bO_0)$ for $\wh\bb(Y_0)$ and $\wh\zeta\FCP(Y_0)$, respectively. The
guarantee is also finite-sample, and its upper bound is tighter than
split conformal prediction's, which depends on the size of the
calibration set alone. Unlike that of split conformal prediction, the
coverage rate of full conformal prediction can be arbitrarily low without further
assumptions \citep{bian2023training}, and one is recovered under
algorithmic stability \citep{liang2025algorithmic}.

No asymptotic distribution comparable to Theorem~\ref{th:scp} is
available for full conformal prediction. Unlike $\wh\zeta\SCP$, which
is a single number, $\wh\zeta\FCP(y)$ changes with the candidate value
and with the new patient, and each of its values is computed at a different $\wh\bb(y)$. The
variation of the estimated half-length across repetitions of the
observational study is measured instead for $\wh\zeta\FCP$ averaged
over new patients.  That average still comes
from quantiles of residuals in the tail of their distribution, the
same source of the large asymptotic variance in Theorem~\ref{th:scp}.
The estimated half-length from full conformal prediction can therefore
be expected to vary widely across repetitions as well.

%%%%%%%%%%%%%%%%%%%%%%%%%%%%%%%%%%%%%%%%%%%%
\subsection{Jackknife+}\label{sec:jackknife}
%%%%%%%%%%%%%%%%%%%%%%%%%%%%%%%%%%%%%%%%%%%%

Jackknife+ avoids both dividing $\bO_1,\ldots,\bO_n$ and recomputing
$\wh\bb$ at every candidate value, leaving out the data from one
patient at a time instead. For each $i=1,\ldots,n$, the data with
$\bO_i$ omitted are used to estimate $\bb$, giving $\wh\bb_{-i}$, and
the residual of $\bO_i$ is computed at $\wh\bb_{-i}$, giving
$\wh r_i\equiv r(\bO_i,\wh\bb_{-i})$. The $n$ leave-one-out estimates
also give $n$ centers for the new patient,
$\wh m_{-i}\equiv m(W_0,\Delta_0,\bZ_0,\wh\bb_{-i})$. The prediction
interval for the new patient pairs each center with its own residual,
$[l\JK(\bO_0),u\JK(\bO_0)]\equiv[Q_{\alpha}\{\wh m_{-i}-\wh r_i:
i=1,\ldots,n\},Q_{1-\alpha}\{\wh m_{-i}+\wh r_i:i=1,\ldots,n\}]$, and
the estimator of the target half-length is half its width,
$\wh\zeta\JK(\bO_0)\equiv\{u\JK(\bO_0)-l\JK(\bO_0)\}/2$.

Like the other two methods, jackknife+ requires only that
$\bO_1,\ldots,\bO_n,\bO_0$ be exchangeable to bound the expected
coverage rate \citep{barber2021predictive}. Each residual, however, is
computed at a different estimate of $\bb$, and the residuals are not
exchangeable with the new patient's in the way the other two methods'
residuals are. The bound is therefore weaker,
$\pr\{Y_0\in[l\JK(\bO_0),u\JK(\bO_0)]\}\ge1-2\alpha$, where the
probability is taken over the new patient's data and
$\bO_1,\ldots,\bO_n$, with $1-2\alpha$ in place of $1-\alpha$ and no
upper bound. The guarantee is also finite-sample. The coverage rate of jackknife+, like that of full conformal
prediction, can be arbitrarily low without further assumptions
\citep{bian2023training}, and one is recovered under algorithmic
stability \citep{liang2025algorithmic}.

Jackknife+ also lacks an asymptotic distribution for its estimated
half-length, and for similar reasons. As in full conformal prediction,
$\wh\zeta\JK(\bO_0)$ changes with the new patient, through the centers
$\wh m_{-i}$, and each of the $n$ terms it combines is computed at a
different $\wh\bb_{-i}$. How much $\wh\zeta\JK$ varies across
repetitions of the observational study is again measured after
averaging over new patients, and that average comes from tail
quantiles, here of $\wh m_{-i}\pm\wh r_i$. The estimated half-length
from jackknife+ can therefore be expected to vary widely as well.

%%%%%%%%%%%%%%%%%%%%%%%%%%%%%%%%%%%%%%%%%%%%%%%%%%%%%%%%%%%%%%%%%%%%%%%%%%%%%%%%%%%%
\section{The PRESCCO method}\label{sec:pred}
%%%%%%%%%%%%%%%%%%%%%%%%%%%%%%%%%%%%%%%%%%%%%%%%%%%%%%%%%%%%%%%%%%%%%%%%%%%%%%%%%%%%

%%%%%%%%%%%%%%%%%%%%%%%%%%%%%%%%%%%%%%%%%%%%%%%%%%%%%%%%%%%%%%%%%%%%%%%%%%%%%%%%%
%\subsection{Recasting the construction of the prediction interval as semiparametric estimation}
\subsection{The estimator of the target half-length}
\label{sec:semiparam}
%%%%%%%%%%%%%%%%%%%%%%%%%%%%%%%%%%%%%%%%%%%%%%%%%%%%%%%%%%%%%%%%%%%%%%%%%%%%%%%%%

Another estimator of $\zeta$ solves the equation that defines it,
$\pr\{r(\bO_0,\bb)\le\zeta\}=1-\alpha$. The probability is taken over
the population, and computing it at a candidate value of $\zeta$
requires the four models that describe the population: the outcome
model $f_{Y|X,\bZ}$, the time-to-event model $f_{X|\bZ}$, the
censoring model $f_{C|\bZ}$, and the fully observed covariate model
$f_{\bZ}$. Under the noninformative censoring mechanism, the four
models combine into the observed-data likelihood
\be
 f_{Y,W,\Delta,\bZ}(y,w,\delta,\bz,\bb)&=& f_{\bZ}(\bz)\left\{f_{Y|X,\bZ}(y,w,\bz,\bb)f_{X|\bZ}(w,\bz)\int_{w}^{\infty}
f_{C|\bZ}(c,\bz)
dc\right\}^{\delta}\n\\
&&\times
\left\{f_{C|\bZ}(w,\bz)\int_{w}^{\infty}f_{Y|X,\bZ}(y,x,\bz,\bb)f_{X|\bZ}(x,\bz)dx
\right\}^{1-\delta},\label{eq:llhd-structure}
\ee
where the first integral is over the censoring times beyond $w$ of a
patient who reached the event, and the second is over the times to the
event beyond $w$ of a patient who left the observational study first.
The four models are identifiable, each uniquely determined by the
observed-data likelihood \citep[Lemma~1]{lee2024robust}, and $\zeta$
can therefore be estimated from $\bO_1,\ldots,\bO_n$ even though most
patients' times to the event are right-censored.

The most direct estimator of $\zeta$ specifies all four models,
computes $\pr\{r(\bO_0,\bb)\le\zeta\}$ from them, and sets that
probability equal to $1-\alpha$ to solve for $\zeta$. Yet the
time-to-event model, the censoring model, and the fully observed
covariate model are often difficult to specify correctly, and
misspecifying any one of them risks creating an inconsistent estimator
of $\zeta$. We instead estimate $\zeta$ with semiparametric theory
\citep{tsiatis2006semiparametric}, leaving the
form of these three models unrestricted and treating them as nuisance
models, $\eta_1\equiv f_{X|\bZ}$, $\eta_2\equiv f_{C|\bZ}$, and
$\eta_3\equiv f_{\bZ}$. Semiparametric theory guards against
misspecifying the nuisance models while keeping the variance of the
estimator as small as possible.

Estimating $\zeta$ also requires estimating $\bb$, and among the
existing estimators of $\bb$ in the right-censored covariate setting,
we use the SPARCC estimator \citep{lee2024robust}. The SPARCC
estimator is also derived from semiparametric theory and from the same
observed-data likelihood in \eqref{eq:llhd-structure}, yet its
derivation cannot yield an estimator of $\zeta$. Semiparametric theory
constructs the estimator with the smallest possible asymptotic
variance from the efficient influence function
\citep{tsiatis2006semiparametric}, and how that function is found
depends on how the quantity being estimated enters the observed-data
likelihood. The parameter $\bb$ enters it through the outcome model
and has a score function, the derivative of the log of the likelihood
with respect to $\bb$. Projecting that score function onto the space
orthogonal to the nuisance tangent space keeps the component that the
nuisance models affect least and gives the efficient influence
function for $\bb$. The target half-length, in contrast, does not
enter the observed-data likelihood. It is a functional of the four
models, defined through a probability computed from all of them at
once, and has no score function to project. The efficient influence function for $\zeta$ is instead found in
the full tangent space, spanned by the score functions of all four
models rather than the three nuisance models alone. It includes a
term for $\bb$ and a term for each nuisance model, as
Proposition~\ref{pro:2} shows.
\begin{Pro}\label{pro:2}
The efficient influence function for $\zeta$ is
$\phi\eff(\bo,\zeta,\bb)\equiv\phi_\bb(\bo,\zeta,\bb)+\sum_{j=1}^3\phi_j(\bo,\zeta,\bb)$,
where
\bse
\phi_\bb(\bo,\zeta,\bb) &\equiv& \bfe\trans\wt\S_\bb(y,w,\delta,\bz,\bb),\\
\phi_1(\bo,\zeta,\bb)&\equiv&\delta a_1(w,\bz) + (1-\delta) \frac{E\{I(X> w)
    a_1(X,\bz)\mid y,\bz,\bb,\eta_1\}}{E\{I(X > w)\mid y,\bz,\bb,\eta_1\}},\\
\phi_2(\bo,\zeta,\bb)&\equiv& \delta \frac{E\{I(C \ge w) a_2(C,\bz)\mid \bz,\eta_2\}}{E\{I(C \ge w)\mid \bz,\eta_2\}} + (1-\delta)a_2 (w,\bz),\\
\phi_3(\bo,\zeta,\bb)&\equiv& a_3(\bz),
\ese
$\wt\S_\bb$ is the efficient score for $\bb$ given in
Section~\ref{sec:supp-tangent}, and
$E[\d\{\zeta - r(\bO,\bb)\}\mid\bb,\eta_1,\eta_2,\eta_3]$, with $\d$
the Dirac delta function, is the density of the residual $r(\bO,\bb)$
at $\zeta$.  Here $a_1(x,\bz)$ satisfies $E\{a_1(X,\bz)\mid \bz,\eta_1\} = 0$ and
\bse
&& E\left[\Delta a_1(W,\bz) + (1-\Delta) \frac{E\{I(X> W)
    a_1(X,\bz)\mid W,Y,\bz,\bb,\eta_1\}}{E\{I(X > W)\mid W,Y,\bz,\bb,\eta_1\}}\mid x,\bz,\bb,\eta_2\right]\\
&=& - \frac{E[I\{r(\bO,\bb) \le \zeta\}\mid x,\bz,\bb,\eta_2] - E[I\{r(\bO,\bb) \le \zeta\}\mid \bz,\bb,\eta_1,\eta_2]}{E[\d\{\zeta - r(\bO,\bb)\}\mid\bb,\eta_1,\eta_2,\eta_3]},
\ese
$a_2(c,\bz)$ satisfies $E\{a_2(C,\bz)\mid \bz,\eta_2\} = 0$ and
\bse
&& E\left[\Delta \frac{E\{I(C \ge W) a_2(C,\bz)\mid W,\bz,\eta_2\}}{E\{I(C \ge W)\mid W,\bz,\eta_2\}}+ (1-\Delta)a_2
    (W,\bz) \mid c,\bz,\eta_1\right]\\
&=& - \frac{E[I\{r(\bO,\bb) \le \zeta\}\mid c,\bz,\bb,\eta_1] - E[I\{r(\bO,\bb) \le \zeta\}\mid \bz,\bb,\eta_1,\eta_2]}{E[\d\{\zeta - r(\bO,\bb)\}\mid\bb,\eta_1,\eta_2,\eta_3]},
\ese
and $a_3(\bz)$ and $\bfe$ are defined as
\bse
a_3(\bz) &=& -\frac{E[I\{r(\bO,\bb) \le \zeta\}\mid \bz,\bb,\eta_1,\eta_2] - (1-\alpha)}{E[\d\{\zeta - r(\bO,\bb)\}\mid\bb,\eta_1,\eta_2,\eta_3]},\\
\bfe &=& E\{\wt\S_\bb(Y,W,\Delta,\bZ,\bb)^{\otimes2}\mid\bb,\eta_1,\eta_2,\eta_3\}^{-1}
\left(\frac{E[\d\{\zeta - r(\bO,\bb)\}\partial r(\bO,\bb)/\partial \bb\mid\bb,\eta_1,\eta_2,\eta_3]}{E[\d\{\zeta - r(\bO,\bb)\}\mid\bb,\eta_1,\eta_2,\eta_3]}\right.\\
&&\left.-\frac{E[I\{r(\bO,\bb) \le \zeta\} \wt\S_\bb(Y,W,\Delta,\bZ,\bb)\mid\bb,\eta_1,\eta_2,\eta_3]}{E[\d\{\zeta - r(\bO,\bb)\}\mid\bb,\eta_1,\eta_2,\eta_3]}\right).
\ese
\end{Pro}
The proof of Proposition~\ref{pro:2} is in Section~\ref{sec:pro2}.

If $\bb$ and the nuisance models were known, $\zeta$ could be
estimated by solving $\sumi\phi\eff(\bO_i,\zeta,\bb)=0$. In practice
they are not, and a researcher specifies working models $\eta_1^*$ and
$\eta_2^\star$ for $\eta_1$ and $\eta_2$, respectively, with quantities computed
under these working models marked by the superscript $^{*\star}$, as
in $\phi\eff^{*\star}$. For $\bb$, we use the SPARCC estimator
$\wh\bb$, which solves $\sumi\wt\S_\bb^{*\star}(\bO_i,\bb)=\0$ under
the same working models. No working model is needed for $\eta_3$,
which enters $\phi\eff$ only through the density of the residual at
$\zeta$, a common factor that drops out of the estimating equation,
and through $\bfe$, whose term
$\bfe\trans\sumi\wt\S_\bb^{*\star}(\bO_i,\wh\bb)$ is zero by the
choice of $\wh\bb$. The estimator $\wh\zeta$ from the PRESCCO method
therefore solves
\be
\sumi \phi\eff^{*\star}(\bO_i,\zeta,\wh\bb)=0,\label{eq:prescco-ee}
\ee
and the prediction interval for the new patient is
$m(W_0,\Delta_0,\bZ_0,\wh\bb)\pm\wh\zeta$.

That the estimating equation \eqref{eq:prescco-ee} estimates 
$\zeta$ defined by $\pr\{r(\bO_0,\bb)\le\zeta\}=1-\alpha$ is not
apparent from its form. The case without right-censoring makes the
connection visible.
\begin{Rem}\label{rem:no-censoring}
When no patient's time to the event is right-censored, that is,
$\pr(X\le C\mid\eta_1,\eta_2)=1$, then $\Delta=1$ and $W=X$. In this
case $a_2=0$, so $\phi_2$ vanishes, $a_1$ requires no integral
equation, and $\phi_1+\phi_3$ reduces to the estimating equation
\be
\label{eq:no-cens-ee}
n^{-1}\sumi \pr\{|Y_i - m_0(X_i,\bZ_i,\wh\bb)|\le \zeta\mid
X_i,\bZ_i,\wh\bb\}=1-\alpha,
\ee
as derived in Section~\ref{sec:no-censoring}. If the outcome model has $Y=m_0(X,\bZ,\bb)+\epsilon$, with the
distribution of $\epsilon$ depending on neither $\bb$ nor $(X,\bZ)$,
then $\wh\zeta$ is the $(1-\alpha)$ quantile of $|\epsilon|$ under the
outcome model. When a patient's time to the event is right-censored,
$\phi_2$ does not vanish and $\phi_1+\phi_3$ does not reduce to
\eqref{eq:no-cens-ee}, since the probability in it conditions on a
time to the event that was not observed. These terms are what make the
connection hard to see in \eqref{eq:prescco-ee}.
\end{Rem}

%%%%%%%%%%%%%%%%%%%%%%%%%%%%%%%%%%%%%%%%%%%%%%%%%%%%%%%%%%%%%%%%%
\subsection{Double robustness, efficiency, and coverage}\label{sec:theory-length}
%%%%%%%%%%%%%%%%%%%%%%%%%%%%%%%%%%%%%%%%%%%%%%%%%%%%%%%%%%%%%%%%%

As for the conformal prediction methods, we evaluate the PRESCCO
method by how much its estimated half-length varies across
observational studies, its expected coverage rate, and its coverage
rate. All three criteria presume that $\wh\zeta$ is consistent for
$\zeta$, and the consistency of $\wh\zeta$ depends on the working
models $\eta_1^*$ and $\eta_2^\star$. Theorem~\ref{th:doubly-robust}
shows that $\wh\zeta$ is doubly robust: only one of the two needs to
be correctly specified.
\begin{Th}\label{th:doubly-robust}
Assume that either $\eta_1^*=\eta_1$ or $\eta_2^\star=\eta_2$. Under
Conditions~\ref{con:a1}--\ref{con:a4} in
Section~\ref{sec:doubly-robust}, the estimator $\wh\zeta$ is
consistent for $\zeta$.
\end{Th}
The proof of Theorem~\ref{th:doubly-robust} is in
Section~\ref{sec:doubly-robust}. Specifying working models might seem a disadvantage against the
conformal prediction methods, which estimate $\zeta$ without them.
With double robustness, a researcher needs only one working model to
be correct, and the observed data indicate which one can be specified
well. Under a high right-censoring rate, most patients' censoring
times are observed, and there are ample data to specify $\eta_2^\star$
well; under a low right-censoring rate, most patients' times to the
event are observed, and there are ample data to specify $\eta_1^*$
well. Whatever the right-censoring rate, the observational study
itself supports one of the two working models.

With one correctly specified working model, $\wh\zeta$ is also
asymptotically normal, and with both, it is semiparametrically
efficient, as Theorem~\ref{th:asympt-normal} shows.
\begin{Th}\label{th:asympt-normal}
Assume Conditions~\ref{con:a1}--\ref{con:a4} in
Section~\ref{sec:doubly-robust} and
Conditions~\ref{con:a5}--\ref{con:a7} in
Section~\ref{sec:asympt-normal}. If either $\eta_1^*=\eta_1$ or
$\eta_2^\star=\eta_2$, then
$n^{1/2}(\wh\zeta - \zeta) \stackrel{d}{\to}\Normal(0,\sigma^{2}\tau^{-2})$,
where $\sigma^2 \equiv \var\{\bh\trans\bxi(\bO,\bb)
+\phi\eff^{*\star}(\bO,\zeta,\bb)\}$, $\bxi$ is the influence function
of $\wh\bb$ in Condition~\ref{con:a7}, $\bh\equiv
E\{\partial\phi\eff^{*\star}(\bO,\zeta,\bb)/\partial\bb\}$,
and $\tau \equiv E\{\partial \phi\eff^{*\star}
(\bO,\zeta,\bb)/ \partial \zeta\}$.
If $\eta_1^*=\eta_1$ and $\eta_2^\star=\eta_2$, then $\wh\zeta$ is
semiparametrically efficient, that is, $n^{1/2}(\wh\zeta - \zeta)
\stackrel{d}{\to}\Normal[0,\var\{\phi\eff(\bO,\zeta,\bb)\}]$.
\end{Th}
Section~\ref{sec:asympt-normal} gives the proof of
Theorem~\ref{th:asympt-normal}. Efficiency means that $\wh\zeta$ has
the smallest asymptotic variance among regular asymptotically linear
estimators of $\zeta$ \citep{tsiatis2006semiparametric}, a class that
includes the sample quantile of residuals on which conformal
prediction is based. The conformal prediction methods use whether each
residual falls below a candidate value, whereas the PRESCCO method
uses the probability that it does, computed from the models. The
asymptotic variances of $\wh\zeta$ and of the sample quantile both
have the density of the residual at $\zeta$ in their denominator, the
source of the large variance in Theorem~\ref{th:scp}. A probability,
however, varies less across patients than whether a single residual
falls below $\zeta$, and $\wh\zeta$ therefore varies less across
observational studies than the estimated half-length from any
conformal prediction method.

The limiting distribution in Theorem~\ref{th:asympt-normal} also
gives a confidence interval for $\zeta$. The confidence interval shows
how precisely the boundary for being atypical is estimated, and
whether a patient whose outcome lies near that boundary would be
identified as atypical across the plausible values of $\zeta$. Full
conformal prediction and jackknife+, lacking a limiting distribution,
provide no such confidence interval.

Theorem~\ref{th:bias} next shows that the smaller variance of
$\wh\zeta$ costs no coverage, keeping the expected coverage rate of
the prediction interval from the PRESCCO method,
$\pr\{r(\bO_0,\wh\bb)\le\wh\zeta\}$, close to the prediction level.
\begin{Th}\label{th:bias}
Let $\bPhi(\bo,\zeta,\bb)\equiv\{\phi\eff^{*\star}(\bo,\zeta,\bb),
\wt\S_\bb^{*\star}(\bo,\bb)\trans\}\trans$, so that
$(\wh\zeta,\wh\bb\trans)\trans$ solves
$\sumi\bPhi(\bO_i,\zeta,\bb)=\0$, and assume that either
$\eta_1^*=\eta_1$ or $\eta_2^\star=\eta_2$. Under
Conditions~\ref{con:c1}--\ref{con:c5} in Section~\ref{sec:bias},
\bse
|\pr\{r(\bO_0,\wh\bb) \le \wh\zeta\} - (1-\alpha)|=O(n^{-1}).
\ese
\end{Th}
The proof of Theorem~\ref{th:bias} is in Section~\ref{sec:bias}. The
expected coverage rate is within order $n^{-1}$ of the prediction
level, the same order that split conformal prediction and full
conformal prediction guarantee at every $n$ and under any distribution
for the population. The remaining error of order $n^{-1}$ can be
reduced to order $n^{-3/2}$, without changing the limiting distribution
in Theorem~\ref{th:asympt-normal}, by subtracting an estimate of the
second-order term from $\wh\zeta$ \citep{tian2022methods}, a
calibration we omit.

Theorem~\ref{th:finite-bound} goes further, from the average across
observational studies to the one a researcher has, and shows that the
coverage rate of the prediction interval from the PRESCCO method,
$\pr\{r(\bO_0,\wh\bb)\le\wh\zeta\mid\wh\zeta,\wh\bb\}$, is close to the
prediction level with high probability.
\begin{Th}\label{th:finite-bound}
Let $\bPhi$ be as in Theorem~\ref{th:bias}, and assume that either
$\eta_1^*=\eta_1$ or $\eta_2^\star=\eta_2$. Under
Conditions~\ref{con:c1}--\ref{con:c5} in Section~\ref{sec:bias} and
Conditions~\ref{con:c6} and \ref{con:c7} in
Section~\ref{sec:finite-bound}, for any $\kappa>0$, there exist
constants $A>0$ and $k>0$ such that
\bse
|\pr\{r(\bO_0,\wh\bb) \le \wh\zeta\mid \wh\zeta,\wh\bb\} - (1-\alpha)|
\le 2M_1\kappa \lambda_1^{-1}n^{-1/2} (\log n)^{1/2}
+ 2M_2\kappa^2\lambda_1^{-2}n^{-1}\log n
\ese
with probability at least $1-An^{-k}$, where $M_1$, $M_2$, and
$\lambda_1$ are the constants in Conditions~\ref{con:c1},
\ref{con:c2}, and \ref{con:c3}, respectively.
\end{Th}
The proof of Theorem~\ref{th:finite-bound}, with the explicit forms of
$A$ and $k$, is in Section~\ref{sec:finite-bound}. The first term of
the bound dominates for large $n$, and the coverage rate in the one
observational study a researcher has is within order
$n^{-1/2}(\log n)^{1/2}$ of the prediction level. With
Theorem~\ref{th:bias}, the prediction interval from the PRESCCO method
thus holds the prediction level both on average across observational
studies and, with high probability, in the one a researcher has.

Among the conformal prediction methods, split conformal prediction has
a bound of this kind without further assumptions, but only because it
sets aside the training set, which enlarges the variance in
Theorem~\ref{th:scp}. For full conformal prediction and jackknife+,
which use every patient, a bound of this kind relies on algorithmic
stability \citep{bian2023training, liang2025algorithmic}. The bound
for the PRESCCO method, which also uses every patient, requires only
smoothness and tail conditions on its estimating equations, of the
kind standard for estimators that solve them.

%%%%%%%%%%%%%%%%%%%%%
\subsection{Computation and inference}\label{sec:computation-inference}
%%%%%%%%%%%%%%%%%%%%%

Solving $\sumi \phi\eff^{*\star}(\bO_i,\zeta,\wh\bb)=0$ directly is
computationally demanding, since $\phi_\bb^{*\star}$, $\phi_1^{*\star}$,
and $\phi_2^{*\star}$ each require an integral equation involving
conditional expectations under the working models. Two simplifications
ease that burden.

The first simplification comes from the choice of $\wh\bb$. The SPARCC
estimator solves $\sumi\wt\S_\bb^{*\star}(\bO_i,\bb)=\0$, and since
$\phi_\bb^{*\star}=\bfe^{*\star\trans}\wt\S_\bb^{*\star}$, the sum
$\sumi\phi_\bb^{*\star}(\bO_i,\zeta,\wh\bb)$ is zero at every $\zeta$.
The estimator $\wh\zeta$ therefore solves
$\sumi\{\phi_1^{*\star}(\bO_i,\zeta,\wh\bb)+\phi_2^{*\star}(\bO_i,\zeta,\wh\bb)
+\phi_3^{*\star}(\bO_i,\zeta,\wh\bb)\}=0$.

The second simplification is a reparametrization of the remaining
$\phi_1^{*\star}+\phi_2^{*\star}+\phi_3^{*\star}$. Rather than compute
the three separately, we use
$\phi_1^{*\star}+\phi_2^{*\star}+\phi_3^{*\star}
=(\phi_1^{*\star}+\phi_3^{*\star})+(\phi_2^{*\star}+\phi_3^{*\star})
-\phi_3^{*\star}$, where terms cancel within each of the two pairs,
leaving integral equations simpler than those for the individual
components.

With this reparametrization, define
\bse
b_1^{*\star}(x,\bz,\zeta,\bb)
&\equiv&
-\{a_1^{*\star}(x,\bz)+a_3^{*\star}(\bz)\}
E[\d\{\zeta-r(\bO,\bb)\}\mid \bb,\eta_1^*,\eta_2^\star],\\
b_2^{*\star}(c,\bz,\zeta,\bb)
&\equiv&
-\{a_2^{*\star}(c,\bz)+a_3^{*\star}(\bz)\}
E[\d\{\zeta-r(\bO,\bb)\}\mid \bb,\eta_1^*,\eta_2^\star],\\
b_3^{*\star}(\bz,\zeta,\bb)
&\equiv&
-a_3^{*\star}(\bz)
E[\d\{\zeta-r(\bO,\bb)\}\mid \bb,\eta_1^*,\eta_2^\star],
\ese
and the combined components become
\bse
\phi_1^{*\star}(\bo,\zeta,\wh\bb)+\phi_3^{*\star}(\bo,\zeta,\wh\bb)
&=&
-E[\d\{\zeta-r(\bO,\wh\bb)\}\mid\wh\bb,\eta_1^*,\eta_2^\star]^{-1}\\
&&\times
\left[
\delta b_1^{*\star}(w,\bz,\zeta,\wh\bb)
+
(1-\delta)
\frac{
E\{I(X>w)b_1^{*\star}(X,\bz,\zeta,\wh\bb)
\mid y,\bz,\wh\bb,\eta_1^*\}
}{
E\{I(X>w)\mid y,\bz,\wh\bb,\eta_1^*\}
}
\right],\\
\phi_2^{*\star}(\bo,\zeta,\wh\bb)+\phi_3^{*\star}(\bo,\zeta,\wh\bb)
&=&
-E[\d\{\zeta-r(\bO,\wh\bb)\}\mid\wh\bb,\eta_1^*,\eta_2^\star]^{-1}\\
&&\times
\left[
\delta
\frac{
E\{I(C\ge w)b_2^{*\star}(C,\bz,\zeta,\wh\bb)
\mid \bz,\eta_2^\star\}
}{
E\{I(C\ge w)\mid \bz,\eta_2^\star\}
}
+
(1-\delta)b_2^{*\star}(w,\bz,\zeta,\wh\bb)
\right],\\
\phi_3^{*\star}(\bo,\zeta,\wh\bb)
&=&
-E[\d\{\zeta-r(\bO,\wh\bb)\}\mid\wh\bb,\eta_1^*,\eta_2^\star]^{-1}
b_3^{*\star}(\bz,\zeta,\wh\bb).
\ese
The three share the nonzero factor
$-E[\d\{\zeta-r(\bO,\wh\bb)\}\mid\wh\bb,\eta_1^*,\eta_2^\star]^{-1}$,
which drops from the estimating equation, and $\wh\zeta$ therefore
solves $\sumi b^{*\star}(\bO_i,\zeta,\wh\bb)=0$, where
\bse
&& b^{*\star}(\bo,\zeta,\wh\bb)\\
&\equiv& \delta b_1^{*\star}(w,\bz,\zeta,\wh\bb) + (1-\delta) \frac{E\{I(X> w)
    b_1^{*\star} (X,\bz,\zeta,\wh\bb)\mid y,\bz,\wh\bb,\eta_1^*\}}{E\{I(X > w)\mid y,\bz,\wh\bb,\eta_1^*\}}\n\\
    &&+\delta \frac{E\{I(C \ge w)b_2^{*\star}
(C,\bz,\zeta,\wh\bb)\mid \bz,\eta_2^\star\}}{E\{I(C \ge w)\mid \bz,\eta_2^\star\}} + (1-\delta)b_2^{*\star} (w,\bz,\zeta,\wh\bb)-b_3^{*\star} (\bz,\zeta,\wh\bb).
\ese
The functions $b_1^{*\star}$ and $b_2^{*\star}$ solve integral
equations simpler than those for $a_1^{*\star}$ and $a_2^{*\star}$, and
$b_3^{*\star}$ has a closed form,
\be
 &&E\left[\Delta b_1^{*\star}(W,\bz,\zeta,\wh\bb) + (1-\Delta) \frac{E\{I(X> W)
    b_1^{*\star}(X,\bz,\zeta,\wh\bb)\mid W,Y,\bz,\wh\bb,\eta_1^*\}}{E\{I(X > W)\mid W,Y,\bz,\wh\bb,\eta_1^*\}}\mid x,\bz,\wh\bb,\eta_2^\star\right]\n\\
&&\qquad\qquad = E[I\{r(\bO,\wh\bb) \le \zeta\}\mid x,\bz,\wh\bb,\eta_2^\star] - (1-\alpha),\label{eq:b1}\\
 &&E\left[\Delta \frac{E\{I(C \ge W) b_2^{*\star}
    (C,\bz,\zeta,\wh\bb)\mid W,\bz,\eta_2^\star\}}{E\{I(C \ge W)\mid W,\bz,\eta_2^\star\}}+ (1-\Delta)b_2^{*\star}
    (W,\bz,\zeta,\wh\bb) \mid c,\bz,\eta_1^*\right]\n\\
&&\qquad\qquad = E[I\{r(\bO,\wh\bb) \le \zeta\}\mid c,\bz,\wh\bb,\eta_1^*] - (1-\alpha),\label{eq:b2}\\
&&b_3^{*\star}(\bz,\zeta,\wh\bb) = E[I\{r(\bO,\wh\bb) \le \zeta\}\mid \bz,\wh\bb,\eta_1^*,\eta_2^\star] - (1-\alpha),\label{eq:b3}
\ee
and Algorithm~\ref{alg:a1} computes $\wh\zeta$ from them.
\begin{algorithm}
\caption{The PRESCCO method}\label{alg:a1}
\begin{minipage}{\linewidth}
\raggedright
\hspace*{\algorithmicindent}\textbf{Input} Working models $\eta_1^*$ and $\eta_2^\star$, prediction level $1-\alpha$, residual function $r$, and observed data $\bO_i = (Y_i,W_i,\Delta_i,\bZ_i)$ for $i=1,\ldots,n$.\\
\hspace*{\algorithmicindent}\textbf{Assumption} $\eta_1^*=\eta_1$ or $\eta_2^\star=\eta_2$.\\
\hspace*{\algorithmicindent}\textbf{Output} The estimator $\wh\zeta$ and the prediction interval $m(W_0,\Delta_0,\bZ_0,\wh\bb)\pm\wh\zeta$.
\end{minipage}
\begin{algorithmic}[1]
\State Solve $\sumi\wt\S_\bb^{*\star}(\bO_i,\bb)=\0$ for $\wh\bb$.
\State Compute $r(\bO_i,\wh\bb)$ for $i=1,\ldots,n$.
\State Solve \eqref{eq:b1} for $b_1^{*\star}(x,\bz,\zeta,\wh\bb)$ on a grid for $(x,\bz,\zeta)$.
\State Solve \eqref{eq:b2} for $b_2^{*\star}(c,\bz,\zeta,\wh\bb)$ on a grid for $(c,\bz,\zeta)$.
\State Evaluate \eqref{eq:b3} for $b_3^{*\star}(\bz,\zeta,\wh\bb)$ on a grid for $(\bz,\zeta)$.
\State Solve $\sumi b^{*\star}(\bO_i,\zeta,\wh\bb) = 0$ for $\wh\zeta$.
\end{algorithmic}
\end{algorithm}

The confidence interval for $\zeta$ requires an estimate of the
asymptotic variance $\sigma^2\tau^{-2}$ in
Theorem~\ref{th:asympt-normal}, and Proposition~\ref{pro:sigmatau}
expresses that variance through the same $b^{*\star}$ that
Algorithm~\ref{alg:a1} computes, so the confidence interval requires
no further integral equations.
\begin{Pro}\label{pro:sigmatau}
    Assume either $\eta_1^*=\eta_1$ or $\eta_2^\star=\eta_2$. Let $\wh\bb$ be the solution of $\sumi\wt\S_\bb^{*\star}(\bO_i,\bb)=\0$.
    Under Conditions \ref{con:a1}--\ref{con:a7}, $\sigma^2\tau^{-2}$ equals
    \be\label{eq:sigma2tau2}
\frac{\var[-E\{\partial b^{*\star}(\bO,\zeta,\bb)/\partial\bb\trans\}E\{\partial\wt \S_\bb^{*\star}(\bO,\bb)/\partial\bb\trans\}^{-1} \wt\S_\bb^{*\star}(\bO,\bb)+b^{*\star}(\bO,\zeta,\bb)]}{E[\d\{\zeta - r(\bO,\bb)\}\mid \bb,\eta_1,\eta_2]^2}.
\ee
\end{Pro}
The proof of Proposition~\ref{pro:sigmatau} is in
Section~\ref{sec:sigmatau}. The denominator of \eqref{eq:sigma2tau2}
is the density of $r(\bO,\bb)$ at $\zeta$, estimated by
$n^{-1}\sumi K_h\{\wh\zeta - r(\bO_i,\wh\bb)\}$, where
$K_h(\cdot)=K(\cdot/h)/h$ with bandwidth $h>0$ and a kernel $K\ge 0$
satisfying $\int K(t)dt=1$. The numerator is estimated by the
empirical mean $\wh E$ and empirical variance $\wh\var$. Replacing
$\bb$ and $\zeta$ with $\wh\bb$ and $\wh\zeta$, respectively, gives
\bse
\wh{\sigma^2\tau^{-2}} = \frac{\wh\var[-\wh E\{\partial
  b^{*\star}(\bO,\wh\zeta,\wh\bb)/\partial\bb\trans\}\wh
  E\{\partial\wt \S_\bb^{*\star}(\bO,\wh\bb)/\partial\bb\trans\}^{-1}
  \wt\S_\bb^{*\star}(\bO,\wh\bb)+b^{*\star}(\bO,\wh\zeta,\wh\bb)]}{[n^{-1}\sumi
  K_h\{\wh\zeta - r(\bO_i,\wh\bb)\}]^2}.
\ese
Like $\wh\zeta$, this estimator needs neither $\bfe^{*\star}$ nor the
density of the residual under the working models. The interval
$\wh\zeta\pm z_{1-\nu/2}(\wh{\sigma^2\tau^{-2}}/n)^{1/2}$ is then an
asymptotic $(1-\nu)$ confidence interval for $\zeta$.

%%%%%%%%%%%%%%%%%%%%%%%%%%
\section{Simulation studies}\label{sec:sim}
%%%%%%%%%%%%%%%%%%%%%%%%%%

%%%%%%%%%%%%%%%%%%%%%%%%%%
\subsection{Simulation settings}
%%%%%%%%%%%%%%%%%%%%%%%%%%

We designed simulation studies to compare the PRESCCO method with the
three conformal prediction methods. Let $\TruncNormal(\mu,\sigma^2;a,b)$ denote
$\Normal(\mu,\sigma^2)$ truncated to $[a,b]$. We generated the
time-to-event covariate $X$ from $\eta_1=\TruncNormal(0,1^2;-1,1)$, the
outcome $Y$ from $Y\mid X\sim\Normal(\beta_1+\beta_2X,4^2)$ with
$\bb=(\beta_1,\beta_2)\trans=(0,3)\trans$, and the censoring time $C$
from $\eta_2=\TruncNormal(\gamma,1^2;-1,1)$ independently of $(X,Y)$.
Taking $\gamma=2,1,0,-1,-2$ gave low (20--30\%), low-to-moderate
(30--40\%), moderate (45--55\%), moderate-to-high (60--70\%), and high
(70--80\%) censoring, respectively. We did not generate fully observed
covariates $\bZ$, so that differences among the methods come from
right-censoring alone.

To evaluate the PRESCCO method under misspecified working models, we
took \\$\eta_1^*=\TruncNormal(-2,1^2;-1,1)$ and
$\eta_2^\star=\TruncNormal(\gamma^\star,1^2;-1,1)$, with
$\gamma^\star=0$ when $\gamma\neq0$ and $\gamma^\star=2$ when
$\gamma=0$. With $m_0(X_0,\bb)=\beta_1+\beta_2X_0$, the three centers
of Section~\ref{sec:setup} give the residual functions
$r_1(\bO_0,\bb)=|Y_0-m_1(W_0,\Delta_0,\bb)|$,
$r_2(\bO_0,\bb)=|Y_0-m_2(W_0,\Delta_0,\bb)|$, and
$r_1^*(\bO_0,\bb)=|Y_0-m_1^*(W_0,\Delta_0,\bb)|$, where $m_1^*$ is
computed under $\eta_1^*$.

For each censoring rate and each residual function, we simulated
$M=1{,}000$ observational studies, each following $n=1{,}000$ patients,
computed the estimated half-length from each study, and approximated the
coverage rate over $N=10{,}000$ independent new patients. 
For full
conformal prediction and jackknife+, the estimated half-length depends
on the new patient, so we averaged the estimated half-length over those
patients, as well. We ran these simulations for the following six
methods, three specifications of the PRESCCO method and the three
conformal prediction methods.
\begin{enumerate}[label=(\arabic*),ref=(\arabic*),start=1]
    \item The PRESCCO method under $(\eta_1,\eta_2)$, with estimated
    half-length $\wh\zeta$ and coverage rate
    $\pr\{r(\bO_0,\wh\bb) \le \wh\zeta\mid\bO_1,\ldots,\bO_n\}$.
    \item The PRESCCO method under $(\eta_1^*,\eta_2)$, with estimated
    half-length $\wh\zeta$ and coverage rate
    $\pr\{r(\bO_0,\wh\bb) \le \wh\zeta\mid\bO_1,\ldots,\bO_n\}$.
    \item The PRESCCO method under $(\eta_1,\eta_2^\star)$, with
    estimated half-length $\wh\zeta$ and coverage rate
    $\pr\{r(\bO_0,\wh\bb) \le \wh\zeta\mid\bO_1,\ldots,\bO_n\}$.
    \item Split conformal prediction, with $n_1=n/2=500$, estimated
    half-length $\wh\zeta\SCP$, and coverage rate
    $\pr\{r(\bO_0,\wh\bb) \le \wh\zeta\SCP\mid\bO_1,\ldots,\bO_n\}$.
  
    \item Full conformal prediction, with estimated half-length
    $E\{\wh\zeta\FCP(\bO_0)\mid\bO_1,\ldots,\bO_n\}$ and coverage rate
    $\pr[r\{\bO_0,\wh\bb(\bO_0)\} \le
    \wh\zeta\FCP(\bO_0)\mid\bO_1,\ldots,\bO_n]$.
    \item Jackknife+, with estimated half-length
    $E\{\wh\zeta\JK(\bO_0)\mid\bO_1,\ldots,\bO_n\}$ and coverage rate
    $\pr\{Y_0\in [l\JK(\bO_0),u\JK(\bO_0)]\mid\bO_1,\ldots,\bO_n\}$.
\end{enumerate}

All methods used the outcome model
$Y\mid X\sim\Normal(\beta_1+\beta_2X,4^2)$ and differ in how they
estimate $\bb$. The PRESCCO method and split conformal prediction use
SPARCC, which Algorithm~\ref{alg:a1} requires for the PRESCCO method
and which split conformal prediction can afford, since it fits $\bb$
once. Full conformal prediction refits $\bb$ at every candidate value
of $y$ and jackknife+ refits it $n$ times. Both therefore use the
complete-case estimator, which uses only the data from patients with
$\Delta=1$ and remains consistent under the noninformative censoring
mechanism.
%The choice does not drive the comparison that follows, since by Theorem~\ref{th:scp}, the variance of an estimated half-length built from residuals is governed by the density of the residual at an extreme quantile rather than by which estimator of $\bb$ is used.

%%%%%%%%%%%%%%%%%%%%%%%%%%%%%%%%%%%
\subsection{Simulation results}
%%%%%%%%%%%%%%%%%%%%%%%%%%%%%%%%%%%

\begin{figure}[!htbp]
    \centering
\includegraphics[width=\linewidth,height=0.38\textheight,keepaspectratio]{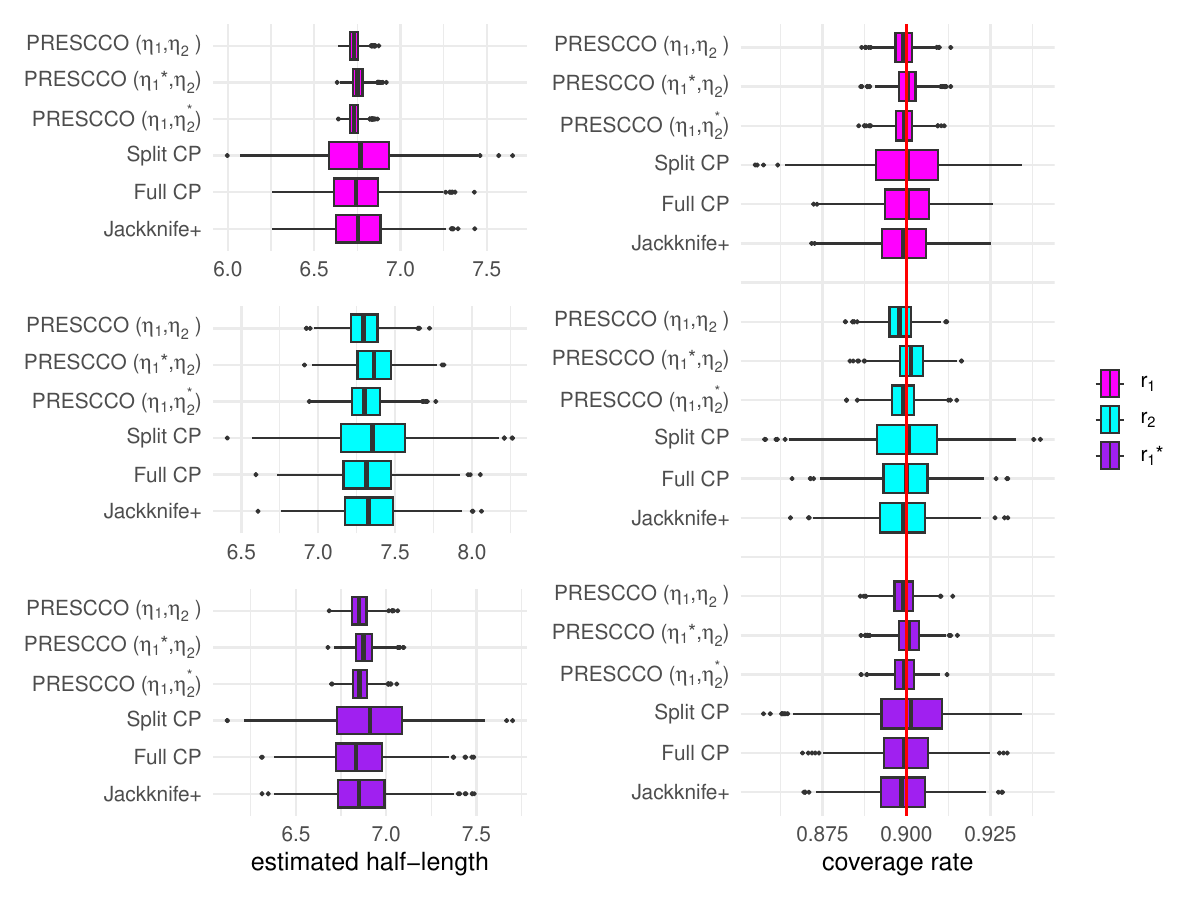}
{\spacingset{1}\caption{Boxplots of the estimated half-length and the
coverage rate under moderate-to-high censoring (60--70\%) across 1,000
simulated observational studies. PRESCCO $(\eta_1,\eta_2)$, PRESCCO
method with both working models correctly specified; $\eta_1^*$,
misspecified time-to-event model; $\eta_2^\star$, misspecified
censoring model; CP, conformal prediction.}\label{fig:f4}}
\end{figure}

Figure~\ref{fig:f4} and Table~\ref{tab:t4} report the results under
moderate-to-high censoring. The results under the remaining censoring
rates, in Section~\ref{sec:addsim}, lead to the same conclusions. Every
method has an expected coverage rate between $0.898$ and $0.901$
against a prediction level of $0.9$, and within each residual function,
the mean estimated half-lengths differ by less than $1\%$. Both results
hold for the PRESCCO method even with one working model misspecified,
which reflects its double robustness (Theorem~\ref{th:doubly-robust}).

The methods differ instead in how much the estimated half-length
varies, as expected from Theorems~\ref{th:scp}
and~\ref{th:asympt-normal}. For the residual function $r_1$, the
standard deviation of $\wh\zeta$ from the PRESCCO method under
$(\eta_1,\eta_2)$ is $0.037$, or $0.5\%$ of its mean. For the three
conformal prediction methods, it is $0.185$ to $0.262$, or $2.7$ to
$3.9\%$, and $r_2$ and $r_1^*$ give the same ordering
(Table~\ref{tab:t4}). Two researchers, each with their own observational study of
patients from the same population, would report nearly the same
estimated half-length under the PRESCCO method. Under any conformal
prediction method, one researcher's estimated half-length can exceed
the other's by enough that an outcome inside one prediction interval
falls outside the other. In Figure~\ref{fig:f4}, the boxplot
for the PRESCCO method lies inside the boxplot of every conformal
prediction method, and the larger standard deviations of the conformal
prediction methods therefore do not come from a few extreme estimates.

\begin{table}[!t]
{\spacingset{1}
\centering
\caption{Mean and standard deviation of the estimated half-length and
  the coverage rate under moderate-to-high censoring (60--70\%) across
  1,000 simulated observational studies. PRESCCO $(\eta_1,\eta_2)$,
  PRESCCO method with both working models correctly specified;
  $\eta_1^*$, misspecified time-to-event model; $\eta_2^\star$,
  misspecified censoring model; CP, conformal prediction; sd, standard
  deviation.}\label{tab:t4}
\scalebox{0.8}{%
\begin{tabular}{cc|cccc}
 & method
 & \multicolumn{2}{c}{estimated half-length}
 & \multicolumn{2}{c}{coverage rate}\\
 & & mean & sd & mean & sd\\
\hline
\multirow{6}{*}{$r_1$}
  & PRESCCO $(\eta_1,\eta_2)$       & 6.733 & 0.037 & 0.899 & 0.004 \\
  & PRESCCO $(\eta_1^*,\eta_2)$     & 6.754 & 0.045 & 0.900 & 0.004 \\
  & PRESCCO $(\eta_1,\eta_2^\star)$ & 6.734 & 0.036 & 0.899 & 0.004 \\
  & Split CP                        & 6.765 & 0.262 & 0.900 & 0.014 \\
  & Full CP                         & 6.743 & 0.185 & 0.900 & 0.010 \\
  & Jackknife+                      & 6.756 & 0.186 & 0.899 & 0.010 \\
\hline
\multirow{6}{*}{$r_2$}
  & PRESCCO $(\eta_1,\eta_2)$       & 7.302 & 0.128 & 0.898 & 0.005 \\
  & PRESCCO $(\eta_1^*,\eta_2)$     & 7.370 & 0.156 & 0.901 & 0.005 \\
  & PRESCCO $(\eta_1,\eta_2^\star)$ & 7.316 & 0.135 & 0.899 & 0.005 \\
  & Split CP                        & 7.362 & 0.304 & 0.900 & 0.013 \\
  & Full CP                         & 7.322 & 0.225 & 0.900 & 0.010 \\
  & Jackknife+                      & 7.334 & 0.226 & 0.899 & 0.010 \\
\hline
\multirow{6}{*}{$r_1^*$}
  & PRESCCO $(\eta_1,\eta_2)$       & 6.852 & 0.060 & 0.899 & 0.004 \\
  & PRESCCO $(\eta_1^*,\eta_2)$     & 6.879 & 0.069 & 0.901 & 0.004 \\
  & PRESCCO $(\eta_1,\eta_2^\star)$ & 6.856 & 0.058 & 0.899 & 0.004 \\
  & Split CP                        & 6.906 & 0.262 & 0.901 & 0.013 \\
  & Full CP                         & 6.846 & 0.189 & 0.900 & 0.010 \\
  & Jackknife+                      & 6.859 & 0.190 & 0.899 & 0.010 \\
\end{tabular}}
\par}
\end{table}

The less variable estimated half-length from the PRESCCO method also
gives a less variable coverage rate. For the residual function $r_1$,
the standard deviation of the coverage rate is $0.004$ for the PRESCCO
method under $(\eta_1,\eta_2)$ and $0.010$ to $0.014$ for the three
conformal prediction methods, and $r_2$ and $r_1^*$ give similar values
(Table~\ref{tab:t4}). Within two standard deviations of its mean, the
coverage rate ranges from $0.891$ to $0.907$ under the PRESCCO method
and from $0.872$ to $0.928$ under split conformal prediction. The
prediction interval from the PRESCCO method therefore contains close to
$90\%$ of new patients' outcomes in the one observational study a
researcher has, in line with Theorem~\ref{th:finite-bound}.

The smaller variability of $\wh\zeta$ and of the coverage rate under
the PRESCCO method persists with one working model misspecified. Under
every pair of working models and every residual function, the standard
deviation of $\wh\zeta$ stays below that of the estimated half-length
from any conformal prediction method, and the standard deviation of the
coverage rate stays at $0.004$ to $0.005$, against $0.010$ to $0.014$
for the conformal prediction methods (Table~\ref{tab:t4}). For $r_1$,
misspecifying $\eta_1$ raises the standard deviation of $\wh\zeta$ from
$0.037$ to $0.045$, still a quarter of the smallest among the conformal
prediction methods.

The reduction in variability from the PRESCCO method shrinks as the
censoring rate rises. From 20--30\% to 70--80\% censoring, the
standard deviation of $\wh\zeta$ from the PRESCCO method under
$(\eta_1,\eta_2)$ for $r_1$ grows ninefold, from $0.006$ to $0.054$,
while  the smallest
among the three conformal prediction methods grows only from $0.177$ to
$0.208$ (Tables~\ref{tab:t1} and~\ref{tab:t5}). The standard deviation
of $\wh\zeta$ is thus thirty times smaller than the smallest among the
conformal prediction methods under 20--30\% censoring and four times
smaller under 70--80\%. The PRESCCO method draws its precision from the probability
computed from the models, and for a patient whose time to the event is
right-censored, that probability must average over every time beyond
$W$. The PRESCCO method therefore loses precision as the censoring
rate rises, unlike the conformal prediction methods, which use one
residual from every patient whether or not the time to the event was
observed.

Right-censoring also affects the center. The center $m_2$ of the
residual function $r_2$ ignores it, treating $W_0$ as though it were
$X_0$ and placing a patient whose time to the event is right-censored
closer to the event than they are.  The target
half-length, and with it the estimated half-length from every method,
is then longer for $r_2$ than for $r_1$, by $1.8\%$ under 20--30\% censoring and $11.7\%$ under
70--80\% (Tables~\ref{tab:t1} and~\ref{tab:t5}), since the prediction
interval must also cover the outcomes of these misplaced patients. The
reduction in variability from the PRESCCO method is also smaller for
$r_2$. Under 20--30\% censoring, the standard deviation of $\wh\zeta$
is nine times smaller than the smallest among the conformal prediction
methods for $r_2$, compared with thirty times for $r_1$. Under
70--80\% censoring, the standard deviation of $\wh\zeta$ is $1.3$ times
smaller for $r_2$, compared with four times for $r_1$. Without the
time-to-event model, the center carries less of the information that
the PRESCCO method uses to reduce variability.

%%%%%%%%%%%%%%%%%%%%%%%%%%%%%%%%%%%%%%%%%%%%%%%%%%%%%%%%%%%%%%%%%%%%%
\section{Identifying atypical patients reliably in Enroll-HD}\label{sec:application}
%%%%%%%%%%%%%%%%%%%%%%%%%%%%%%%%%%%%%%%%%%%%%%%%%%%%%%%%%%%%%%%%%%%%%

In Huntington disease, a mutation of at least 40
cytosine-adenine-guanine (CAG) repeats in the huntingtin gene
guarantees that the patient develops the disease, and a genetic test
identifies patients who carry the mutation years, often decades,
before the disease is clinically detectable. Those years are when a
therapy could delay the disease, and patients who are atypical during
them are the ones whose biology could suggest how. The Huntington
Disease Integrated Staging System marks the point at which the disease
becomes clinically detectable as Stage 2,
and cognitive and motor test scores recorded before Stage 2 worsen as
patients approach Stage 2 \citep{tabrizi2022biological}. We therefore
check how reliably the PRESCCO method and the three conformal
prediction methods identify patients who are atypical before Stage 2
in Enroll-HD, an international observational study of Huntington
disease.

We analyzed the $n=2{,}809$ patients in Enroll-HD who carry the
mutation, are at least 18 years of age, were before Stage 2 at study
entry, and returned for at least one visit after entry. The outcome
$Y$ is a score recorded at study entry, and we analyze four scores
separately, three cognitive test scores (Stroop Color, Stroop Word,
and Stroop Interference) and the composite Unified Huntington Disease
Rating Scale (cUHDRS) score, which combines motor, cognitive, and
daily function. The time-to-event covariate
$X$ is the time from study entry to Stage 2, and the censoring time
$C$ is the time from study entry to study exit, both in years. For
$77.2\%$ of patients, follow-up ends before Stage 2, and their time to
Stage 2 is right-censored. The fully observed covariates
$\bZ=(Z_1,Z_2,Z_3)$, also recorded at study entry, are the normalized
CAG-age-product (CAP) score $Z_1$, which combines the number of CAG
repeats with age to measure cumulative exposure to the mutation, sex $Z_2$, and an International Standard
Classification of Education (ISCED) level of at least 4, $Z_3$. Each
affects test scores, and
including them keeps a patient from being identified as atypical
because of these factors alone.

We hold out $25\%$ of the patients, $N=702$, as new patients and
compute the prediction intervals from the data of the other $75\%$, so
that no held-out patient's data take part in computing the prediction
interval their score is compared with. Split conformal prediction
further divides these data into a training set and a calibration set.
Every method uses the prediction level $1-\alpha=0.9$, the residual
functions $r_1$, $r_2$, and $r_1^*$ from the simulation studies, and
the outcome model
$Y\mid X,\bZ\sim\Normal\{(1,X,\bZ\trans,X\bZ\trans)\bb,\sigma^2\}$,
whose interactions let the scores decline at a rate that depends on
the covariates. For the PRESCCO method, we fit
$X\mid\bZ\sim\TruncNormal\{(1,\bZ\trans)\balpha_1,\tau_1^2;0,8\}$ and
$C\mid\bZ\sim\TruncNormal\{(1,\bZ\trans)\balpha_2,\tau_2^2;0,8\}$ to
the same data by maximum likelihood. The upper bound of 8 years is the
longest follow-up. With slight abuse of notation, we write $\eta_1$ and
$\eta_2$ for the fitted models. Neither fitted model can be verified
on real data, although with $77.2\%$ of times to Stage 2 right-censored, most
censoring times are observed, and there are ample data to specify the
censoring model well. We also take
$\eta_1^*=\eta_2^\star=\TruncNormal(4,2^2;0,8)$, which ignore the
covariates. The PRESCCO method is applied under $(\eta_1,\eta_2)$,
$(\eta_1^*,\eta_2)$, and $(\eta_1,\eta_2^\star)$, each of which keeps
one fitted working model, and if the fitted models are adequate, the
three estimated half-lengths should be close.

We evaluate the resulting prediction intervals by their coverage rate
alone, since Enroll-HD is a single observational study. For each
held-out patient, we check whether their score $Y_0$ falls inside
their prediction interval, and the fraction $\wh p$ of the $N$
held-out scores that fall inside estimates the coverage rate. Given
the prediction intervals and the independence of patients' data, the
number of held-out scores that fall inside their prediction intervals
is binomial, with $N$ trials and success probability equal to the
coverage rate. The band
$\wh p\pm z_{1-\nu/2}\{\wh p(1-\wh p)/N\}^{1/2}$ is then an
approximate $(1-\nu)$ confidence interval for the coverage rate. The
coverage rate is close to the prediction level of $0.9$ when the band
contains $0.9$, and deviates from the prediction level otherwise. We
report the band at $1-\nu=0.75$ rather than $0.95$, since a narrower
band is a stricter test of whether the coverage rate is close to the
prediction level, and give the bands at other values of $1-\nu$ in
Tables~\ref{tab:t6r1}--\ref{tab:t6r1star} of
Section~\ref{sec:add-real-data}.

\begin{table}[!ht]
{\spacingset{1}
\centering
\caption{Estimated half-length of the prediction interval, coverage
rate, and band at $1-\nu=0.75$ for four scores recorded at study
entry, based on the Enroll-HD dataset. SC, Stroop Color; SW, Stroop
Word; SI, Stroop Interference; cUHDRS, composite Unified Huntington
Disease Rating Scale. PRESCCO $(\eta_1,\eta_2)$, PRESCCO method with
both working models fitted to the data; $\eta_1^*$, misspecified
time-to-event model; $\eta_2^\star$, misspecified censoring model; CP,
conformal prediction. Coverage rates whose band does not contain
$0.9$ are marked in red.}\label{tab:t6_rowwise}
\setlength{\tabcolsep}{3pt}
\scalebox{0.8}{%
\begin{tabular}{cc|ccc|ccc}
  &  & \multicolumn{3}{c|}{SC} & \multicolumn{3}{c}{SW} \\
\cline{3-8}
$r$ & method
& \shortstack{estimated\\half-\\length} & \shortstack{coverage\\rate} & band
& \shortstack{estimated\\half-\\length} & \shortstack{coverage\\rate} & band\\
\hline
\multirow{6}{*}{$r_1$}
& PRESCCO $(\eta_1,\eta_2)$
& 21.904 & 0.902 & [0.889,0.915]
& 27.056 & 0.899 & [0.886,0.912] \\
& PRESCCO $(\eta_1^*,\eta_2)$
& 21.894 & 0.899 & [0.886,0.912]
& 27.085 & 0.902 & [0.889,0.915] \\
& PRESCCO $(\eta_1,\eta_2^\star)$
& 21.604 & 0.902 & [0.889,0.915]
& 26.813 & 0.900 & [0.887,0.913] \\
& Split CP
& 21.299 & 0.891 & [0.878,0.905]
& 25.539 & {\red 0.880} & {\red [0.866,0.894]} \\
& Full CP
& 23.635 & 0.889 & [0.875,0.902]
& 27.770 & {\red 0.886} & {\red [0.872,0.900]} \\
& Jackknife+
& 21.138 & {\red 0.848} & {\red [0.832,0.864]}
& 26.915 & {\red 0.873} & {\red [0.858,0.887]} \\
\hline
\multirow{6}{*}{$r_2$}
& PRESCCO $(\eta_1,\eta_2)$
& 22.500 & 0.903 & [0.890,0.916]
& 27.840 & 0.900 & [0.887,0.913] \\
& PRESCCO $(\eta_1^*,\eta_2)$
& 22.479 & 0.900 & [0.887,0.913]
& 27.884 & 0.896 & [0.883,0.909] \\
& PRESCCO $(\eta_1,\eta_2^\star)$
& 21.847 & 0.899 & [0.886,0.912]
& 27.111 & 0.897 & [0.884,0.910] \\
& Split CP
& 21.939 & {\red 0.883} & {\red [0.869,0.897]}
& 26.356 & {\red 0.880} & {\red [0.866,0.894]} \\
& Full CP
& 24.407 & 0.890 & [0.876,0.904]
& 28.896 & 0.890 & [0.876,0.904] \\
& Jackknife+
& 21.271 & {\red 0.839} & {\red [0.823,0.855]}
& 27.579 & {\red 0.877} & {\red [0.863,0.891]} \\
\hline
\multirow{6}{*}{$r_1^*$}
& PRESCCO $(\eta_1,\eta_2)$
& 21.727 & 0.903 & [0.890,0.916]
& 26.881 & 0.899 & [0.886,0.912] \\
& PRESCCO $(\eta_1^*,\eta_2)$
& 21.756 & 0.903 & [0.890,0.916]
& 26.852 & 0.900 & [0.887,0.913] \\
& PRESCCO $(\eta_1,\eta_2^\star)$
& 21.606 & 0.903 & [0.890,0.916]
& 26.784 & 0.899 & [0.886,0.912] \\
& Split CP
& 21.109 & 0.891 & [0.878,0.905]
& 25.451 & {\red 0.878} & {\red [0.864,0.893]} \\
& Full CP
& 23.270 & 0.887 & [0.873,0.901]
& 27.124 & {\red 0.886} & {\red [0.872,0.900]} \\
& Jackknife+
& 20.875 & {\red 0.854} & {\red [0.838,0.869]}
& 27.038 & {\red 0.883} & {\red [0.869,0.897]} \\
\hline
\end{tabular}}

\vspace{1em}

\scalebox{0.8}{%
\begin{tabular}{cc|ccc|ccc}
  &  & \multicolumn{3}{c|}{SI} & \multicolumn{3}{c}{cUHDRS} \\
\cline{3-8}
$r$ & method
& \shortstack{estimated\\half-\\length} & \shortstack{coverage\\rate} & band
& \shortstack{estimated\\half-\\length} & \shortstack{coverage\\rate} & band \\
\hline
\multirow{6}{*}{$r_1$}
& PRESCCO $(\eta_1,\eta_2)$
& 16.684 & 0.901 & [0.888,0.914]
& 2.333  & 0.896 & [0.883,0.909] \\
& PRESCCO $(\eta_1^*,\eta_2)$
& 16.657 & 0.901 & [0.888,0.914]
& 2.342  & 0.896 & [0.883,0.909] \\
& PRESCCO $(\eta_1,\eta_2^\star)$
& 16.380 & 0.902 & [0.889,0.915]
& 2.266  & 0.895 & [0.881,0.908] \\
& Split CP
& 16.020 & 0.890 & [0.876,0.904]
& 2.307  & 0.890 & [0.877,0.904] \\
& Full CP
& 16.327 & 0.887 & [0.873,0.901]
& 3.128  & {\red 0.917} & {\red [0.905,0.929]} \\
& Jackknife+
& 16.307 & {\red 0.882} & {\red [0.868,0.896]}
& 2.303  & {\red 0.755} & {\red [0.737,0.774]} \\
\hline
\multirow{6}{*}{$r_2$}
& PRESCCO $(\eta_1,\eta_2)$
& 17.274 & 0.905 & [0.893,0.918]
& 2.476  & 0.909 & [0.897,0.922] \\
& PRESCCO $(\eta_1^*,\eta_2)$
& 17.267 & 0.905 & [0.893,0.918]
& 2.487  & 0.909 & [0.897,0.922] \\
& PRESCCO $(\eta_1,\eta_2^\star)$
& 16.680 & 0.899 & [0.886,0.912]
& 2.312  & 0.889 & [0.875,0.902] \\
& Split CP
& 16.317 & 0.893 & [0.880,0.907]
& 2.386  & 0.893 & [0.880,0.907] \\
& Full CP
& 17.486 & 0.901 & [0.888,0.914]
& 3.191  & {\red 0.915} & {\red [0.903,0.927]} \\
& Jackknife+
& 16.324 & {\red 0.869} & {\red [0.854,0.883]}
& 2.376  & {\red 0.748} & {\red [0.729,0.767]} \\
\hline
\multirow{6}{*}{$r_1^*$}
& PRESCCO $(\eta_1,\eta_2)$
& 16.501 & 0.904 & [0.891,0.917]
& 2.291  & 0.895 & [0.881,0.908] \\
& PRESCCO $(\eta_1^*,\eta_2)$
& 16.493 & 0.902 & [0.889,0.915]
& 2.296  & 0.892 & [0.878,0.905] \\
& PRESCCO $(\eta_1,\eta_2^\star)$
& 16.367 & 0.902 & [0.889,0.915]
& 2.259  & 0.895 & [0.881,0.908] \\
& Split CP
& 16.055 & 0.898 & [0.885,0.911]
& 2.285  & 0.892 & [0.878,0.905] \\
& Full CP
& 16.124 & {\red 0.881} & {\red [0.867,0.895]}
& 3.097  & {\red 0.917} & {\red [0.905,0.929]} \\
& Jackknife+
& 16.368 & {\red 0.881} & {\red [0.867,0.895]}
& 2.270  & {\red 0.761} & {\red [0.743,0.780]} \\
\hline
\end{tabular}}
\par}
\end{table}

The PRESCCO method has a coverage rate close to the prediction level
in all thirty-six of its sets of prediction intervals, one for each
combination of the three pairs of working models, three residual
functions, and four scores (Table~\ref{tab:t6_rowwise}). Its coverage
rates lie between $0.889$ and $0.909$, so it identifies $9.1$ to
$11.1\%$ of the held-out patients as atypical, against the $10\%$ set
by the prediction level. These coverage rates rely on the outcome
model as well as the working models, and their closeness to the
prediction level in every set suggests that the outcome model is
adequate. The pair of working models changes the estimated half-length
little. For the Stroop Color score under $r_1$, the three pairs give
$21.904$, $21.894$, and $21.604$, and in eleven of the twelve
combinations of score and residual function, the three estimated
half-lengths differ by less than $4\%$. This agreement is what we
expected if the fitted working models are adequate, and the agreement
also means that the patients the PRESCCO method identifies as atypical
change little with the working models we used.

The conformal prediction methods, in contrast, identify too many
patients as atypical or too few. Jackknife+ has a coverage rate below
the prediction level in all twelve of its sets of prediction
intervals. On the cUHDRS score, its coverage rate is $0.748$ to
$0.761$, below even the $1-2\alpha=0.8$ that jackknife+ guarantees on
average across observational studies, and it identifies about one
held-out patient in four as atypical rather than one in ten. Full
conformal prediction has a coverage rate below the prediction level in
three of its sets of prediction intervals and above the prediction
level on the cUHDRS score, at $0.915$ to $0.917$. On the cUHDRS score,
the estimated half-length from full conformal prediction is about a
third longer than that from the PRESCCO method under
$(\eta_1,\eta_2)$. Split conformal prediction comes closest, with a
coverage rate close to the prediction level in eight of its twelve
sets of prediction intervals and slightly below in the other four, at
$0.878$ to $0.883$.

Identifying too many patients as atypical, as jackknife+ does,
includes patients whose scores are common at their point in the
disease course, and studying them would reveal nothing about the
disease. Identifying too few, as full conformal prediction does on the
cUHDRS score, misses patients whose biology could suggest a therapy.
In practice, a researcher computes the prediction interval from all
the patients an observational study followed, with no held-out scores
to reveal either error, and must rely on the method to avoid both. The
PRESCCO method can be relied on in this way. Its coverage rate in the
one study a researcher has is close to the prediction level with high
probability (Theorem~\ref{th:finite-bound}), and in Enroll-HD, at a
right-censoring rate of $77.2\%$, its coverage rate was close to the
prediction level in every set of prediction intervals.

%%%%%%%%%%%%%%%%%%%%%%%%%%%%%%%%%%%%%%%%%%%%%%%
\section{Discussion}\label{sec:discussion}
%%%%%%%%%%%%%%%%%%%%%%%%%%%%%%%%%%%%%%%%%%%%%%%

Conformal prediction is the standard framework for prediction
intervals, and adapted to a right-censored covariate, it keeps the
expected coverage rate close to the prediction level under any
distribution for the population. Its estimated half-length, however,
varies across observational studies, and in the one study a researcher
has, it can identify too many or too few patients as atypical.
Recasting the target half-length as a functional of the models that
describe the population, rather than as a sample quantile of residuals,
brings semiparametric theory to prediction intervals. The PRESCCO
method keeps both the expected coverage rate and the coverage rate
close to the prediction level, and adds double robustness,
semiparametric efficiency, and a confidence interval for the target
half-length.

Double robustness applies to the time-to-event and censoring models
but not to the outcome model. The PRESCCO method and conformal
prediction both compute their center from the outcome model, and a
misspecified outcome model can identify a patient as atypical because
the center is wrong, not because their score is uncommon at their
point in the disease course. The expected coverage rate and the
coverage rate of the PRESCCO method also rely on the outcome model,
whereas the expected coverage rate of conformal prediction does not,
yet in Enroll-HD the coverage rate of the PRESCCO method stayed close
to the prediction level in every set of prediction intervals. Extending
double robustness to the outcome model is a direction for future work.

The PRESCCO method centers each prediction interval at the new
patient's point in the disease course but uses the same target
half-length at every point. A conditional target half-length
$\zeta(w,\delta,\bz)$, defined by
$\pr\{r(\bO_0,\bb)\le\zeta(W_0,\Delta_0,\bZ_0)\mid
W_0,\Delta_0,\bZ_0\}=1-\alpha$, would let the half-length change with
$(W_0,\Delta_0,\bZ_0)$ and identify a proportion $\alpha$ of new
patients as atypical at every value. The supremum of
$\zeta(w,\delta,\bz)$, its largest value, could then  identify only patients
whose outcomes lie outside the range common at any point in the
disease course. These patients would be fewer, but they would be the
strongest candidates for biology that slows or speeds the disease, and
the first to study when following up each patient is costly.
Estimating $\zeta(w,\delta,\bz)$, and through it the supremum, is a
quantile regression problem
\citep{yu1998local} under a right-censored
covariate, and the observed-data likelihood and working models of the
PRESCCO method give a starting point for that problem.

Even without these extensions, the PRESCCO method needs neither a
finished observational study nor patients who have had the event.  In a study still under way, the times
to the event of patients who have not yet had the event are
right-censored at their follow-up so far, and the PRESCCO method can be
applied to the data collected to date. A patient whose score is
atypical can then be identified years before the event and followed
from that point on, while there is still time to learn what sets them
apart. Atypical patients are few in any one study, but the PRESCCO
method gives nearly the same estimated half-length in two studies of
the same population, and atypical patients identified in several
studies can be brought together to look for what they share. The same
right-censored time to the event arises in Alzheimer disease,
Parkinson disease, and amyotrophic lateral sclerosis, wherever a study
follows patients toward an event that most have yet to reach. In each,
the PRESCCO method lets researchers find, in the one study they have,
the atypical patients whose biology could point to new therapies.

\bigskip
\noindent\textbf{Disclosure statement.} The authors have no conflicts
of interest to declare.

\medskip
\noindent\textbf{Data availability statement.} Data used in this work
were provided by the CHDI Foundation, Inc. The Enroll-HD study would
not be possible without the vital contributions of the research
participants and their families. These data are available upon request
from the CHDI Foundation.

\medskip
\noindent\textbf{Supplementary material.} Supplementary materials
include derivations for all theoretical claims and additional numerical
results.
%The \texttt{R} package \texttt{prescco} is available at \url{https://github.com/khnhan/prescco}. (commented for anonymity)

\bibliographystyle{apalike}
  \bibliography{predref}

\clearpage
\newpage

\appendix

\begin{center}
{\Large\bf Supplementary Material to ``PRESCCO: Efficient Prediction Intervals under a
Right-Censored Covariate''}
\end{center}

\counterwithin{equation}{section}
\counterwithin{table}{section}
\counterwithin{figure}{section}
\renewcommand{\theHequation}{\thesection.\arabic{equation}}
\renewcommand{\theHtable}{\thesection.\arabic{table}}
\renewcommand{\theHfigure}{\thesection.\arabic{figure}}
\setcounter{Pro}{0}\renewcommand{\thePro}{\thesection.\arabic{Pro}}
\setcounter{Th}{0}\renewcommand{\theTh}{\thesection.\arabic{Th}}
\setcounter{Lem}{0}\renewcommand{\theLem}{\thesection.\arabic{Lem}}
\setcounter{Rem}{0}\renewcommand{\theRem}{\thesection.\arabic{Rem}}
\setcounter{Cor}{0}\renewcommand{\theCor}{\thesection.\arabic{Cor}}

\section{Proofs and technical derivations}
%%%%%%%%%%%%%%%%%%%%%
\subsection{Proof of Theorem~\ref{th:scp}}\label{sec:scp}

\begin{proof}
The result follows from the standard asymptotic normality of sample
quantiles. Conditional on $\wh\bb$, the calibration-set residuals
$r(\bO_{n_1+1},\wh\bb),\ldots,r(\bO_n,\wh\bb)$ are independent draws
from the distribution of $r(\bO_0,\wh\bb)$, and $\zeta\SCP$ is the
corresponding $(1-\alpha)$ population quantile.
\end{proof}

%%%%%%%%%%%%%%%%%%%%%
\subsection{The efficient score for $\bb$ and the tangent space}
\label{sec:supp-tangent}
%%%%%%%%%%%%%%%%%%%%%

This section defines the efficient score $\wt\S_\bb$ that appears in
Proposition~\ref{pro:2} and gives the tangent space used in its proof.
The tangent space $\calT$ is the closure of the linear span of the score
functions of all parametric submodels for $\bb$, $\eta_1$, $\eta_2$,
and $\eta_3$ in the observed-data likelihood \eqref{eq:llhd-structure}
\citep{tsiatis2006semiparametric}.

Let $\S_\bb^F$ and $\S_\bb$ denote the score functions for $\bb$ under
the full-data likelihood and the observed-data likelihood,
\bse
\S_\bb^F(y,x,\bz,\bb) &\equiv& \partial\log\{f_{Y|X,\bZ}(y,x,\bz,\bb)\} /
\partial\bb,\\
\S_\bb(y,w,\delta,\bz,\bb) &\equiv&  \partial\log\{f_{Y,W,\Delta,\bZ}(y,w,\delta,\bz,\bb)\} /
\partial\bb\\
&=& \delta \S_\bb^{F}(y,w,\bz,\bb)+
(1-\delta)\frac{E\{I(X>w)\S_\bb^{F}(y,X,\bz,\bb)\mid
  y,\bz,\bb,\eta_1\}}{E\{I(X>w)\mid y,\bz,\bb,\eta_1\}}.
\ese
The two agree when the time to the event is observed. When it is not,
$\S_\bb$ averages $\S_\bb^F$ over the times to the event beyond $w$, as
the center $m_1$ averages $m_0$. The efficient score subtracts a
function $\ba$ from $\S_\bb^F$ before that averaging, which removes
from $\S_\bb$ the part already accounted for by the time-to-event
model,
\bse
\wt \S_\bb(y,w,\delta,\bz,\bb) &\equiv& \delta \{\S_\bb^{F}(y,w,\bz,\bb) - \ba(w,\bz,\bb)\}\\
&&+
(1-\delta)\frac{E[I(X>w)\{\S_\bb^{F}(y,X,\bz,\bb)-\ba(X,\bz,\bb)\}\mid y,\bz,\bb,\eta_1]}{E\{I(X>w)\mid y,\bz,\bb,\eta_1\}},
\ese
where $\ba(x,\bz,\bb)$ is the solution to
\be
&&E\{I(x\le C)\mid x,\bz,\eta_2\} \ba(x,\bz,\bb) +
E\left[I(x>C)\frac{E\{I(X>C)\ba(X,\bz,\bb)
      \mid C,Y,\bz,\bb,\eta_1\}}{E\{I(X>C)\mid C,Y,\bz,\bb,\eta_1\}}\mid x,\bz,\bb,\eta_2\right] \n\\
&&= E\left[I(x>C)\frac{E\{I(X>C)\S_\bb^{F}(Y,X,\bz,\bb)
      \mid C,Y,\bz,\bb,\eta_1\}}{E\{I(X>C)\mid C,Y,\bz,\bb,\eta_1\}}\mid x,\bz,\bb,\eta_2\right].\label{eq:e10}
\ee
The SPARCC estimator solves $\sumi\wt\S_\bb(\bO_i,\bb)=\0$
\citep{lee2024robust}. Proposition~\ref{pro:tangent} gives the tangent
space as the direct sum of four orthogonal subspaces, one for $\bb$
and one for each nuisance model.

\begin{Pro}\label{pro:tangent}
The tangent space is
$\calT = \wt \Lambda_\bb \oplus \Lambda_1 \oplus \Lambda_2 \oplus \Lambda_3$,
where
\bse
\wt \Lambda_\bb &\equiv& \left\{\bfe\trans \wt
  \S_\bb(y,w,\delta,\bz,\bb):\bfe\in \mathbb{R}^{d_{\bb}}\right\},\\
\Lambda_1 &\equiv& \left[\delta a_1(w,\bz) + (1-\delta) \frac{E\{I(X> w)
    a_1(X,\bz)\mid y,\bz,\bb,\eta_1\}}{E\{I(X > w)\mid y,\bz,\bb,\eta_1\}} : E \{a_1 (X, \bz) \mid
  \bz,\eta_1\} = 0 \right],\\
\Lambda_2 &\equiv& \left[\delta \frac{E\{I(C \ge w) a_2(C,\bz)\mid \bz,\eta_2\}}{E\{I(C \ge w)\mid \bz,\eta_2\}} + (1-\delta)a_2 (w,\bz):
  E \left\{a_2 (C,\bz) \mid   \bz,\eta_2\right\} = 0\right],\\
\Lambda_3 &\equiv&\left[a_3(\bz): E \{a_3 (\bZ)\mid\eta_3\} = 0 \right].
\ese
\end{Pro}

\begin{proof}
Since $f_{Y|X,\bZ}$ is parameterized by $\bb$, the tangent space
associated with $\bb$ is $\Lambda_{\bb}=\{\bfe\trans
\S_\bb(y,w,\delta,\bz) \}$. Following \cite{lee2024robust},
the tangent space associated with $\eta_1$, $\eta_2$, and
$\eta_3$ is $\Lambda_1\oplus\Lambda_2\oplus\Lambda_3$. Hence
$\calT=\Lambda_\bb+\Lambda_1\oplus\Lambda_2\oplus\Lambda_3$.
Since
\bse
&&E\{\ba (X,\bz) \mid \bz\}\\
&=& E\left(E\{I(X\le C)\mid X,\bz\} \ba (X,\bz) +
E\left[I(X>C)\frac{E\{I(X>C)\ba(X,\bz)
      \mid C,Y,\bz\}}{E\{I(X>C)\mid C,Y,\bz\}} \mid X,\bz\right] \mid \bz \right) \\
&=& E\left(E\left[I(X>C)\frac{E\{I(X>C)\S_\bb^F(Y,X,\bz)
      \mid C,Y,\bz\}}{E\{I(X>C)\mid C,Y,\bz\}}\mid X,\bz\right]\mid \bz\right)\\
&=& E\left[I(X>C)\frac{E\{I(X>C)\S_\bb^F(Y,X,\bz)
      \mid C,Y,\bz\}}{E\{I(X>C)\mid C,Y,\bz\}}\mid \bz\right]\\
&=&E[I(X>C)E\{\S_\bb^F(Y,X,\bz)\mid C,X,\bz\}
      \mid \bz]\\
&=&\0,
\ese
we have that
$\wt\S_\bb(y,w,\delta,\bz)-\S_\bb(y,w,\delta,\bz)\in\Lambda_1\subset
\Lambda_1\oplus\Lambda_2\oplus\Lambda_3$. Meanwhile, $\S_\bb(y,w,\delta,\bz)\perp \Lambda_2$ since for any $a_2 (c,\bz)$ with $E \left\{a_2 (C,\bz) \mid \bz\right\} = 0$,
\bse
&&E\left(\S_\bb(Y,W,\Delta,\bZ)\left[\Delta \frac{E\{I(C \ge W) a_2
(C,\bZ)\mid W,\bZ\}}{E\{I(C \ge W)\mid W,\bZ\}} + (1-\Delta)a_2 (W,\bZ)\right]\right)\\
&=& E\left[\Delta \S_\bb^F(Y,X,\bZ)\frac{E\{I(C \ge X) a_2
(C,\bZ)\mid X,\bZ\}}{E\{I(C \ge X)\mid X,\bZ\}}\right.\\
&&\left.+ (1-\Delta)\frac{E\{I(X>C)\S_\bb^F(Y,X,\bZ)\mid
  C,Y,\bZ\}}{E\{I(X>C)\mid C,Y,\bZ\}}a_2 (C,\bZ)\right] \\
&=& E\{I(X \le C)\S_\bb^F(Y,X,\bZ) a_2
(C,\bZ)+ I(X>C)\S_\bb^F(Y,X,\bZ)a_2 (C,\bZ)\} \\
&=& E\{\S_\bb^F(Y,X,\bZ) a_2
(C,\bZ)\} \\
&=& E[E\{\S_\bb^F(Y,X,\bZ) \mid X,C,\bZ\} a_2
(C,\bZ)] \\
&=& \0.
\ese
Also, $\S_\bb(y,w,\delta,\bz)\perp \Lambda_3$ since for any $a_3 (\bz)$ with $E \left\{a_3 (\bZ)\right\} = 0$,
\bse
&&E\{\S_\bb(Y,W,\Delta,\bZ)a_3
(\bZ)\}\\
&=& E\left(E\left[\Delta \S_\bb^F(Y,X,\bZ)+ (1-\Delta)\frac{E\{I(X>C)\S_\bb^F(Y,X,\bZ)\mid
  C,Y,\bZ\}}{E\{I(X>C)\mid C,Y,\bZ\}}\mid \bZ\right]a_3 (\bZ)\right) \\
&=& E[E\{I(X \le C)\S_\bb^F(Y,X,\bZ)+ I(X>C)\S_\bb^F(Y,X,\bZ)\mid \bZ\}a_3 (\bZ)] \\
&=& E[E\{\S_\bb^F(Y,X,\bZ) \mid X,\bZ\} a_3(\bZ)] \\
&=& \0.
\ese
Thus, $\S_\bb(y,w,\delta,\bz)\perp (\Lambda_2 \oplus \Lambda_3)$. Combining this with $\wt\S_\bb(y,w,\delta,\bz)-\S_\bb(y,w,\delta,\bz)\in\Lambda_1\perp(\Lambda_2 \oplus \Lambda_3)$, we get
$\wt\S_\bb(y,w,\delta,\bz)\perp
(\Lambda_2\oplus\Lambda_3)$.  Furthermore, note that
\be
&&E\{\wt\S_\bb(Y,W,\Delta,\bz)\mid x,\bz\} \n\\
&=& E\left(I(x\le C) \{\S_\bb^F(Y,x,\bz) - \ba(x,\bz)\}\right.\n\\
&&\left.+ I(x>C)\frac{E[I(X>C)\{\S_\bb^F(Y,X,\bz)-\ba(X,\bz)\}\mid C,Y,\bz]}{E\{I(X>C)\mid C,Y,\bz\}}\mid x,\bz\right)\n\\
&=& E\{I(x\le C) \S_\bb^F(Y,x,\bz)\mid x,\bz\}\n\\
&=& E\{I(x\le C) \mid x,\bz \}E\{\S_\bb^F(Y,x,\bz)\mid x,\bz\}\n\\
&=&\0, \label{eq:e11}
\ee
where the second equality holds by \eqref{eq:e10}, and the fourth equality holds since $Y\indep C \mid X,\bZ$.
Then
$\wt\S_\bb(y,w,\delta,\bz)\perp
\Lambda_1$ since for any $a_1 (x,\bz)$ with $E \{a_1 (X,\bz)\mid\bz\} = 0$,
\bse
&& E\left(\wt\S_\bb(Y,W,\Delta,\bZ)
  \left[\Delta a_1(X,\bZ) + (1-\Delta) \frac{E\{I(X> C)
    a_1 (X,\bZ)\mid C,Y,\bZ\}}{E\{I(X > C)\mid C, Y,\bZ\}} \right]\right)\\
&=& E\left[\wt\S_\bb(Y,X,1,\bZ)\Delta a_1(X,\bZ)
+I(X>C)\wt\S_\bb(Y,C,0,\bZ)\frac{E\{I(X> C)
    a_1 (X,\bZ)\mid C, Y,\bZ\}}{E\{I(X > C)\mid C, Y,\bZ\}} \right]\\
&=& E\{\wt\S_\bb(Y,X,1,\bZ)\Delta a_1(X,\bZ)
+\wt\S_\bb(Y,C,0,\bZ)I(X> C)
    a_1 (X,\bZ) \}\\
&=&E\{\wt\S_\bb(Y,W,\Delta,\bZ)a_1
(X,\bZ)\}\\
&=&E[E\{\wt\S_\bb(Y,W,\Delta,\bZ)\mid X,\bZ\}a_1
(X,\bZ)]\\
&=&  \0,
\ese
where the last equality holds by \eqref{eq:e11}.

As a result, we obtain that  $\wt\S_\bb(y,w,\delta,\bz)\perp
\Lambda_1\oplus\Lambda_2\oplus\Lambda_3$.
Hence
$\calT
=\wt \Lambda_\bb\oplus\Lambda_1\oplus\Lambda_2\oplus\Lambda_3$.
\end{proof}

%%%%%%%%%%%%%%%%%%%%%
\subsection{Proof of Proposition~\ref{pro:2}}\label{sec:pro2}
%%%%%%%%%%%%%%%%%%%%%

Because the four subspaces in Proposition~\ref{pro:tangent} are
orthogonal, each can be projected onto separately, and the efficient
influence function is the sum of four terms, $\phi_\bb\in\wt\Lambda_\bb$
and $\phi_j\in\Lambda_j$ for $j=1,2,3$.

\begin{proof}
Let
$\eta_1(x,\bz,\bg_1)$, $\eta_2(c,\bz,\bg_2)$, and $\eta_3(\bz,\bg_3)$
be arbitrary parametric submodels for $\eta_1(x,\bz)$, $\eta_2(c,\bz)$,
and $\eta_3(\bz)$, respectively. Let $\zeta=\zeta(\bb,\bg_1,\bg_2,\bg_3)$.
Let $\S_{\bb}(\bo)= \S_\bb(y,w,\delta,\bz,\bb)$, $\wt\S_{\bb}(\bo)= \wt\S_\bb(y,w,\delta,\bz,\bb)$, and $\S_{\bg_j}(\bo)= \partial\log\{f_{Y,W,\Delta,\bZ}(y,w,\delta,\bz)\} /
\partial\bg_j$ for $j=1,2,3$.
Then it suffices to prove that $\phi\eff \in \calT$ and
${\partial\zeta(\bb,\bg_1,\bg_2,\bg_3)}/{\partial\bb} =
E\{\phi\eff(\bO,\zeta)\S_\bb(\bO)\}$ and
${\partial\zeta(\bb,\bg_1,\bg_2,\bg_3)}/{\partial\bg_j}= E\{\phi\eff(\bO,\zeta)\S_{\bg_j}(\bO)\}$
for $j=1, 2, 3$.

Note that $\phi_\bb\in\wt\Lambda_\bb$
and $\phi_j\in\Lambda_j$ for $j=1, 2, 3$, hence
$\phi\eff\in\calT$.
Also, note that $\zeta$ is defined as
\bse
E[I\{r(\bO,\bb) \le \zeta\}]=\pr\{r(\bO,\bb) \le \zeta\} = 1-\alpha.
\ese
Taking the derivative with respect to $\bb$, we obtain
\bse
\0 &=& \frac{\partial E[I\{r(\bO,\bb) \le \zeta\}]}{\partial \bb}\\
 &=& E[\d\{\zeta - r(\bO,\bb)\}\{\frac{\partial \zeta}{\partial \bb} - \frac{\partial r(\bO,\bb)}{\partial \bb}\}] + E[I\{r(\bO,\bb) \le \zeta\}\S_\bb(\bO)],
\ese
which leads to
\bse
\frac{\partial \zeta}{\partial \bb}
 &=&  \frac{E[\d\{\zeta - r(\bO,\bb)\}\partial r(\bO,\bb)/\partial \bb - I\{r(\bO,\bb) \le \zeta\}\S_\bb(\bO)]}{E[\d\{\zeta - r(\bO,\bb)\}]}.
\ese
Similarly, for $\bg_1$, $\bg_2$, and $\bg_3$, we have
\bse
\frac{\partial \zeta}{\partial \bg_j}
 = \frac{E[\d\{\zeta - r(\bO,\bb)\}\partial r(\bO,\bb)/\partial \bg_j -
   I\{r(\bO,\bb) \le \zeta\}\S_{\bg_j}(\bO)]}{E[\d\{\zeta - r(\bO,\bb)\}]}=-\frac{E[
   I\{r(\bO,\bb) \le \zeta\}\S_{\bg_j}(\bO)]}{E[\d\{\zeta - r(\bO,\bb)\}]}.
\ese

First, let $\S_{\bg_1}^F(x,\bz) = \partial\log\{f_{X|\bZ}(x,\bz,\bg_1)\}/\partial\bg_1$. Then
\bse
\S_{\bg_1}(\bo) = \delta \S_{\bg_1}^F(w,\bz) + (1-\delta) \frac{E\{I(X> w)
    \S_{\bg_1}^F(X,\bz)\mid y,\bz\}}{E\{I(X > w)\mid y,\bz\}}.
\ese
Since $\S_{\bg_1}(\bo) \in \Lambda_1$, we have that
\bse
&&E\{\phi\eff(\bO,\zeta)\S_{\bg_1}(\bO)\} - \frac{\partial \zeta}{\partial \bg_1}\\
&=&E\{\phi_1(\bO,\zeta)\S_{\bg_1}(\bO)\} - \frac{\partial \zeta}{\partial \bg_1}\\
&=& E\left\{\left(\frac{I\{r(\bO,\bb) \le \zeta\}}{E[\d\{\zeta - r(\bO,\bb)\}]} +\Delta a_1(W,\bZ) + (1-\Delta) \frac{E\{I(X> W)
    a_1(X,\bZ)\mid W,Y,\bZ\}}{E\{I(X > W)\mid W,Y,\bZ\}}\right)\right.\\
&&\left.\times[\Delta \S_{\bg_1}^F(W,\bZ) + (1-\Delta) \frac{E\{I(X> W)
    \S_{\bg_1}^F(X,\bZ)\mid W,Y,\bZ\}}{E\{I(X > W)\mid W,Y,\bZ\}}]\right\}.
\ese

For any function $h(c,y,\bz)$, we know that
\bse
&&E[(1-\Delta) h(W,Y,\bZ) \frac{E\{I(X> W)
    \S_{\bg_1}^F(X,\bZ)\mid W,Y,\bZ\}}{E\{I(X > W)\mid W,Y,\bZ\}}]\\
&=& E[ h(C,Y,\bZ) E\{I(X> C)
    \S_{\bg_1}^F(X,\bZ)\mid C,Y,\bZ\}]\\
&=& E\{I(X> C) h(C,Y,\bZ)
    \S_{\bg_1}^F(X,\bZ)\}\\
&=& E\{(1-\Delta) h(W,Y,\bZ)
    \S_{\bg_1}^F(X,\bZ)\},
\ese
so we get
\bse
&&E\{\phi\eff(\bO,\zeta)\S_{\bg_1}(\bO)\} - \frac{\partial \zeta}{\partial \bg_1}\\
&&= E\left\{\left(\frac{I\{r(\bO,\bb) \le \zeta\}}{E[\d\{\zeta - r(\bO,\bb)\}]} + \Delta a_1(W,\bZ) + (1-\Delta) \frac{E\{I(X> W)
    a_1(X,\bZ)\mid W,Y,\bZ\}}{E\{I(X > W)\mid W,Y,\bZ\}}\right)\S_{\bg_1}^F(X,\bZ)\right\}.
\ese
Using the definition of $a_1(X,\bZ)$, we further obtain
\bse
&& E\{\phi\eff(\bO,\zeta)\S_{\bg_1}(\bO)\}  - \frac{\partial \zeta}{\partial \bg_1}\\
&&= E\left\{\left(\frac{I\{r(\bO,\bb) \le \zeta\}}{E[\d\{\zeta - r(\bO,\bb)\}]} - \frac{E[I\{r(\bO,\bb) \le \zeta\}\mid X,\bZ] - E[I\{r(\bO,\bb) \le \zeta\}\mid \bZ]}{E[\d\{\zeta - r(\bO,\bb)\}]}\right)\S_{\bg_1}^F(X,\bZ)\right\}\\
&&= E\left(\frac{E[I\{r(\bO,\bb) \le \zeta\}\mid \bZ]}{E[\d\{\zeta - r(\bO,\bb)\}]}\S_{\bg_1}^F(X,\bZ)\right)\\
&&= E\left(\frac{E[I\{r(\bO,\bb) \le \zeta\}\mid \bZ]}{E[\d\{\zeta - r(\bO,\bb)\}]}E\{\S_{\bg_1}^F(X,\bZ)\mid \bZ\}\right)\\
&&=\0,
\ese
 i.e., ${\partial \zeta}/{\partial \bg_1}=E\{\phi\eff(\bO,\zeta)\S_{\bg_1}(\bO)\}$.

Second, let $\S_{\bg_2}^F(c,\bz) = \partial\log\{f_{C|\bZ}(c,\bz,\bg_2)\}/\partial \bg_2$. Then
\bse
\S_{\bg_2}(\bo) = \delta \frac{E\{I(C \ge w) \S_{\bg_2}^F
    (C,\bz)\mid \bz\}}{E\{I(C \ge w)\mid \bz\}} + (1-\delta)\S_{\bg_2}^F (w,\bz).
\ese
Since $\S_{\bg_2}(\bo) \in \Lambda_2$, given
  $C\indep Y\mid X,\bZ$, we have
\bse
&&E\{\phi\eff(\bO,\zeta)\S_{\bg_2}(\bO)\} - \frac{\partial \zeta}{\partial \bg_2}\\
&=&E\{\phi_2(\bO,\zeta)\S_{\bg_2}(\bO)\} - \frac{\partial \zeta}{\partial \bg_2}\\
&=& E\left\{\left(\frac{I\{r(\bO,\bb) \le \zeta\}}{E[\d\{\zeta - r(\bO,\bb)\}]} + \Delta \frac{E\{I(C \ge W) a_2
    (C,\bZ)\mid W,\bZ\}}{E\{I(C \ge W)\mid W,\bZ\}}+ (1-\Delta)a_2
    (W,\bZ)\right)\right.\\
    &&\times \left. [\Delta \frac{E\{I(C \ge W) \S_{\bg_2}^F (C,\bZ)\mid W,\bZ\}}{E\{I(C \ge W)\mid W,\bZ\}} + (1-\Delta)\S_{\bg_2}^F (W,\bZ)] \right\}.
\ese
For any function $h(x,y,\bz)$, we know that
\bse
&&E[\Delta h(W,Y,\bZ) \frac{E\{I(C \ge W) \S_{\bg_2}^F (C,\bZ)\mid W,\bZ\}}{E\{I(C \ge W)\mid W,\bZ\}}]\\
&=&E[E\{I(C \ge X)\mid X,Y,\bZ\} h(X,Y,\bZ) \frac{E\{I(C \ge X) \S_{\bg_2}^F (C,\bZ)\mid X,\bZ\}}{E\{I(C \ge X)\mid X,\bZ\}}]\\
&=&E[h(X,Y,\bZ) E\{I(C \ge X) \S_{\bg_2}^F (C,\bZ)\mid X,\bZ\}]\\
&=& E\{I(X \le C) h(X,Y,\bZ) \S_{\bg_2}^F (C,\bZ)\}\\
&=& E\{\Delta h(W,Y,\bZ)\S_{\bg_2}^F(C,\bZ)\},
\ese
so we get
\bse
&&E\{\phi\eff(\bO,\zeta)\S_{\bg_2}(\bO)\} - \frac{\partial \zeta}{\partial \bg_2}\\
&&= E\left\{\left(\frac{I\{r(\bO,\bb) \le \zeta\}}{E[\d\{\zeta - r(\bO,\bb)\}]} + \Delta \frac{E\{I(C \ge W) a_2
    (C,\bZ)\mid W,\bZ\}}{E\{I(C \ge W)\mid W,\bZ\}}+ (1-\Delta)a_2
    (W,\bZ)\right)\S_{\bg_2}^F (C,\bZ) \right\}.
\ese
Using the definition of $a_2(c,\bz)$, we have
\bse
&&E\{\phi\eff(\bO,\zeta)\S_{\bg_2}(\bO)\} - \frac{\partial \zeta}{\partial \bg_2}\\
&&= E\left\{\left(\frac{I\{r(\bO,\bb) \le \zeta\}}{E[\d\{\zeta - r(\bO,\bb)\}]} - \frac{E[I\{r(\bO,\bb) \le \zeta\}\mid C,\bZ] - E[I\{r(\bO,\bb) \le \zeta\}\mid \bZ]}{E[\d\{\zeta - r(\bO,\bb)\}]}\right)\S_{\bg_2}^F (C,\bZ) \right\}\\
&&= E\left(\frac{E[I\{r(\bO,\bb) \le \zeta\}\mid \bZ]}{E[\d\{\zeta - r(\bO,\bb)\}]}\S_{\bg_2}^F (C,\bZ) \right)\\
&&= E\left(\frac{E[I\{r(\bO,\bb) \le \zeta\}\mid \bZ]}{E[\d\{\zeta - r(\bO,\bb)\}]}E\{\S_{\bg_2}^F (C,\bZ) \mid \bZ\}\right)\\
&&= \0,
\ese
i.e., ${\partial \zeta}/{\partial \bg_2}=E\{\phi\eff(\bO,\zeta)\S_{\bg_2}(\bO)\}$.

Third, let $\S_{\bg_3}^F(\bz) = \partial\log\{f_{\bZ}(\bz,\bg_3)\}/\partial \bg_3$. Then $\S_{\bg_3}(\bo) = \S_{\bg_3}^F(\bz)$.
Since $\S_{\bg_3}(\bo) \in \Lambda_3$, we have that
\bse
&&E\{\phi\eff(\bO,\zeta)\S_{\bg_3}(\bO)\} - \frac{\partial \zeta}{\partial \bg_3}\\
&&=E\{\phi_3(\bO,\zeta)\S_{\bg_3}(\bO)\} - \frac{\partial \zeta}{\partial \bg_3}\\
&&= E\left\{\left(\frac{I\{r(\bO,\bb) \le \zeta\}}{E[\d\{\zeta - r(\bO,\bb)\}]} + a_3(\bZ)\right)\S_{\bg_3}^F(\bZ)\right\}\\
&&= E\left\{\left(\frac{I\{r(\bO,\bb) \le \zeta\}}{E[\d\{\zeta - r(\bO,\bb)\}]} - \frac{E[I\{r(\bO,\bb) \le \zeta\}\mid \bZ] - E[I\{r(\bO,\bb) \le \zeta\}]}{E[\d\{\zeta - r(\bO,\bb)\}]}\right)\S_{\bg_3}^F(\bZ)\right\}\\
&&= E\left(\frac{E[I\{r(\bO,\bb) \le \zeta\}]}{E[\d\{\zeta - r(\bO,\bb)\}]}\S_{\bg_3}^F(\bZ)\right)\\
&&= \0,
\ese
i.e., ${\partial \zeta}/{\partial \bg_3}=E\{\phi\eff(\bO,\zeta)\S_{\bg_3}(\bO)\}$.

Lastly, noting that  $\S_\bb(\bo) = \wt\S_\bb(\bo) + \{\S_\bb(\bo) -
\wt\S_\bb(\bo)\} \in \wt\Lambda_\bb \oplus\Lambda_1$,
we have
\bse
&&E\{\phi\eff(\bO,\zeta)\S_{\bb}(\bO)\}\\
&=& E(\{\phi_\bb(\bO,\zeta)+\phi_1(\bO,\zeta)\}[\wt\S_\bb(\bO) +
\{\S_\bb(\bO) - \wt\S_\bb(\bO)\}])\\
&=& E\{\phi_\bb(\bO,\zeta)\wt\S_\bb(\bO) \}
+ E[\phi_1(\bO,\zeta)\{\S_\bb(\bO) - \wt\S_\bb(\bO)\}]\\
&=& E[\{\bfe\trans\wt\S_\bb(Y,W,\Delta,\bZ)\}\wt\S_\bb(Y,W,\Delta,\bZ) ]\\
&& +
   E\left([\Delta a_1(W,\bZ) + (1-\Delta) \frac{E\{I(X> W)
    a_1 (X,\bZ)\mid W,Y,\bZ\}}{E\{I(X > W)\mid W,Y,\bZ\}}]\right.\\
&&\left.\times[\Delta \ba(W,\bZ) + (1-\Delta) \frac{E\{I(X> W)
       \ba (X,\bZ)\mid W,Y,\bZ\}}{E\{I(X > W)\mid W,Y,\bZ\}}]\right)\\
&=&E\left[\wt\S_\bb(Y,W,\Delta,\bZ)^{\otimes2}\bfe +
\Delta a_1(W,\bZ)\ba(W,\bZ)\right.\\
&&\left.+ (1-\Delta) \frac{E\{I(X> W)
    a_1(X,\bZ)\mid W,Y,\bZ\}}{E\{I(X > W)\mid W,Y,\bZ\}}\frac{E\{I(X> W)
    \ba (X,\bZ)\mid W,Y,\bZ\}}{E\{I(X > W)\mid W,Y,\bZ\}}\right]\\
&=& E\{\wt\S_\bb(Y,W,\Delta,\bZ)^{\otimes2}\}\bfe
+E\left[\Delta a_1 (W,\bZ)\ba(X,\bZ)\right.\\
&&\left.+ (1-\Delta) \frac{E\{I(X> W)
    a_1 (X,\bZ)\mid W,Y,\bZ\}}{E\{I(X > W)\mid W,Y,\bZ\}}\ba (X,\bZ)\right]\\
&=&
\frac{E[\d\{\zeta - r(\bO,\bb)\}\partial r(\bO,\bb)/\partial \bb]}{E[\d\{\zeta - r(\bO,\bb)\}]}-\frac{E[I\{r(\bO,\bb) \le \zeta\} \wt\S_\bb(\bO)]}{E[\d\{\zeta - r(\bO,\bb)\}]}\\
&&-
E\left(
    \frac{E[I\{r(\bO,\bb) \le \zeta\}\mid X,\bZ]-E[I\{r(\bO,\bb) \le
        \zeta\}\mid \bZ]}{E[\d\{\zeta -
        r(\bO,\bb)\}]}\ba(X,\bZ)\right)\\
&=&
\frac{E[\d\{\zeta - r(\bO,\bb)\}\partial r(\bO,\bb)/\partial \bb
-I\{r(\bO,\bb) \le \zeta\}\S_\bb(\bO)
]}{E[\d\{\zeta - r(\bO,\bb)\}]}\\
&=&\frac{\partial\zeta}{\partial\bb},
\ese
where the third-to-last equality holds by the definition of $a_1 (x,\bz)$
  and the definition of $\bfe$, and the second-to-last
equality uses $E\{\S_\bb(\bO)-\ba(X,\bZ)\mid\bO\}=\wt\S_\bb(\bO)$ and $E\{\ba(X,\bZ)\mid\bZ\}=\0$.
\end{proof}

%%%%%%%%%%%%%%%%%%%%%
\subsection{Derivation for Remark~\ref{rem:no-censoring}}\label{sec:no-censoring}
%%%%%%%%%%%%%%%%%%%%%

We derive the efficient influence function $\phi\eff(\bo,\zeta,\bb)$
and the estimator $\wh\zeta$ of Remark~\ref{rem:no-censoring}, where
no patient's time to the event is right-censored, that is,
$\pr(X\le C\mid \eta_1,\eta_2)=1$. In this case, $\Delta=1$ and $W=X$,
so $m_1(W,\Delta,\Z,\bb)$ and $m_2(W,\Delta,\Z,\bb)$ both reduce to
$m_0(X,\Z,\bb)$.

Since $r(\bO,\bb)=|Y-m_0(X,\Z,\bb)|$ is independent of $C$ given
$\Z$, we have
\bse
 a_2(c,\z) = - \frac{E[I\{|Y-m_0(X,\z,\bb)| \le \zeta\}\mid c,\z,\bb,\eta_1] - E[I\{|Y-m_0(X,\z,\bb)| \le \zeta\}\mid \z,\bb,\eta_1,\eta_2]}{E[\d\{\zeta - r(\bO,\bb)\}\mid\bb,\eta_1,\eta_2]} = 0.
\ese
Also, $a_1(x,\bz)$ simplifies to
\bse
a_1(x,\bz)
&=& - \frac{E[I\{|Y - m_0(x, \z,\bb)|\le \zeta\}\mid x,\z,\bb] - E[I\{|Y - m_0(x, \z,\bb)| \le \zeta\}\mid \z,\bb,\eta_1,\eta_2]}{E[\d\{\zeta - r(\bO,\bb)\}\mid\bb,\eta_1,\eta_2]}.
\ese
Moreover, $\ba(x,\z,\bb)=\0$ and
$\wt\S_\bb(y,w,\delta,\z,\bb)=\S_\bb^F(y,x,\z,\bb)$, which gives
$\phi_\bb(\bo,\zeta,\bb) = \bfe\trans\S_\bb^F(y,x,\z,\bb)$. Then
$\phi\eff(\bo,\zeta,\bb)$ can be written as
\be \label{eq:phieff-nocensoring}
\qquad\phi\eff(\bo,\zeta,\bb)=
\bfe\trans\S_\bb^F(y,x,\z,\bb) - \frac{E[I\{|Y - m_0(x, \z,\bb)| \le \zeta\}\mid x,\z,\bb] - (1-\alpha)}{E[\d\{\zeta - r(\bO,\bb)\}\mid\bb,\eta_1,\eta_2]}.
\ee

Suppose that we choose $\wh\bb$ to be the maximum likelihood estimator
for $\bb$, that is, the solution of
$\sumi\S_\bb^F(Y_i,X_i,\Z_i,\bb)=\0$. Then $\wh\zeta$ solves
\bse
n^{-1}\sumi E[I\{|Y_i - m_0(X_i, \Z_i,\wh\bb)|\le \zeta\}\mid X_i,\Z_i,\wh\bb]=1-\alpha.
\ese
Since the conditional expectation of an indicator is a conditional
probability, this equation is
\bse
n^{-1}\sumi \pr\{|Y_i - m_0(X_i,\Z_i,\wh\bb)|\le \zeta\mid
X_i,\Z_i,\wh\bb\}=1-\alpha,
\ese
which is \eqref{eq:no-cens-ee}. If the outcome model satisfies
$Y = m_0(X, \Z,\bb)+\epsilon$, where the distribution of $\epsilon$
does not depend on $\bb$, then the estimating equation becomes
\bse
n^{-1}\sumi \pr(|\epsilon | \le \zeta\mid X_i,\Z_i)=1-\alpha,
\ese
which depends only on the covariates $X_i$ and $\Z_i$. If in addition
the conditional distribution of $\epsilon$ given $(X,\Z)$ does not
depend on $(X,\Z)$, then the estimating equation becomes
$\pr(|\epsilon | \le \zeta)=1-\alpha$, so $\wh\zeta$ is the
$(1-\alpha)$ quantile of $|\epsilon|$, a nonrandom quantity that does
not depend on $\bO_1,\ldots,\bO_n$.

%%%%%%%%%%%%%%%%%%%%%
\subsection{Proof of Theorem \ref{th:doubly-robust}}\label{sec:doubly-robust}
%%%%%%%%%%%%%%%%%%%%%

Theorem~\ref{th:doubly-robust} assumes the following conditions,
standard for the consistency of estimators that solve an estimating
equation \citep{newey1994large}. In
Conditions~\ref{con:a1}--\ref{con:a3}, $\theta$ and $\brho$ denote
generic parameters in the same spaces as $\zeta$ and $\bb$,
respectively.
\begin{enumerate}[label=(A\arabic*),ref=(A\arabic*),start=1]
    \item\label{con:a1}
    The true parameter $\zeta$ is contained in a compact set $\Omega$ in
    $\mathbb{R}$, and $E\{\phi\eff^{*\star}(\bO,\theta,\bb)\}=0$ has a
    unique solution $\theta=\zeta$ in $\Omega$.
    \item\label{con:a2}
    $E\{\sup_{\theta,\brho}|\phi\eff^{*\star}(\bO,\theta,\brho)|\}<\infty$.
    \item\label{con:a3}
    The mapping $(\theta,\brho\trans)\trans \mapsto
    \phi\eff^{*\star}(\bo,\theta,\brho)$ is continuous for all
    $\bo=(y,w,\delta,\z)$.
    \item\label{con:a4}
    The estimator $\wh\bb$ is consistent for $\bb$.
\end{enumerate}
The proof rests on Lemma~\ref{lem:l1}, which shows that
$\phi\eff^{*\star}$ has mean zero at $\zeta$ when either working model
is correctly specified.

\begin{Lem} \label{lem:l1}
\begin{enumerate}
[label=(\roman*),ref=(\roman*),start=1]
    \item \label{item1} If $\eta_1^* = \eta_1$, then $E\{\phi_\bb^{\star}(\bO,\zeta,\bb)\} = E\{\phi_1^{\star}(\bO,\zeta,\bb)\} = E\{\phi_2^{\star}(\bO,\zeta,\bb)+\phi_3^{\star}(\bO,\zeta,\bb)\} = 0$. Thus, $E\{\phi\eff^{\star}(\bO,\zeta,\bb)\}=0$.
    \item \label{item2} If $\eta_2^\star = \eta_2$, then $E\{\phi_\bb^{*}(\bO,\zeta,\bb)\} = E\{\phi_1^{*}(\bO,\zeta,\bb)+\phi_3^{*}(\bO,\zeta,\bb)\} = E\{\phi_2^{*}(\bO,\zeta,\bb)\} = 0$. Thus, $E\{\phi\eff^{*}(\bO,\zeta,\bb)\}=0$.
\end{enumerate}
\end{Lem}
%%%%%%%%%%%%%%%%%%%%%%%%
%%%%%%%%%%%%%%%%%%%%%%%%
\begin{proof}[Proof of Lemma \ref{lem:l1}]
\ref{item1} Suppose that $\eta_1^* = \eta_1$. Then 
\bse
&&E\{\wt\S_\bb^{\star}(Y,W,\Delta,\Z)\mid\bb,\eta_1,\eta_2\}\n\\
&=&E\left(\Delta \{\S_\bb^{F}(Y,W,\bZ) - \ba^{\star}(W,\Z)\}\right.\n\\
&&+
\left.(1-\Delta)\frac{E[I(X>W)\{\S_\bb^{F}(Y,X,\bZ)-\ba^{\star}(X,\Z)\}\mid W,Y,\Z,\bb,\eta_1]}{E\{I(X>W)\mid W,Y,\Z,\bb,\eta_1\}}\mid\bb,\eta_1,\eta_2\right)\n\\
&=&E\{\S_\bb^{F}(Y,X,\bZ) - \ba^{\star}(X,\Z)\mid \bb,\eta_1\}\n\\
&=& E\{\S_\bb^{F}(Y,X,\bZ) - \ba^{\star}(X,\Z)\mid \bb,\eta_1,\eta_2^\star\}\n\\
&=&E\left\{E\left(\Delta \{\S_\bb^{F}(Y,W,\bZ) - \ba^{\star}(W,\Z)\}\right.\right.\n\\
&&\left.\left.+(1-\Delta)\frac{E[I(X>W)\{\S_\bb^{F}(Y,X,\bZ)-\ba^{\star}(X,\Z)\}\mid W,Y,\Z,\bb,\eta_1]}{E\{I(X>W)\mid W,Y,\Z,\bb,\eta_1\}}\mid X,\Z,\bb,\eta_1,\eta_2^\star\right)\mid\eta_1\right\}\n\\
&=&E[E\{\Delta \S_\bb^{F}(Y,W,\bZ)\mid X,\Z,\bb,\eta_2^\star\}\mid\eta_1]\n\\
&=&E[I(X\le C)E\{\S_\bb^{F}(Y,X,\bZ)\mid X,C,\Z,\bb\}\mid\eta_1,\eta_2^\star]\n\\
&=&\0,
\ese
where the first equality holds by the definition of
$\wt\S_\bb^{\star}$, and the fifth equation holds by the definition
of $\ba^{\star}$. Thus, $E\{\phi_\bb^{\star}(\bO,\zeta)\}=0$. Next, 
\bse
&&E\{\phi_1^{\star}(\bO,\zeta)\mid\bb,\eta_1,\eta_2\}\\
&=& E\left[\Delta a_1^{\star}(W,\bZ) + (1-\Delta) \frac{E\{I(X> W)
a_1^{\star}(X,\bZ)\mid W,Y,\Z,\bb,\eta_1\}}{E\{I(X > W)\mid W,Y,\Z,\bb,\eta_1\}}\mid \bb,\eta_1,\eta_2\right]\n\\
&=& E\{a_1^{\star}(X,\bZ)\mid\eta_1\}\n\\
&=& E[E\{a_1^{\star}(X,\bZ)\mid\bZ,\eta_1\}]\n\\
&=& 0,
\ese
where the first equality holds by the definition of $\phi_1^{\star}$. 
Also,
\bse
&&E\{\phi_2^{\star}(\bO,\zeta)+\phi_3^{\star}(\bO,\zeta)\mid\bb,\eta_1,\eta_2\}\\
&=& E\left[\Delta \frac{E\{I(C \ge W) a_2^{\star}
    (C,\Z)\mid W,\Z,\eta_2^\star\}}{E\{I(C \ge W)\mid W,\Z,\eta_2^\star\}}+ (1-\Delta)a_2^{\star}
    (W,\Z)+a_3^{\star}(\Z)\mid \eta_1,\eta_2\right]\n\\
&=& E\left(E\left[\Delta \frac{E\{I(C \ge W) a_2^{\star}
    (C,\Z)\mid W,\Z,\eta_2^\star\}}{E\{I(C \ge W)\mid W,\Z,\eta_2^\star\}}+ (1-\Delta)a_2^{\star}
    (W,\Z) \mid C,\Z,\eta_1\right]+a_3^{\star}(\Z)\mid\eta_2\right)\n\\
&=& E\left(- \frac{E[I\{r(\bO,\bb) \le \zeta\}\mid C,\Z,\bb,\eta_1] - E[I\{r(\bO,\bb) \le \zeta\}\mid \Z,\bb,\eta_1,\eta_2^\star]}{E[\d\{\zeta - r(\bO,\bb)\}\mid\bb,\eta_1^*,\eta_2^\star]}\right.\\
&&\left.- \frac{E[I\{r(\bO,\bb) \le \zeta\}\mid \Z,\bb,\eta_1,\eta_2^\star] - E[I\{r(\bO,\bb) \le \zeta\}\mid\bb,\eta_1,\eta_2]}{E[\d\{\zeta - r(\bO,\bb)\}\mid\bb,\eta_1^*,\eta_2^\star]}\mid\eta_2\right)\\
&=& E\left(- \frac{E[I\{r(\bO,\bb) \le \zeta\}\mid C,\Z,\bb,\eta_1] - E[I\{r(\bO,\bb) \le \zeta\}\mid\bb,\eta_1,\eta_2]}{E[\d\{\zeta - r(\bO,\bb)\}\mid\bb,\eta_1^*,\eta_2^\star]}\mid\eta_2\right)\\
&=& 0,
\ese
where the first equality holds by the definitions of $\phi_2^{\star}$
and $\phi_3^{\star}$, and the third equality holds by the definitions
of $a_2^{\star}$ and $a_3^{\star}$.

\ref{item2} Suppose that $\eta_2^\star = \eta_2$. Then
\bse
&&E\{\wt\S_\bb^{*}(Y,W,\Delta,\Z)\mid\bb,\eta_1,\eta_2\}\\
&=&E\left(\Delta \{\S_\bb^{F}(Y,W,\bZ) - \ba^{*}(W,\Z)\}\right.\\
&&+
\left.(1-\Delta)\frac{E[I(X>W)\{\S_\bb^{F}(Y,X,\bZ)-\ba^{*}(X,\Z)\}\mid W,Y,\Z,\bb,\eta_1^*]}{E\{I(X>W)\mid W,Y,\Z,\bb,\eta_1^*\}}\mid\bb,\eta_1,\eta_2\right)\\
&=&E\left\{E\left(\Delta \{\S_\bb^{F}(Y,W,\bZ) - \ba^{*}(W,\Z)\}\right.\right.\\
&&+
\left.\left.(1-\Delta)\frac{E[I(X>W)\{\S_\bb^{F}(Y,X,\bZ)-\ba^{*}(X,\Z)\}\mid W,Y,\Z,\bb,\eta_1^*]}{E\{I(X>W)\mid W,Y,\Z,\bb,\eta_1^*\}}\mid X,\Z,\bb,\eta_2\right)\mid\eta_1\right\}\\
&=&E[E\{\Delta \S_\bb^{F}(Y,W,\bZ)\mid X,\Z,\bb,\eta_2\}\mid\eta_1]\\
&=&E[I(X\le C)E\{\S_\bb^{F}(Y,X,\bZ)\mid X,C,\Z,\bb\}\mid\eta_1,\eta_2]\\
&=& \0,
\ese
where the first equality holds by the definition of $\wt\S_\bb^{*}$,
and the third equality holds by the definition of $\ba^{*}$. 
  This implies that $E\{\phi_\bb^{*}(\bO,\zeta)\} = 0$.
Next, 
\bse
&&E\{\phi_1^{*}(\bO,\zeta)+\phi_3^{*}(\bO,\zeta)\mid\bb,\eta_1,\eta_2\}\\
&=& E\left[\Delta a_1^{*}(W,\bZ) + (1-\Delta) \frac{E\{I(X> W)
a_1^{*}(X,\bZ)\mid W,Y,\Z,\bb,\eta_1^*\}}{E\{I(X > W)\mid W,Y,\Z,\bb,\eta_1^*\}}+a_3^{*}(\bZ)\mid \bb,\eta_1,\eta_2\right]\n\\
&=& E\left(E\left[\Delta a_1^{*}(W,\bZ) + (1-\Delta) \frac{E\{I(X> W)
a_1^{*}(X,\bZ)\mid W,Y,\Z,\bb,\eta_1^*\}}{E\{I(X > W)\mid W,Y,\Z,\bb,\eta_1^*\}}\mid X,\Z,\bb,\eta_2\right]+a_3^{*}(\bZ)\mid \eta_1\right)\n\\
&=& E\left(-\frac{E[I\{r(\bO,\bb) \le \zeta\}\mid X,\Z,\bb,\eta_2] - E[I\{r(\bO,\bb) \le \zeta\}\mid \Z,\bb,\eta_1^*,\eta_2]}{E[\d\{\zeta - r(\bO,\bb)\}\mid\bb,\eta_1^*,\eta_2]}\right.\\
&&\left.-\frac{E[I\{r(\bO,\bb) \le \zeta\}\mid \Z,\bb,\eta_1^*,\eta_2]-E[I\{r(\bO,\bb) \le \zeta\}\mid \bb,\eta_1,\eta_2]}{E[\d\{\zeta - r(\bO,\bb)\}\mid\bb,\eta_1^*,\eta_2]}\mid \eta_1\right)\n\\
&=& E\left(-\frac{E[I\{r(\bO,\bb) \le \zeta\}\mid X,\Z,\bb,\eta_2] -E[I\{r(\bO,\bb) \le \zeta\}\mid \bb,\eta_1,\eta_2]}{E[\d\{\zeta - r(\bO,\bb)\}\mid\bb,\eta_1^*,\eta_2]}\mid \eta_1\right)\n\\
&=& 0,
\ese
where the first equality holds by the definitions of $\phi_1^{*}$ and
$\phi_3^{*}$, and the third equality holds by the definitions of
$a_1^{*}$ and $a_3^{*}$. Lastly,
\bse
&&E\{\phi_2^{*}(\bO,\zeta)\mid\bb,\eta_1,\eta_2\}\\
&=& E\left[\Delta \frac{E\{I(C \ge W) a_2^{*}
    (C,\Z)\mid W,\Z,\eta_2\}}{E\{I(C \ge W)\mid W,\Z,\eta_2\}}+ (1-\Delta)a_2^{*}
    (W,\Z)\mid \eta_1,\eta_2\right]\n\\
&=& E[E\{a_2^{*}(C,\Z)\mid \eta_1\}\mid\eta_2]\n\\
&=& 0,
\ese
where the first equality holds by the definition of $\phi_2^{*}$.
\end{proof}

\begin{proof}[Proof of Theorem \ref{th:doubly-robust}]
First, by Conditions
\ref{con:a1}--\ref{con:a3}, we apply 
\cite{newey1994large}, Lemma 2.4 to obtain 
that
\bse
\sup_{\|\brho-\bb\|_2\le\epsilon}\sup_{\theta\in\Omega}|n^{-1}\sumi\phi\eff^{*\star}
(\bO_i,\theta,\brho)- E\{\phi\eff^{*\star}
(\bO,\theta,\brho)\}|=o_p(1)
\ese
for any $\epsilon>0$.
By Condition \ref{con:a4}, there exists a sequence $(\epsilon_n)$ such that $\epsilon_n \downarrow 0$ and $P(\|\wh\bb-\bb\|_2>\epsilon_n)\to 0$. Then we have that
\be
&&\sup_{\theta\in\Omega}|n^{-1}\sumi\phi\eff^{*\star}
(\bO_i,\theta,\wh\bb)- E\{\phi\eff^{*\star}
(\bO,\theta,\bb)\}|\n\\
&\le&\sup_{\|\brho-\bb\|_2\le\epsilon_n}\sup_{\theta\in\Omega}|n^{-1}\sumi\phi\eff^{*\star}
(\bO_i,\theta,\brho)- E\{\phi\eff^{*\star}
(\bO,\theta,\brho)\}|\n\\
&&+\sup_{\|\brho-\bb\|_2\le\epsilon_n}\sup_{\theta\in\Omega}|E\{\phi\eff^{*\star}
(\bO,\theta,\brho)\}- E\{\phi\eff^{*\star}
(\bO,\theta,\bb)\}|+ o_p(1)\n\\
&=& o_p(1),\label{eq:e13}
\ee
where the first term of the third line converges to zero by Condition
\ref{con:a3}. 

Let $Q_0(\theta)=-[E\{\phi\eff^{*\star} (\bO,\theta,\bb)\}]^2$ and
$\wh Q_n(\theta)=-[n^{-1}\sumi\phi\eff^{*\star}
(\bO_i,\theta,\wh\bb)]^2$.
Then $Q_0(\theta)$ is uniquely maximized at
  $\theta=\zeta$ in $\Omega$,
    and $\Omega$ is compact by Condition
  \ref{con:a1}. Also, $Q_0(\theta)$ is continuous by Conditions
  \ref{con:a2} and \ref{con:a3} with the dominated convergence
  theorem. Lastly, by \eqref{eq:e13} and Condition \ref{con:a3}, $\wh
  Q_n(\theta)$ converges uniformly in probability to $Q_0(\theta)$ on
  $\theta \in \Omega$. Thus, by Lemma~\ref{lem:l1}, we can apply 
  \cite{newey1994large}, Theorem 2.1 to obtain that $\wh\zeta$ is consistent for
  $\zeta$. 
\end{proof}
 %%%%%%%%%%%%%%%%%%%%%

%%%%%%%%%%%%%%%%%%%%%
\subsection{Proof of Theorem \ref{th:asympt-normal}}\label{sec:asympt-normal}
%%%%%%%%%%%%%%%%%%%%%

Theorem~\ref{th:asympt-normal} assumes
Conditions~\ref{con:a1}--\ref{con:a4} and the following three,
standard in semiparametric theory.
\begin{enumerate}[label=(A\arabic*),ref=(A\arabic*),start=5]
    \item\label{con:a5}
    $E\{\sup_{\theta,\brho}\|\partial \phi\eff^{*\star}
    (\bO,\theta,\brho)/\partial(\theta,\brho\trans)\trans\|_2\}<\infty$.
    \item\label{con:a6}
    The mapping $(\theta,\brho\trans)\trans \mapsto \partial
    \phi\eff^{*\star} (\bo,\theta,\brho)/\partial\theta$ is continuous
    for all $\bo=(y,w,\delta,\z)$.
    \item\label{con:a7}
    The estimator $\wh\bb$ is asymptotically linear with influence
    function $\bxi(\bo,\bb)$, that is,
    \bse
    \wh\bb-\bb=n^{-1}\sumi \bxi(\bO_i,\bb)+o_p(n^{-1/2}),
    \ese
    which holds for the estimators of $\bb$ available in the
    right-censored covariate setting \citep{matsouaka2020regression,
vazquez2026estimators,  lee2024robust}.
\end{enumerate}
Theorem~\ref{th:asympt-normal} holds for any $\wh\bb$ satisfying
Condition~\ref{con:a7}, but $\eta_3$ drops out of
\eqref{eq:prescco-ee} only when $\wh\bb$ is the SPARCC estimator. The
efficiency in Theorem~\ref{th:asympt-normal} is among regular
asymptotically linear estimators of $\zeta$
\citep{tsiatis2006semiparametric}.

\begin{proof}
By Conditions \ref{con:a1}, \ref{con:a5}, and \ref{con:a6}, we
can apply \cite{newey1994large}, Lemma 2.4 and an argument similar to that used to prove \eqref{eq:e13} to obtain
\be
\sup_{\theta \in \Omega}|n^{-1}\sumi\frac{\partial\phi\eff^{*\star} (\bO_i,\theta,\wh\bb)}{\partial\theta}- E\{\frac{\partial\phi\eff^{*\star} (\bO,\theta,\bb)}{\partial\theta}\}|=o_p(1).\label{eq:e14}
\ee

Next,
\bse
0 &=& n^{-1/2}\sumi \phi\eff^{*\star}(\bO_i,\wh\zeta,\wh\bb)\\
&=&
n^{-1/2}\sumi\frac{\partial\phi\eff^{*\star}(\bO_i,\zeta,\bb)}{\partial\zeta}(\wh\zeta-\zeta)
+n^{-1/2}\sumi\frac{\partial\phi\eff^{*\star}(\bO_i,\zeta,\bb)}{\partial\bb\trans}(\wh\bb-\bb)\\
&&+ n^{1/2} o_p(\|(\wh\zeta,\wh\bb\trans)\trans-(\zeta,\bb\trans)\trans\|_2)+ n^{-1/2}\sumi \phi\eff^{*\star}(\bO_i,\zeta,\bb)\\
&=&
n^{-1/2}\sumi\frac{\partial\phi\eff^{*\star}(\bO_i,\zeta,\bb)}{\partial\zeta}(\wh\zeta-\zeta)
+n^{-1/2}\sumi\frac{\partial\phi\eff^{*\star}(\bO_i,\zeta,\bb)}{\partial\bb\trans}(\wh\bb-\bb)\\
&&+ n^{-1/2}\sumi \phi\eff^{*\star}(\bO_i,\zeta,\bb)
+n^{1/2} o_p(|\wh\zeta-\zeta|)+o_p(1)\\
&=& n^{1/2}(\wh\zeta-\zeta)\{\tau + o_p(1)\} 
 + \{\bh+o_p(1)\}\trans \{n^{-1/2}\sumi \bxi(\bO_i,\bb)+o_p(1)\}\\
 &&+ n^{-1/2}\sumi \phi\eff^{*\star}(\bO_i,\zeta,\bb) + o_p(n^{1/2}|\wh\zeta-\zeta|)+o_p(1)\\
 &=& n^{1/2}(\wh\zeta-\zeta)\tau
 + n^{-1/2}\sumi \{\bh\trans\bxi(\bO_i,\bb)+\phi\eff^{*\star}(\bO_i,\zeta,\bb)\}+ o_p(n^{1/2}|\wh\zeta-\zeta|)+o_p(1),
\ese
where $\bh= E\{\partial\phi\eff^{*\star}(\bO,\zeta,\bb)/\partial\bb\}$.
Here, the second equality holds by Taylor's
theorem,
the third equality holds by the $n^{1/2}$-consistency of $\wh\bb$,
the fourth equality holds by \eqref{eq:e14} and the law of large
numbers, and
the last equality holds by the
  fact that the second term is $O_p(1)$, hence
  $n^{1/2}(\wh\zeta-\zeta)=O_p(1)$.
Thus, 
\bse
n^{1/2}(\wh\zeta-\zeta) &=& -\tau^{-1} n^{-1/2}\sumi \{\bh\trans\bxi(\bO_i,\bb)+\phi\eff^{*\star}(\bO_i,\zeta,\bb)\}+o_p(1)\\
&\stackrel{d}{\to}& \Normal(0,\tau^{-2}\sigma^2),
\ese
where Slutsky's theorem and the central limit theorem are used to prove
asymptotic normality.

In particular, assume that $\eta_1^*=\eta_1$ and $\eta_2^\star=\eta_2$. 
Since
\bse
\tau &=& E\{\frac{\partial\phi\eff(\bO,\zeta,\bb)}{\partial\zeta}\}\\
&=& E\left[\frac{\partial\bfe\trans}{\partial\zeta}\wt\S_\bb(\bO) + \Delta \frac{\partial a_1(W,\bZ)}{\partial\zeta} + (1-\Delta) \frac{E\{I(X> W)
    \partial a_1(X,\bZ)/\partial\zeta\mid W,Y,\Z,\bb,\eta_1\}}{E\{I(X > W)\mid W,Y,\Z,\bb,\eta_1\}}\right.\\
   &&\left.+\Delta \frac{E\{I(C \ge W) \partial a_2
(C,\Z)/\partial\zeta\mid W,\Z,\eta_2\}}{E\{I(C \ge W)\mid W,\Z,\eta_2\}} + (1-\Delta)\frac{\partial a_2(W,\bZ)}{\partial\zeta} + \frac{\partial a_3(\bZ)}{\partial\zeta}\mid \bb,\eta_1,\eta_2\right]\\
&=& \frac{\partial\bfe\trans}{\partial\zeta}E\{\wt\S_\bb(\bO)\}
- E\left\{\frac{\partial}{\partial\zeta}\left(\frac{E[I\{r(\bO,\bb) \le \zeta\}\mid X,\Z,\bb,\eta_2]}{E[\d\{\zeta - r(\bO,\bb)\}\mid\bb,\eta_1,\eta_2]}\right.\right.\\
&&-\frac{E[I\{r(\bO,\bb) \le \zeta\}\mid \Z,\bb,\eta_1,\eta_2]}{E[\d\{\zeta - r(\bO,\bb)\}\mid\bb,\eta_1,\eta_2]}\\
&&\left.\left.+\frac{E[I\{r(\bO,\bb) \le \zeta\}\mid C,\Z,\bb,\eta_1] - (1-\alpha)}{E[\d\{\zeta - r(\bO,\bb)\}\mid\bb,\eta_1,\eta_2]}\right)\mid \bb,\eta_1,\eta_2\right\}\\
&=& - E\left(\frac{E[\d\{\zeta - r(\bO,\bb)\}\mid X,\Z,\bb,\eta_2] - E[\d\{\zeta - r(\bO,\bb)\}\mid \Z,\bb,\eta_1,\eta_2]}{E[\d\{\zeta - r(\bO,\bb)\}\mid\bb,\eta_1,\eta_2]}\right.\\
&&\left.+\frac{E[\d\{\zeta - r(\bO,\bb)\}\mid C,\Z,\bb,\eta_1]}{E[\d\{\zeta - r(\bO,\bb)\}\mid\bb,\eta_1,\eta_2]}\mid \bb,\eta_1,\eta_2\right)\\
&=&-\frac{E[\d\{\zeta - r(\bO,\bb)\}\mid \bb,\eta_1,\eta_2]}{E[\d\{\zeta - r(\bO,\bb)\}\mid\bb,\eta_1,\eta_2]},
\ese
we have that $\tau=-1$. In addition,
\bse
\bh&=&E\{\frac{\partial\phi\eff(\bO,\zeta,\bb)}{\partial\bb}\}\\
&=&E\{\frac{\dd\phi\eff(\bO,\zeta,\bb)}{\dd\bb}\}-E\{\frac{\partial\phi\eff(\bO,\zeta,\bb)}{\partial\zeta}\}\frac{\partial\zeta}{\partial\bb}\\
&=&-E\{\phi\eff(\bO,\zeta,\bb)\S_\bb(\bO)\} + \frac{\partial\zeta}{\partial\bb},
\ese
where the last line holds since $E\{\phi\eff(\bO,\zeta,\bb)\mid\bb\} = 0$ and $\tau=-1$.
Note that $\phi_j \in \Lambda_j$ for $j=1,2,3$, while
$\wt\S_\bb\in\wt\Lambda_\bb$ and $\S_\bb-\wt\S_\bb\in\Lambda_1$. 
Since $\wt\Lambda_\bb$ and $\Lambda_j$ for $j=1,2,3$ are orthogonal to
each other, we get that, for $j=2,3$,
$E\{\phi_j(\bO,\zeta,\bb)\S_\bb(\bO)\}=\0$, and for $j=1$,  
\bse
&&E\{\phi_1(\bO,\zeta,\bb)\S_\bb(\bO)\}\\
&=& E[\phi_1(\bO,\zeta,\bb)\{\S_\bb(\bO)-\wt\S_\bb(\bO)\}]\\
&=&E\left[\Delta a_1(W,\Z)\ba (W,\bZ)\right.\\
&&\left.+ (1-\Delta)\frac{E\{I(X>W)a_1(X,\bZ)
      \mid W,Y,\Z\}}{E\{I(X>W)\mid W,Y,\Z\}}\frac{E\{I(X>W)\ba(X,\bZ)
      \mid W,Y,\Z\}}{E\{I(X>W)\mid W,Y,\Z\}}\right]\\
&=&E\left(\ba (X,\bZ)\left[\Delta a_1(W,\Z)+ (1-\Delta)\frac{E\{I(X>W)a_1(X,\bZ)
      \mid W,Y,\Z\}}{E\{I(X>W)\mid W,Y,\Z\}}\right]\right)\\
&=&-E\left(\ba (X,\bZ)\left[\frac{E[I\{r(\bO,\bb) \le \zeta\}\mid X,\Z] - E[I\{r(\bO,\bb) \le \zeta\}\mid \Z]}{E[\d\{\zeta - r(\bO,\bb)\}]}\right]\right)\\
&=& -\frac{E\{E(\ba (X,\bZ)[I\{r(\bO,\bb) \le \zeta\}]\mid X,\Z)\}}{E[\d\{\zeta - r(\bO,\bb)\}]}\\
&& +\frac{E(E\{\ba (X,\bZ)\mid \Z\}[I\{r(\bO,\bb) \le \zeta\}])}{E[\d\{\zeta - r(\bO,\bb)\}]}\\
&=& -\frac{E(\ba (X,\bZ)[I\{r(\bO,\bb) \le \zeta\}])}{E[\d\{\zeta - r(\bO,\bb)\}]},
\ese
where the last equality holds since $E\{\ba (X,\bZ)\mid \Z\}=\0$. 
Moreover, 
\bse
&&E\{\phi_\bb(\bO,\zeta,\bb)\S_\bb(\bO)\}\\
&=& E\{\phi_\bb(\bO,\zeta,\bb)\wt\S_\bb(\bO)\}\\
&=& E\{\wt\S_\bb(\bO)^{\otimes2}\bfe\} \\
&=& \frac{E[\d\{\zeta - r(\bO,\bb)\}\partial r(\bO,\bb)/\partial \bb]-E[I\{r(\bO,\bb) \le \zeta\}\{\S_\bb(\bO)-\ba(X,\bZ)\}]}{E[\d\{\zeta - r(\bO,\bb)\}]}.
\ese 
Thus, we get that
\bse
\bh &=& -E\{\phi_1(\bO,\zeta,\bb)\S_\bb(\bO)\}-E\{\phi_\bb(\bO,\zeta,\bb)\S_\bb(\bO)\}+ \frac{\partial\zeta}{\partial\bb}\\
&=& \frac{E[\d\{\zeta - r(\bO,\bb)\}\partial \{\zeta - r(\bO,\bb)\}/\partial \bb]+E[I\{r(\bO,\bb) \le \zeta\}\S_\bb(\bO)]}{E[\d\{\zeta - r(\bO,\bb)\}]}.
\ese 
By taking the derivative of $E[I\{r(\bO,\bb) \le \zeta\}\mid\bb]=1-\alpha$ with respect to $\bb$, we have that
\bse
E\left[\d\{\zeta - r(\bO,\bb)\} \frac{\partial\{\zeta-r(\bO,\bb)\}}{\partial \bb}\right]+E[I\{r(\bO,\bb) \le \zeta\}\S_\bb(\bO)] = \0,
\ese
which implies that  $\bh= \0$. Lastly, substituting $\tau=-1$ and
$\bh=\0$, we get $\sqrt n(\wh\zeta - \zeta)
\stackrel{d}{\to}\Normal[0,\var\{\phi\eff(\bO,\zeta,\bb)\}]$. 
\end{proof}

\subsection{Proof of Theorem \ref{th:bias}}
\label{sec:bias}

Theorem~\ref{th:bias} assumes the following conditions on the system
$\sumi\bPhi(\bO_i,\zeta,\bb)=\0$. Conditions~\ref{con:c1} and
\ref{con:c2} hold when the coverage rate is a smooth function of
$(\zeta,\bb\trans)\trans$. Condition~\ref{con:c3} holds when
$\eta_1^*$ and $\eta_2^\star$ are close to $\eta_1$ and $\eta_2$,
respectively, since
$-E\{\partial\bPhi(\bO,\zeta,\bb)/\partial(\zeta,\bb\trans)\}$ is block diagonal with diagonal blocks $1$ and the efficient information
matrix for $\bb$ when both are correctly specified.
Conditions~\ref{con:c4} and \ref{con:c5} control the moments of
$\bPhi$ and its derivatives, respectively.
\begin{enumerate}[label=(C\arabic*),ref=(C\arabic*),start=1]
    \item\label{con:c1}
    The $\ell_2$ norm of the derivative of $\pr\{r(\bO_0,\bb) \le
    \zeta\mid \zeta,\bb\}$ satisfies
    \bse
    M_1 \equiv \left\|\frac{\partial \pr\{r(\bO_0,\bb) \le \zeta\mid
    \zeta,\bb\}}{\partial(\zeta,\bb\trans)\trans}\right\|_2<\infty.
    \ese
    \item\label{con:c2}
    The spectral norm of the Hessian matrix of $\pr\{r(\bO_0,\brho)
    \le \theta\mid \theta,\brho\}$ is uniformly bounded over
    $(\theta,\brho\trans)\trans$, that is,
    \bse
    M_2 \equiv \sup_{\theta,\brho}\left\|\frac{\partial^2\pr\{r(\bO_0,\brho)
    \le \theta\mid \theta,\brho\}}{\partial(\theta,\brho\trans)\trans
    \partial(\theta,\brho\trans)}\right\|_2<\infty.
    \ese
    \item\label{con:c3}
    The matrix $E\{\partial\bPhi(\bO,\zeta,\bb)/\partial(\zeta,\bb\trans)\}$
    is invertible, and its smallest singular value $\lambda_1$ is
    positive.
    \item\label{con:c4}
    The vector $\bPhi(\bo,\zeta,\bb)$ and the matrix
    $\partial\bPhi(\bo,\zeta,\bb)/\partial(\zeta,\bb\trans)$ have
    finite second moments, that is,
    \bse
    M_3\equiv\|E\{\bPhi^{\otimes2}(\bO,\zeta,\bb)\}\|_2<\infty, \qquad
    M_4\equiv\left\|E\left[\left\{\frac{\partial\bPhi(\bO,\zeta,\bb)}
    {\partial(\zeta,\bb\trans)}\right\}^{\otimes2}\right]\right\|_2<\infty.
    \ese
    \item\label{con:c5}
    For any fixed $\bv$ and $\bo$, the three-dimensional tensor
    $\partial^2\bPhi(\bo,\theta,\brho)/\{\partial(\theta,\brho\trans)\trans
    \partial(\theta,\brho\trans)\}$ is continuous with respect to
    $(\theta,\brho\trans)\trans$, and
    \bse
    M_5\equiv E\left[\sup_{\|\bv\|_2=1}\sup_{\theta,\brho}\left\|\left\{
    \bv\trans\frac{\partial^2\bPhi_j(\bO,\theta,\brho)}
    {\partial(\theta,\brho\trans)\trans\partial(\theta,\brho\trans)}
    \right\}_{j=1}^{d_\bb+1}\right\|_2^2\right]<\infty,
    \ese
    where $(\a_j\trans)_{j=1}^d\equiv(\a_1,\dots, \a_d)\trans$ for
    vectors $\a_j, j=1, \dots, d$, and $\bPhi_j$ is the $j$-th
    component of $\bPhi$.
\end{enumerate}
Theorems~\ref{th:bias} and \ref{th:finite-bound} hold more generally
for any $\wh\bb$ solving $\sumi\bpsi(\bO_i,\bb,\eta_1^*,\eta_2^\star)=\0$
with $E\{\bpsi(\bO,\bb,\eta_1^*,\eta_2^\star)\}=\0$, with
$\wt\S_\bb^{*\star}$ in $\bPhi$ replaced by $\bpsi$. Under standard
regularity conditions, such a $\wh\bb$ satisfies
Condition~\ref{con:a7} with influence function
$\bxi(\bo,\bb)=-E\{\partial\bpsi(\bO,\bb,\eta_1^*,\eta_2^\star)/
\partial\bb\trans\}^{-1}\bpsi(\bo,\bb,\eta_1^*,\eta_2^\star)$.

\begin{proof}
   For notational brevity, we define
\bse
\wh\bA\equiv
n^{-1}\sumi\frac{\partial\bPhi(\bO_i,\zeta,\bb)}{\partial(\zeta,\bb\trans)},\,
\bA\equiv
E\frac{\partial\bPhi(\bO_i,\zeta,\bb)}{\partial(\zeta,\bb\trans)},
\mbox { and }
\wh\bB_j(\theta,\brho)\equiv n^{-1}\sumi\frac{\partial^2\bPhi_j(\bO_i,\theta,\brho)}{\partial(\theta,\brho\trans)\trans\partial(\theta,\brho\trans)}.
\ese
By Taylor's theorem, we have that 
\bse
&&|E[I\{r(\bO_0,\wh\bb) \le \wh\zeta\}]-(1-\alpha)|\\
&=&|E(E[I\{r(\bO_0,\wh\bb) \le \wh\zeta\}\mid\wh\zeta,\wh\bb]-(1-\alpha))|\\
&=&
\left|E\left(\frac{\partial E[I\{r(\bO_0,\bb) \le \zeta\mid\zeta,\bb\}]}{\partial(\zeta,\bb\trans)}\{(\wh\zeta,\wh\bb\trans)\trans-(\zeta,\bb\trans)\trans\}\right)\right.\\
&&\left.
+\frac{1}{2}E\left(\{(\wh\zeta,\wh\bb\trans)-(\zeta,\bb\trans)\}\frac{\partial^2E[I\{r(\bO_0,\wt\bb) \le \wt\zeta\mid \wt\zeta,\wt\bb\}]}{\partial(\zeta,\bb\trans)\trans\partial(\zeta,\bb\trans)}\{(\wh\zeta,\wh\bb\trans)\trans-(\zeta,\bb\trans)\trans\}\right)\right|\\
&\le& M_1\|E\{(\wh\zeta,\wh\bb\trans)\trans-(\zeta,\bb\trans)\trans\}\|_2+\frac{M_2}{2}E\{\|(\wh\zeta,\wh\bb\trans)\trans-(\zeta,\bb\trans)\trans\|_2^2\},
\ese
where $(\wt\zeta,\wt\bb\trans)\trans$
 is on the line connecting $(\wh\zeta,\wh\bb\trans)\trans$ and
  $(\zeta,\bb\trans)\trans$.
First, by the $n^{1/2}$-consistency of
$(\wh\zeta,\wh\bb\trans)\trans$, we have
$E\{\|(\wh\zeta,\wh\bb\trans)\trans-(\zeta,\bb\trans)\trans\|_2^2\}=O(n^{-1})$.
Then it suffices to prove
$\|E\{(\wh\zeta,\wh\bb\trans)\trans-(\zeta,\bb\trans)\trans\}\|_2=O(n^{-1})$. 
Now, by Taylor's theorem, 
\bse
\0&=&\sumi\bPhi(\bO_i,\wh\zeta,\wh\bb)\\
&=&\sumi\bPhi(\bO_i,\zeta,\bb)+\sumi\frac{\partial\bPhi(\bO_i,\zeta,\bb)}{\partial(\zeta,\bb\trans)}\{(\wh\zeta,\wh\bb\trans)\trans-(\zeta,\bb\trans)\trans\}\\
&&+\left[
  \{(\wh\zeta,\wh\bb\trans)-(\zeta,\bb\trans)\}\frac{1}{2}\sumi\frac{\partial^2\bPhi_j(\bO_i,\wc\zeta,\wc\bb)}{\partial(\zeta,\bb\trans)\trans\partial(\zeta,\bb\trans)}\{(\wh\zeta,\wh\bb\trans)\trans-(\zeta,\bb\trans)\trans\}\right]_{j=1}^{d_\bb+1}\\
&=&\sumi\bPhi(\bO_i,\zeta,\bb)+n\wh\bA\{(\wh\zeta,\wh\bb\trans)\trans-(\zeta,\bb\trans)\trans\}\\
&&+\left[
  \{(\wh\zeta,\wh\bb\trans)-(\zeta,\bb\trans)\}\frac{n}{2}\wh\bB_j(\wc\zeta,\wc\bb)\{(\wh\zeta,\wh\bb\trans)\trans-(\zeta,\bb\trans)\trans\}\right]_{j=1}^{d_\bb+1},
\ese
where $(\wc\zeta,\wc\bb\trans)\trans$
 is on the line connecting $(\wh\zeta,\wh\bb\trans)\trans$ and
  $(\zeta,\bb\trans)\trans$. 
Then
\be\label{eq:e15}
&&\|E\{(\wh\zeta,\wh\bb\trans)\trans-(\zeta,\bb\trans)\trans\}\|_2\n\\
&=&
\left\|E\left(\wh\bA^{-1} n^{-1}\sumi\bPhi(\bO_i,\zeta,\bb)+
\wh\bA^{-1} 
\frac{1}{2}\{(\wh\zeta,\wh\bb\trans)-(\zeta,\bb\trans)\}
\wh\bB_j(\wc\zeta,\wc\bb)\{(\wh\zeta,\wh\bb\trans)\trans-(\zeta,\bb\trans)\trans\}
\right)\right\|_2\n\\
&=&
\left\|E\{
      (\wh\bA^{-1}-\bA^{-1})
    n^{-1}\sumi\bPhi(\bO_i,\zeta,\bb)\}\right.\n\\
&&+\left.E\left(
        \wh\bA^{-1}\frac{1}{2}\left[\{(\wh\zeta,\wh\bb\trans)-(\zeta,\bb\trans)\}
\wh\bB_j(\wc\zeta,\wc\bb)\{(\wh\zeta,\wh\bb\trans)\trans-(\zeta,\bb\trans)\trans\}\right]_{j=1}^{d_\bb+1}
\right)\right\|_2\n\\
&\le&\|E\{
      (\wh\bA^{-1}-\bA^{-1})
    n^{-1}\sumi\bPhi(\bO_i,\zeta,\bb)\}\|_2\n\\
    &&+\left\|E\left(
        \wh\bA^{-1}\frac{1}{2}\left[\{(\wh\zeta,\wh\bb\trans)-(\zeta,\bb\trans)\}
\wh\bB_j(\wc\zeta,\wc\bb)\{(\wh\zeta,\wh\bb\trans)\trans-(\zeta,\bb\trans)\trans\}\right]_{j=1}^{d_\bb+1}
\right)\right\|_2\n\\
&\le& E\{\|
   (\wh\bA^{-1}-\bA^{-1})
    n^{-1}\sumi\bPhi(\bO_i,\zeta,\bb)\|_2\}\n\\
    &&+E\left(\left\|
        \wh\bA^{-1}\frac{1}{2}\left[\{(\wh\zeta,\wh\bb\trans)-(\zeta,\bb\trans)\}
\wh\bB_j(\wc\zeta,\wc\bb)\{(\wh\zeta,\wh\bb\trans)\trans-(\zeta,\bb\trans)\trans\}\right]_{j=1}^{d_\bb+1}\right\|_2
\right)\n\\
&\le&
E\left[\|\bA^{-1}\|_2\|\wh\bA-\bA\|_2
  \|\wh\bA^{-1}\|_2\|n^{-1}\sumi\bPhi(\bO_i,\zeta,\bb)\|_2\right]\n\\
&&+\frac{1}{2}E\left(\|\wh\bA^{-1}\|_2\left\|\left[\{(\wh\zeta,\wh\bb\trans)-(\zeta,\bb\trans)\}\wh\bB_j(\wc\zeta,\wc\bb)\{(\wh\zeta,\wh\bb\trans)\trans-(\zeta,\bb\trans)\trans\}\right]_{j=1}^{d_\bb+1}\right\|_2
\right),
\ee
where the
first inequality holds by the triangle inequality, the
second inequality holds by Jensen's inequality using the fact that
$\|\cdot\|_2$ is convex, and the last inequality holds by the
definition of spectral norm. 
We now give upper bounds for each component. First, the continuous mapping
theorem and Condition \ref{con:c3} yield 
$
\|\wh\bA^{-1}\|_2\to\|\bA^{-1}\|_2 = \lambda_1^{-1} 
$
almost surely. By Condition \ref{con:c4}, we get
$\|n^{-1}\sumi\bPhi(\bO_i,\zeta,\bb)\|_2 = O_p(n^{-1/2})$,
and $\|\wh\bA-\bA\|_2 = O_p(n^{-1/2})$. Hence the first term in
\eqref{eq:e15} is of order $O_p(n^{-1})$.
Lastly, Condition \ref{con:c5} leads to
\bse
&&\left\|\left[\{(\wh\zeta,\wh\bb\trans)-(\zeta,\bb\trans)\}\wh\bB_j(\wc\zeta,\wc\bb)\{(\wh\zeta,\wh\bb\trans)\trans-(\zeta,\bb\trans)\trans\}\right]_{j=1}^{d_\bb+1}\right\|_2\n\\
&\le& \{M_5^{1/2}+o_p(1)\} O\{\|(\wh\zeta,\wh\bb\trans)\trans-(\zeta,\bb\trans)\trans\|_2^2\}\n\\
&=& O_p(n^{-1}).
\ese
Inserting the upper bounds into \eqref{eq:e15}, we get
$\|E\{(\wh\zeta,\wh\bb\trans)\trans-(\zeta,\bb\trans)\trans\}\|_2 =
O(n^{-1})$, which completes the proof. 
\end{proof}

\subsection{Proof of Theorem \ref{th:finite-bound}}\label{sec:finite-bound}

Theorem~\ref{th:finite-bound} assumes
Conditions~\ref{con:c1}--\ref{con:c5} and two more. Bounding the
coverage rate, rather than the expected coverage rate, requires a
concentration inequality for $n^{-1}\sumi\bPhi(\bO_i,\zeta,\bb)$, and
the inequality we use holds when each component of $\bPhi$ has a
sub-exponential tail. A random variable $V$ is
$(\nu,\omega)$-sub-exponential when $E[e^{t\{V-E(V)\}}] \le
e^{\nu^2t^2/2}$ for all $|t|<1/\omega$ \citep{wainwright2019high}.
Sub-Gaussian tails would be too restrictive in this setting, and
weaker moment conditions would not give the inequality.
Condition~\ref{con:c7} gives each component of $\bPhi$ a
sub-exponential tail, and Condition~\ref{con:c6} places
$(\wh\zeta,\wh\bb\trans)\trans$ in a small neighborhood of
$(\zeta,\bb\trans)\trans$, so that Taylor expansions and derivative
bounds apply uniformly there.
\begin{enumerate}[label=(C\arabic*),ref=(C\arabic*),start=6]
    \item\label{con:c6}
    $\|(\wh\zeta,\wh\bb\trans)\trans-(\zeta,\bb\trans)\trans\|_2 < c_0$
    for a fixed small $c_0>0$ such that $2c_0M_5^{1/2} < \lambda_1$.
    \item\label{con:c7}
    For $j=1, \ldots,d_\bb+1$, $\bPhi_j(\bO,\zeta,\bb)$ is
    $(\nu_j,\omega_j)$-sub-exponential.
\end{enumerate}
Section~\ref{sec:sub-exp} verifies Condition~\ref{con:c7} for a linear
outcome model with normal errors when no patient's time to the event
is right-censored, where $\phi\eff$ takes the explicit form derived for
Remark~\ref{rem:no-censoring}. Both terms of the bound in
Theorem~\ref{th:finite-bound} shrink as $\lambda_1$ grows. A larger
$\lambda_1$ means that the estimating equations change faster near
$(\zeta,\bb\trans)\trans$, and the same sampling error in the
equations then produces a smaller error in
$(\wh\zeta,\wh\bb\trans)\trans$.

\begin{proof}
By Taylor's theorem,
\bse
&& |E[I\{r(\bO_0,\wh\bb) \le \wh\zeta\}\mid \wh\zeta,\wh\bb] - (1-\alpha)|\\
&=& |E[I\{r(\bO_0,\wh\bb) \le \wh\zeta\}\mid \wh\zeta,\wh\bb] - E[I\{r(\bO_0,\bb) \le \zeta\}\mid \zeta,\bb]|\\
&=& \left|\frac{\partial E[I\{r(\bO_0,\bb) \le \zeta\}\mid \zeta,\bb]}{\partial(\zeta,\bb\trans)\trans} \{(\wh\zeta,\wh\bb\trans)\trans-(\zeta,\bb\trans)\trans\} \right.\\
&&\left.+ \frac{1}{2}
\{(\wh\zeta,\wh\bb\trans)-(\zeta,\bb\trans)\}\frac{\partial E[I\{r(\bO_0,\wt\bb) \le \wt\zeta\}\mid \wt\zeta,\wt\bb]}{\partial(\zeta,\bb\trans)\partial(\zeta,\bb\trans)\trans} \{(\wh\zeta,\wh\bb\trans)\trans-(\zeta,\bb\trans)\trans\}\right|\\
&\le& M_1 \|(\wh\zeta,\wh\bb\trans)\trans-(\zeta,\bb\trans)\trans\|_2 + \frac{M_2}{2} \|(\wh\zeta,\wh\bb\trans)\trans-(\zeta,\bb\trans)\trans\|_2^2,
\ese
where the second equality holds by Taylor's theorem,
$(\wt\zeta,\wt\bb\trans)\trans$ is on the line connecting
  $(\wh\zeta,\wh\bb\trans)\trans$ and $(\zeta,\bb\trans)\trans$,
and the inequality holds by Conditions \ref{con:c1} and \ref{con:c2}.

Recall the definitions
\bse
\wh\bA\equiv
n^{-1}\sumi\frac{\partial\bPhi(\bO_i,\zeta,\bb)}{\partial(\zeta,\bb\trans)},\,
\bA\equiv
E\frac{\partial\bPhi(\bO_i,\zeta,\bb)}{\partial(\zeta,\bb\trans)},
\mbox { and }
\wh\bB_j(\theta,\brho)\equiv n^{-1}\sumi\frac{\partial^2\bPhi_j(\bO_i,\theta,\brho)}{\partial(\theta,\brho\trans)\trans\partial(\theta,\brho\trans)}
\ese
in the proof of Theorem \ref{th:bias}.
Taylor's theorem leads to
\bse
\0&=&\sumi\bPhi(\bO_i,\wh\zeta,\wh\bb)\\
&=&\sumi\bPhi(\bO_i,\zeta,\bb)+\sumi\frac{\partial\bPhi(\bO_i,\zeta,\bb)}{\partial(\zeta,\bb\trans)}\{(\wh\zeta,\wh\bb\trans)\trans-(\zeta,\bb\trans)\trans\}\\
&&+\left[
  \{(\wh\zeta,\wh\bb\trans)-(\zeta,\bb\trans)\}\frac{1}{2}\sumi\frac{\partial^2\bPhi_j(\bO_i,\wc\zeta,\wc\bb)}{\partial(\zeta,\bb\trans)\trans\partial(\zeta,\bb\trans)}\{(\wh\zeta,\wh\bb\trans)\trans-(\zeta,\bb\trans)\trans\}\right]_{j=1}^{d_\bb+1}\\
&=&\sumi\bPhi(\bO_i,\zeta,\bb)+n\wh\bA\{(\wh\zeta,\wh\bb\trans)\trans-(\zeta,\bb\trans)\trans\}\\
&&+\left[
  \{(\wh\zeta,\wh\bb\trans)-(\zeta,\bb\trans)\}\frac{n}{2}\wh\bB_j(\wc\zeta,\wc\bb)\{(\wh\zeta,\wh\bb\trans)\trans-(\zeta,\bb\trans)\trans\}\right]_{j=1}^{d_\bb+1},
\ese
where $(\wc\zeta,\wc\bb\trans)\trans$
 is on the line connecting $(\wh\zeta,\wh\bb\trans)\trans$ and
  $(\zeta,\bb\trans)\trans$. 
Meanwhile, if $V_1,...,V_n$ are independent and identically distributed $(\nu,\omega)$-sub-exponential random variables, then
 for all $|t|<1/\omega$,
    \bse
    E[e^{t\{\sumi V_i-E(\sumi V_i)\}}] = \prod_{i=1}^n E[e^{t\{ V_i-E(V_i)\}}] \le e^{n\nu^2t^2/2}.
    \ese
That is,
$\sumi V_i$ is $(n^{1/2}\nu,\omega)$-sub-exponential. 
Moreover, \cite{wainwright2019high}, Proposition 2.9 states that if
$V$ is $(\nu,\omega)$-sub-exponential, then
\bse
P\{|V-E(V)|>t\}\le 2\exp\{-\min(\frac{t^2}{2\nu^2}, \frac{t}{2\omega})\}.
\ese 
Combining the results above, we get that
\bse
&&P\{\|n^{-1}\sumi\bPhi(\bO_i,\zeta,\bb)\|_2 > n^{-1/2}(\log n)^{1/2} \kappa \}\\
&\le& \sum_{j=1}^{d_\bb+1} P\{|n^{-1}\sumi\bPhi_j(\bO_i,\zeta,\bb)| > n^{-1/2}(\log n)^{1/2} \kappa (d_\bb+1)^{-1/2}\}\\
&\le& \sum_{j=1}^{d_\bb+1} 2\exp[-\min\{\frac{(\log n) \kappa ^2(d_\bb+1)^{-1}}{2\nu_j^2}, \frac{n^{1/2}(\log n)^{1/2} \kappa (d_\bb+1)^{-1/2}}{2\omega_j}\}]\\
&\le& \sum_{j=1}^{d_\bb+1} 2[n^{-\kappa ^2(d_\bb+1)^{-1}2^{-1}\nu_j^{-2}}+\exp\{- \frac{n^{1/2}(\log n)^{1/2} \kappa (d_\bb+1)^{-1/2}}{2\omega_j}\}].
\ese 

Moreover, by Chebyshev's inequality,
\bse
P(\|\wh\bA-\bA\|_2>\frac{\lambda_1}{4})&\le&\frac{4^2}{\lambda_1^2}E[\|n^{-1}\sumi\frac{\partial\bPhi(\bO_i,\zeta,\bb)}{\partial(\zeta,\bb\trans)}-
E\{\frac{\partial\bPhi(\bO_i,\zeta,\bb)}{\partial(\zeta,\bb\trans)}\}
\|_2^2]\\
&\le&\frac{4^2}{\lambda_1^2}E[\|n^{-1}\sumi\frac{\partial\bPhi(\bO_i,\zeta,\bb)}{\partial(\zeta,\bb\trans)}-
E\{\frac{\partial\bPhi(\bO_i,\zeta,\bb)}{\partial(\zeta,\bb\trans)}\}
\|_F^2]\\
&\le&\frac{4^2}{n\lambda_1^2}\trace(E[\{\frac{\partial\bPhi(\bO,\zeta,\bb)}{\partial(\zeta,\bb\trans)}\}^{\otimes2}]).
\ese
Consider $(\theta,\brho\trans)\trans$ satisfying $\|(\theta,\brho\trans)\trans-(\zeta,\bb\trans)\trans\|_2 < c_0$. Since $(\theta,\brho\trans)\trans$ lies in a compact set,
  by
  Chebyshev's inequality and Condition \ref{con:c5},
\be
&&P\left(n^{-1} \sumi\sup_{\|\bv\|_2=1}\sup_{\theta,\brho}\left\|\left\{\bv\trans \frac{\partial^2\bPhi_j(\bO_i,\theta,\brho)}{\partial(\theta,\brho\trans)\trans\partial(\theta,\brho\trans)}\right\}_{j=1}^{d_\bb+1}\right\|_2\right.\n\\
&&\qquad\left.-E\left[\sup_{\|\bv\|_2=1}\sup_{\theta,\brho}\left\|\left\{\bv\trans \frac{\partial^2\bPhi_j(\bO,\theta,\brho)}{\partial(\theta,\brho\trans)\trans\partial(\theta,\brho\trans)}\right\}_{j=1}^{d_\bb+1}\right\|_2\right]
>\frac{\lambda_1}{2c_0}-M_5^{1/2}\right)\n\\
&\le& (\frac{\lambda_1}{2c_0}-M_5^{1/2})^{-2}n^{-1}\var\left[\sup_{\|\bv\|_2=1}\sup_{\theta,\brho}\left\|\bv\trans\left\{\frac{\partial^2\bPhi_j(\bO,\theta,\brho)}{\partial(\theta,\brho\trans)\trans\partial(\theta,\brho\trans)}\right\}_{j=1}^{d_\bb+1}\right\|_2\right]\n\\
&\le& (\frac{\lambda_1}{2c_0}-M_5^{1/2})^{-2}\frac{M_5}{n},\label{eq:e17}
\ee
where $M_5$ is defined in Condition \ref{con:c5}. 
Note that by Jensen's inequality,
\be
E\left[\sup_{\|\bv\|_2=1}\sup_{\theta,\brho}\left\|\left\{\bv\trans \frac{\partial^2\bPhi_j(\bO,\theta,\brho)}{\partial(\theta,\brho\trans)\trans\partial(\theta,\brho\trans)}\right\}_{j=1}^{d_\bb+1}\right\|_2\right] \le M_5^{1/2}. \label{eq:e18}
\ee
Also, by the triangle inequality,
\be
&&\sup_{\|\bv\|_2=1}\left\|\left\{\bv\trans\wh\bB_j(\wc\zeta,\wc\bb)\right\}_{j=1}^{d_\bb+1}\right\|_2\n\\
&\le& \sup_{\|\bv\|_2=1}\sup_{\theta,\brho}\left\|\left\{\bv\trans\wh\bB_j(\theta,\brho)\right\}_{j=1}^{d_\bb+1}\right\|_2\n\\
&=&\sup_{\|\bv\|_2=1}\sup_{\theta,\brho}\left\|\left\{\bv\trans n^{-1} \sumi\frac{\partial^2\bPhi_j(\bO_i,\theta,\brho)}{\partial(\theta,\brho\trans)\trans\partial(\theta,\brho\trans)}\right\}_{j=1}^{d_\bb+1}\right\|_2\n\\
&\le&n^{-1} \sup_{\|\bv\|_2=1}\sup_{\theta,\brho}\sumi\left\|\left\{\bv\trans \frac{\partial^2\bPhi_j(\bO_i,\theta,\brho)}{\partial(\theta,\brho\trans)\trans\partial(\theta,\brho\trans)}\right\}_{j=1}^{d_\bb+1}\right\|_2\n\\
&\le&n^{-1} \sumi\sup_{\|\bv\|_2=1}\sup_{\theta,\brho}\left\|\left\{\bv\trans \frac{\partial^2\bPhi_j(\bO_i,\theta,\brho)}{\partial(\theta,\brho\trans)\trans\partial(\theta,\brho\trans)}\right\}_{j=1}^{d_\bb+1}\right\|_2.\label{eq:e19}
\ee
Denote $M_6\equiv
\{\lambda_1/(2c_0)-M_5^{1/2}\}^{-2}M_5$.
Combining \eqref{eq:e18} and \eqref{eq:e19} into \eqref{eq:e17}, we get that 
\bse
P\left[n^{-1} \sumi\sup_{\|\bv\|_2=1}\left\|\left\{\bv\trans\wh\bB_j(\wc\zeta,\wc\bb)\right\}_{j=1}^{d_\bb+1}\right\|_2
>\frac{\lambda_1}{2c_0}\right]\le M_6 n^{-1}.
\ese
 Then, with probability at least $1-M_6 n^{-1}$, 
\bse
\left\|\left[
  \{(\wh\zeta,\wh\bb\trans)-(\zeta,\bb\trans)\}\wh\bB_j(\wc\zeta,\wc\bb)\right]_{j=1}^{d_\bb+1}\right\|_2
&\le&\|(\wh\zeta,\wh\bb\trans)-(\zeta,\bb\trans)\|_2\sup_{\|\bv\|_2=1}\left\|\left\{\bv\trans\wh\bB_j(\wc\zeta,\wc\bb)\right\}_{j=1}^{d_\bb+1}\right\|_2\\
  &\le&\frac{\lambda_1}{2}.
\ese

Thus, with probability at least $1-4^2\trace(E[\{\partial\bPhi(\bO,\zeta,\bb)/\partial(\zeta,\bb\trans)\}^{\otimes2}])\lambda_1^{-2}n^{-1}-M_6n^{-1}$, we have
\bse
&&\left\|\wh\bA+\frac{1}{2}\left[
  \{(\wh\zeta,\wh\bb\trans)-(\zeta,\bb\trans)\}\wh\bB_j(\wc\zeta,\wc\bb)\right]_{j=1}^{d_\bb+1}\right\|_2 \\
  &\ge&\|\bA\|_2-\|\wh\bA-\bA\|_2-\frac{1}{2}\left\|\left[
  \{(\wh\zeta,\wh\bb\trans)-(\zeta,\bb\trans)\}\wh\bB_j(\wc\zeta,\wc\bb)\right]_{j=1}^{d_\bb+1}\right\|_2 \\
  &\ge& \lambda_1-\frac{\lambda_1}{4}-\frac{\lambda_1}{4}\\
  &=&\frac{\lambda_1}{2}.
\ese
Finally, with probability at least 
\bse
&&1-4^2\trace(E[\{\frac{\partial\bPhi(\bO,\zeta,\bb)}{\partial(\zeta,\bb\trans)}\}^{\otimes2}])\lambda_1^{-2}n^{-1}-M_6n^{-1}\\
&&-\sum_{j=1}^{d_\bb+1} 2[n^{-\kappa ^2(d_\bb+1)^{-1}2^{-1}\nu_j^{-2}}+\exp\{- \frac{n^{1/2}(\log n)^{1/2} \kappa (d_\bb+1)^{-1/2}}{2\omega_j}\}],
\ese
we have that
\bse
\|(\wh\zeta,\wh\bb\trans)\trans-(\zeta,\bb\trans)\trans\|_2&=&\left\|\left(\wh\bA+\frac{1}{2}\left[
  \{(\wh\zeta,\wh\bb\trans)-(\zeta,\bb\trans)\}\wh\bB_j(\wc\zeta,\wc\bb)\right]_{j=1}^{d_\bb+1}\right)^{-1}\{n^{-1}\sumi\bPhi(\bO_i,\zeta,\bb)\}\right\|_2\\
  &\le&\left\|\left(\wh\bA+\left[
  \{(\wh\zeta,\wh\bb\trans)-(\zeta,\bb\trans)\}\frac{1}{2}\wh\bB_j(\wc\zeta,\wc\bb)\right]_{j=1}^{d_\bb+1}\right)^{-1}\right\|_2\|n^{-1}\sumi\bPhi(\bO_i,\zeta,\bb)\|_2\\
  &\le&\frac{2\kappa n^{-1/2}(\log n)^{1/2}}{\lambda_1},
\ese
which implies that
\bse
|E[I\{r(\bO_0,\wh\bb) \le \wh\zeta\}\mid \wh\zeta,\wh\bb] - (1-\alpha)| \le \frac{2M_1\kappa n^{-1/2} (\log n)^{1/2}}{\lambda_1}+ \frac{2M_2\kappa ^2n^{-1}\log n}{\lambda_1^2}.
\ese
Since $e^{-x} \le x^{-2}$ for $x>0$,
we can take $k=\min\{1,\kappa ^2(d_\bb+1)^{-1}2^{-1}\min_{j=1,...,d_\bb+1}(\nu_j^{-2})\}$ and
\bse
A&=&4^2\trace(E[\{\frac{\partial\bPhi(\bO,\zeta,\bb)}{\partial(\zeta,\bb\trans)}\}^{\otimes2}])\lambda_1^{-2}+M_6+\sum_{j=1}^{d_\bb+1}
[2+8\omega_j^2\kappa^{-2} (d_\bb+1)].
\ese
Then we obtain the desired result.
\end{proof}

%%%%%%%%%%%%%%%%%%%%%
\subsection{Verification of Condition~\ref{con:c7}}\label{sec:sub-exp}
%%%%%%%%%%%%%%%%%%%%%

We verify Condition~\ref{con:c7} when no patient's time to the event
is right-censored, so that $\phi\eff$ takes the explicit form
\eqref{eq:phieff-nocensoring} derived for
Remark~\ref{rem:no-censoring}. Suppose the outcome follows the linear
model $Y=(X,\Z)\trans\bb+\epsilon$ with $\epsilon\sim
\Normal(0,\sigma_\epsilon^2)$ and $(X,\Z)\trans\sim
\Normal(\bmu,\bSigma)$, and that
$\bpsi(\bO,\bb,\eta_1^*,\eta_2^\star)=\S_\bb^F(Y,X,\Z,\bb)$ is the
estimating function for $\wh\bb$.

The $j$-th component of $\S_\bb^F(Y,X,\Z,\bb)=(X,\Z)\trans\epsilon$ is
$\{2^{1/2}\sigma_\epsilon(\mu_j^2+\Sigma_{jj})^{1/2},
2^{1/2}\sigma_\epsilon\Sigma_{jj}^{1/2}\}$-sub-exponential, and
generally not sub-Gaussian. By \eqref{eq:phieff-nocensoring},
$\phi\eff(\bO,\zeta,\bb)$ is
$(\sigma_\epsilon[4\{\bfe\trans(\bmu\bmu\trans+\bSigma)\bfe\}+2^{-1}
d_0^{-2}]^{1/2}, 2^{3/2}\sigma_\epsilon
(\bfe\trans\bSigma\bfe)^{1/2})$-sub-exponential, where $d_0$ is the
density of $|V|$ at its $(1-\alpha)$ quantile for $V\sim
\Normal(0,1)$. Every component of $\bPhi(\bO,\zeta,\bb)$ therefore has
a sub-exponential tail, so Condition~\ref{con:c7} holds.

Right-censoring introduces structure this calculation does not
address, so the verification is a baseline for the right-censored case
rather than a proof of it.

\subsection{Proof of Proposition~\ref{pro:sigmatau}}\label{sec:sigmatau}
 \begin{proof}
For $\wh\bb$ chosen to be the solution of $\sumi\wt\S_\bb^{*\star}(\bO_i,\bb)=\0$, the influence function of $\wh\bb$ is
\bse
\bxi(\bo,\bb) = -E\{\frac{\partial\wt \S_\bb^{*\star}(\bO,\bb)}{\partial\bb\trans}\}^{-1} \wt\S_\bb^{*\star}(\bo,\bb),
\ese
which leads to
\bse
\sigma^2 &=& \var\{\bh\trans\bxi(\bO,\bb) +\phi\eff^{*\star}(\bO,\zeta,\bb)\} \\
&=& \var[E\{\frac{\partial\bfe^{*\star\rm T}\wt\S_\bb^{*\star}(\bO,\bb)}{\partial\bb} + \frac{\partial\sum_{j=1}^3\phi_j^{*\star}(\bO,\zeta,\bb)}{\partial\bb}\}\trans\bxi(\bO,\bb) \\
&&+\bfe^{*\star\rm T}\wt\S_\bb^{*\star}(y,w,\delta,\z)+\sum_{j=1}^3\phi_j^{*\star}(\bO,\zeta,\bb)]\\
&=& \var[\bfe^{*\star\rm T}E\{\frac{\partial\wt\S_\bb^{*\star}(\bO,\bb)}{\partial\bb\trans}\} \bxi(\bO,\bb) + E\{ \frac{\partial\sum_{j=1}^3\phi_j^{*\star}(\bO,\zeta,\bb)}{\partial\bb}\}\trans\bxi(\bO,\bb) \\
&&+\bfe^{*\star\rm T}\wt\S_\bb^{*\star}(y,w,\delta,\z)+\sum_{j=1}^3\phi_j^{*\star}(\bO,\zeta,\bb)]\\
&=& \var[-E\{ \frac{\partial\sum_{j=1}^3\phi_j^{*\star}(\bO,\zeta,\bb)}{\partial\bb}\}\trans E\{\frac{\partial\wt \S_\bb^{*\star}(\bO,\bb)}{\partial\bb\trans}\}^{-1} \wt\S_\bb^{*\star}(\bO,\bb)+\sum_{j=1}^3\phi_j^{*\star}(\bO,\zeta,\bb)].
\ese
Also, $\tau$ can be represented as
\bse
\tau &=& E\{\frac{\partial\phi\eff^{*\star}(\bO,\zeta,\bb)}{\partial\zeta}\}\\
&=& E\left[\frac{\partial\bfe^{*\star\rm T}}{\partial\zeta}\wt\S_\bb^{*\star}(\bO) + \Delta \frac{\partial a_1^{*\star}(W,\bZ)}{\partial\zeta} + (1-\Delta) \frac{E\{I(X> W)
    \partial a_1^{*\star}(X,\bZ)/\partial\zeta\mid W,Y,\Z,\bb,\eta_1^*\}}{E\{I(X > W)\mid W,Y,\Z,\bb,\eta_1^*\}}\right.\\
   &&\left.+\Delta \frac{E\{I(C \ge W) \partial a_2^{*\star}
(C,\Z)/\partial\zeta\mid W,\Z,\eta_2^\star\}}{E\{I(C \ge W)\mid W,\Z,\eta_2^\star\}} + (1-\Delta)\frac{\partial a_2^{*\star}(W,\bZ)}{\partial\zeta} + \frac{\partial a_3^{*\star}(\bZ)}{\partial\zeta}\mid \bb,\eta_1,\eta_2\right]\\
&=& \frac{\partial\bfe^{*\star\rm T}}{\partial\zeta}E\{\wt\S_\bb^{*\star}(\bO)\}
- E\left\{\frac{\partial}{\partial\zeta}\left(\frac{E[I\{r(\bO,\bb) \le \zeta\}\mid X,\Z,\bb,\eta_2^\star] - E[I\{r(\bO,\bb) \le \zeta\}\mid \Z,\bb,\eta_1^*,\eta_2^\star]}{E[\d\{\zeta - r(\bO,\bb)\}\mid\bb,\eta_1^*,\eta_2^\star]}\right.\right.\\
&&\left.\left.+\frac{E[I\{r(\bO,\bb) \le \zeta\}\mid C,\Z,\bb,\eta_1^*] - (1-\alpha)}{E[\d\{\zeta - r(\bO,\bb)\}\mid\bb,\eta_1^*,\eta_2^\star]}\right)\mid \bb,\eta_1,\eta_2\right\}\\
&=& - E\left(\frac{E[\d\{\zeta - r(\bO,\bb)\}\mid X,\Z,\bb,\eta_2^\star] - E[\d\{\zeta - r(\bO,\bb)\}\mid \Z,\bb,\eta_1^*,\eta_2^\star]}{E[\d\{\zeta - r(\bO,\bb)\}\mid\bb,\eta_1^*,\eta_2^\star]}\right.\\
&&\left.+\frac{E[\d\{\zeta - r(\bO,\bb)\}\mid C,\Z,\bb,\eta_1^*]}{E[\d\{\zeta - r(\bO,\bb)\}\mid\bb,\eta_1^*,\eta_2^\star]}\mid \bb,\eta_1,\eta_2\right)\\
&=&-\frac{E[\d\{\zeta - r(\bO,\bb)\}\mid \bb,\eta_1,\eta_2]}{E[\d\{\zeta - r(\bO,\bb)\}\mid\bb,\eta_1^*,\eta_2^\star]}.
\ese
Thus, $\sigma^2\tau^{-2}$, the asymptotic variance of $n^{1/2}(\wh\zeta - \zeta)$, can be rewritten as
\bse
\sigma^2\tau^{-2} = \frac{\var[-E\{\partial b^{*\star}(\bO,\zeta,\bb)/\partial\bb\trans\}E\{\partial\wt \S_\bb^{*\star}(\bO,\bb)/\partial\bb\trans\}^{-1} \wt\S_\bb^{*\star}(\bO,\bb)+b^{*\star}(\bO,\zeta,\bb)]}{E[\d\{\zeta - r(\bO,\bb)\}\mid \bb,\eta_1,\eta_2]^2}.
\ese

 \end{proof}

\section{Additional numerical results}
\subsection{Additional simulation results}
\label{sec:addsim}

Figures \ref{fig:f1}--\ref{fig:f5} and Tables
\ref{tab:t1}--\ref{tab:t5}
contain the additional simulation results under low, low-to-moderate,
moderate, and high censoring.

\begin{figure}[!htbp]
  \centering
  \includegraphics[width=\linewidth]{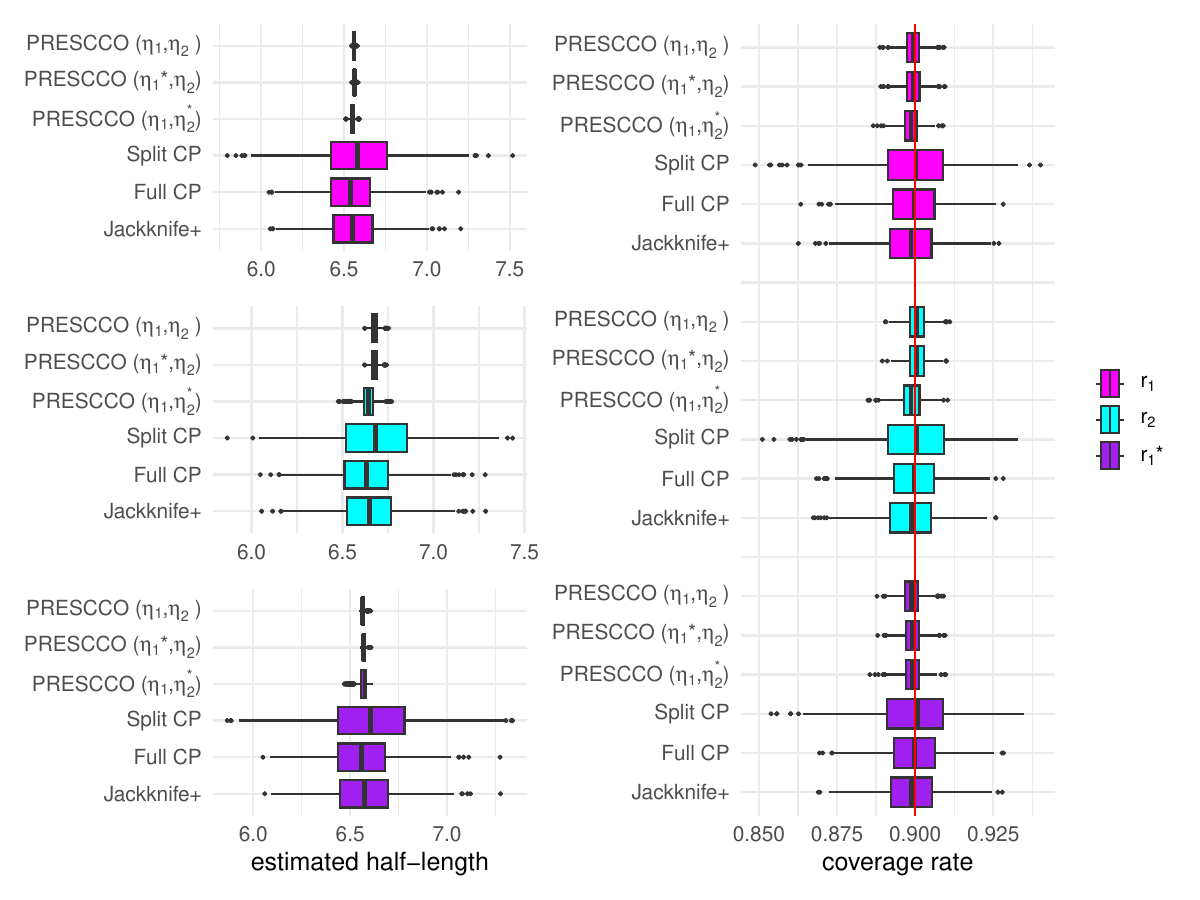}
  \caption{Boxplots of the estimated half-length and the coverage rate
    under low censoring (20--30\%) across 1,000 simulated
    observational studies. PRESCCO $(\eta_1,\eta_2)$, PRESCCO method
    with both working models correctly specified; $\eta_1^*$,
    misspecified time-to-event model; $\eta_2^\star$, misspecified
    censoring model; CP, conformal prediction.}
  \label{fig:f1}
\end{figure}

\begin{table}[!htbp]
\centering
\caption{Mean and standard deviation of the estimated half-length and
  the coverage rate under low censoring (20--30\%) across 1,000
  simulated observational studies. PRESCCO $(\eta_1,\eta_2)$, PRESCCO
  method with both working models correctly specified; $\eta_1^*$,
  misspecified time-to-event model; $\eta_2^\star$, misspecified
  censoring model; CP, conformal prediction; sd, standard deviation.}
\label{tab:t1}
\begin{tabular}{cc|cccc}
  \multicolumn{6}{c}{Low censoring}\\
 & method
 & \multicolumn{2}{c}{estimated half-length}
 & \multicolumn{2}{c}{coverage rate}\\
 & & mean & sd & mean & sd\\
\hline
\multirow{6}{*}{$r_1$}
  & PRESCCO $(\eta_1,\eta_2)$
  & 6.560 &0.006 & 0.899 &0.003 \\
  & PRESCCO $(\eta_1^*,\eta_2)$
  & 6.563 &0.006 & 0.899 &0.003\\
  & PRESCCO $(\eta_1,\eta_2^\star)$
  & 6.549 &0.014 & 0.899 &0.003 \\
  & Split CP
  & 6.587 &0.252 & 0.900 &0.014 \\
  & Full CP
  & 6.543 &0.177 & 0.899 &0.010 \\
  & Jackknife+
  & 6.556 &0.178 & 0.899 &0.010 \\
\hline
\multirow{6}{*}{$r_2$}
  & PRESCCO $(\eta_1,\eta_2)$
  & 6.678 &0.020 & 0.901 &0.003 \\
  & PRESCCO $(\eta_1^*,\eta_2)$
  & 6.676 &0.019 & 0.901 &0.003 \\
  & PRESCCO $(\eta_1,\eta_2^\star)$
  & 6.642 &0.043 & 0.899 &0.004 \\
  & Split CP
  & 6.679 &0.258 & 0.900 &0.014 \\
  & Full CP
  & 6.636 &0.183 & 0.900 &0.010 \\
  & Jackknife+
  & 6.649 &0.183 & 0.899 &0.010 \\
\hline
\multirow{6}{*}{$r_1^*$}
  & PRESCCO $(\eta_1,\eta_2)$
  & 6.566 &0.009 & 0.899 &0.003 \\
  & PRESCCO $(\eta_1^*,\eta_2)$
  & 6.572 &0.009 & 0.899 &0.003 \\
  & PRESCCO $(\eta_1,\eta_2^\star)$
  & 6.569 &0.025 & 0.899 &0.003 \\
  & Split CP
  & 6.605 &0.252 & 0.900 &0.014 \\
  & Full CP
  & 6.559 &0.177 & 0.900 &0.010 \\
  & Jackknife+
  & 6.572 &0.177 & 0.899 &0.010 \\
\end{tabular}
\end{table}
\clearpage
\newpage

\begin{figure}[!htbp]
    \centering
\includegraphics[width=\linewidth]{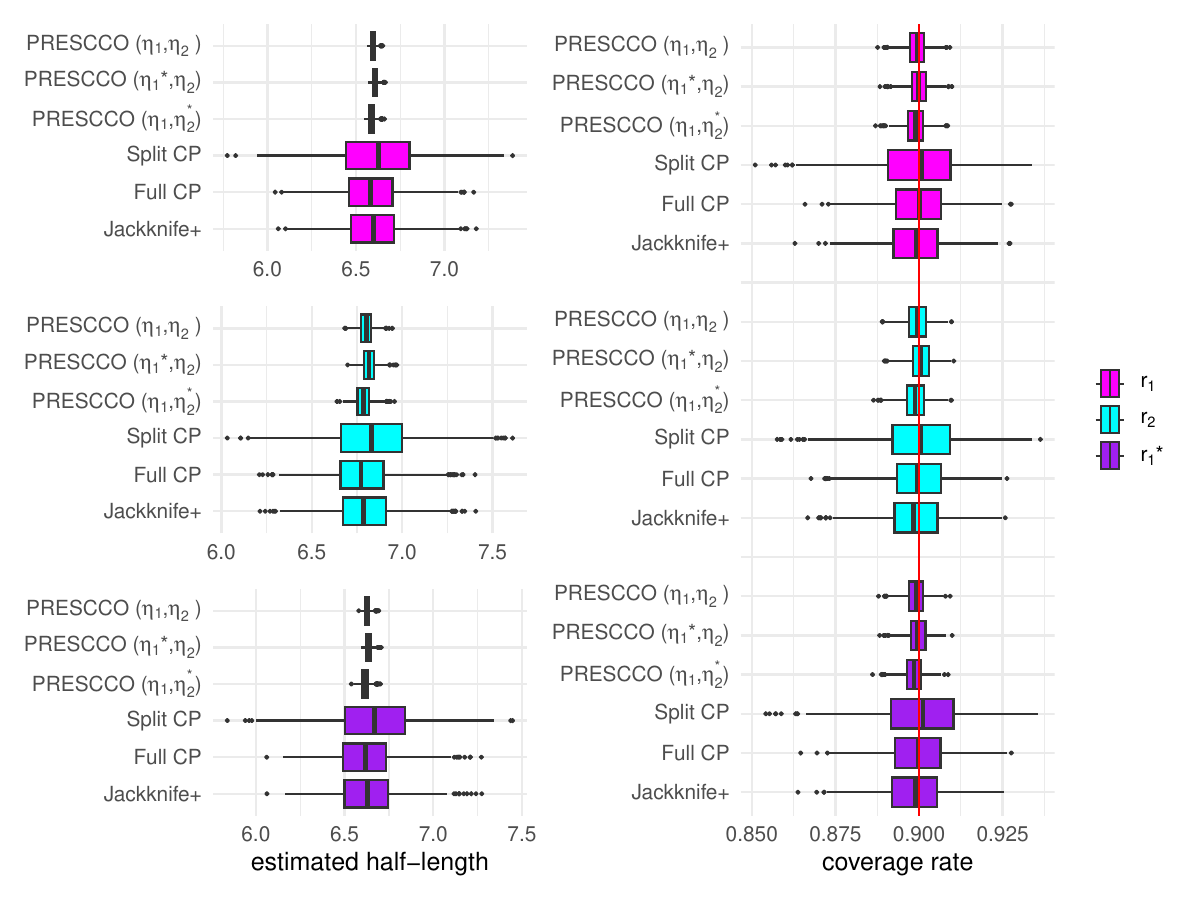}
\caption{Boxplots of the estimated half-length and the coverage
    rate under low-to-moderate censoring (30--40\%) across 1,000
    simulated observational studies. PRESCCO $(\eta_1,\eta_2)$,
    PRESCCO method with both working models correctly specified;
    $\eta_1^*$, misspecified time-to-event model; $\eta_2^\star$,
    misspecified censoring model; CP, conformal prediction.}
    \label{fig:f2}
\end{figure}

\begin{table}[!htbp]
\centering
\caption{Mean and standard deviation of the estimated half-length and
  the coverage rate under low-to-moderate censoring (30--40\%) across
  1,000 simulated observational studies. PRESCCO $(\eta_1,\eta_2)$,
  PRESCCO method with both working models correctly specified;
  $\eta_1^*$, misspecified time-to-event model; $\eta_2^\star$,
  misspecified censoring model; CP, conformal prediction; sd, standard
  deviation.}
\label{tab:t2}
\begin{tabular}{cc|cccc}
  \multicolumn{6}{c}{Low-to-moderate censoring}\\
 & method
 & \multicolumn{2}{c}{estimated half-length}
 & \multicolumn{2}{c}{coverage rate}\\
 & & mean & sd & mean & sd\\
\hline
\multirow{6}{*}{$r_1$}
  & PRESCCO $(\eta_1,\eta_2)$
  & 6.597 &0.015 & 0.899 &0.003 \\
  & PRESCCO $(\eta_1^*,\eta_2)$
  & 6.608 &0.017 & 0.900 &0.003 \\
  & PRESCCO $(\eta_1,\eta_2^\star)$
  & 6.589 &0.018 & 0.899 &0.003 \\
  & Split CP
  & 6.628 &0.256 & 0.900 &0.014 \\
  & Full CP
  & 6.588 &0.182 & 0.900 &0.010 \\
  & Jackknife+
  & 6.601 &0.182 & 0.899 &0.010 \\
\hline
\multirow{6}{*}{$r_2$}
  & PRESCCO $(\eta_1,\eta_2)$
  & 6.800 &0.041 & 0.900 &0.004 \\
  & PRESCCO $(\eta_1^*,\eta_2)$
  & 6.818 &0.042 & 0.901 &0.004 \\
  & PRESCCO $(\eta_1,\eta_2^\star)$
  & 6.786 &0.048 & 0.899 &0.004 \\
  & Split CP
  & 6.831 &0.259 & 0.900 &0.013 \\
  & Full CP
  & 6.781 &0.188 & 0.900 &0.010 \\
  & Jackknife+
  & 6.793 &0.189 & 0.899 &0.010 \\
\hline
\multirow{6}{*}{$r_1^*$}
  & PRESCCO $(\eta_1,\eta_2)$
  & 6.627 &0.018 & 0.899 &0.003 \\
  & PRESCCO $(\eta_1^*,\eta_2)$
  & 6.637 &0.019 & 0.900 &0.003 \\
  & PRESCCO $(\eta_1,\eta_2^\star)$
  & 6.615 &0.024 & 0.899 &0.003 \\
  & Split CP
  & 6.672 &0.250 & 0.901 &0.013 \\
  & Full CP
  & 6.619 &0.182 & 0.900 &0.010 \\
  & Jackknife+
  & 6.632 &0.182 & 0.899 &0.010 \\
\end{tabular}
\end{table}

\newpage
\clearpage

\begin{figure}[!htbp]
    \centering
\includegraphics[width=\linewidth]{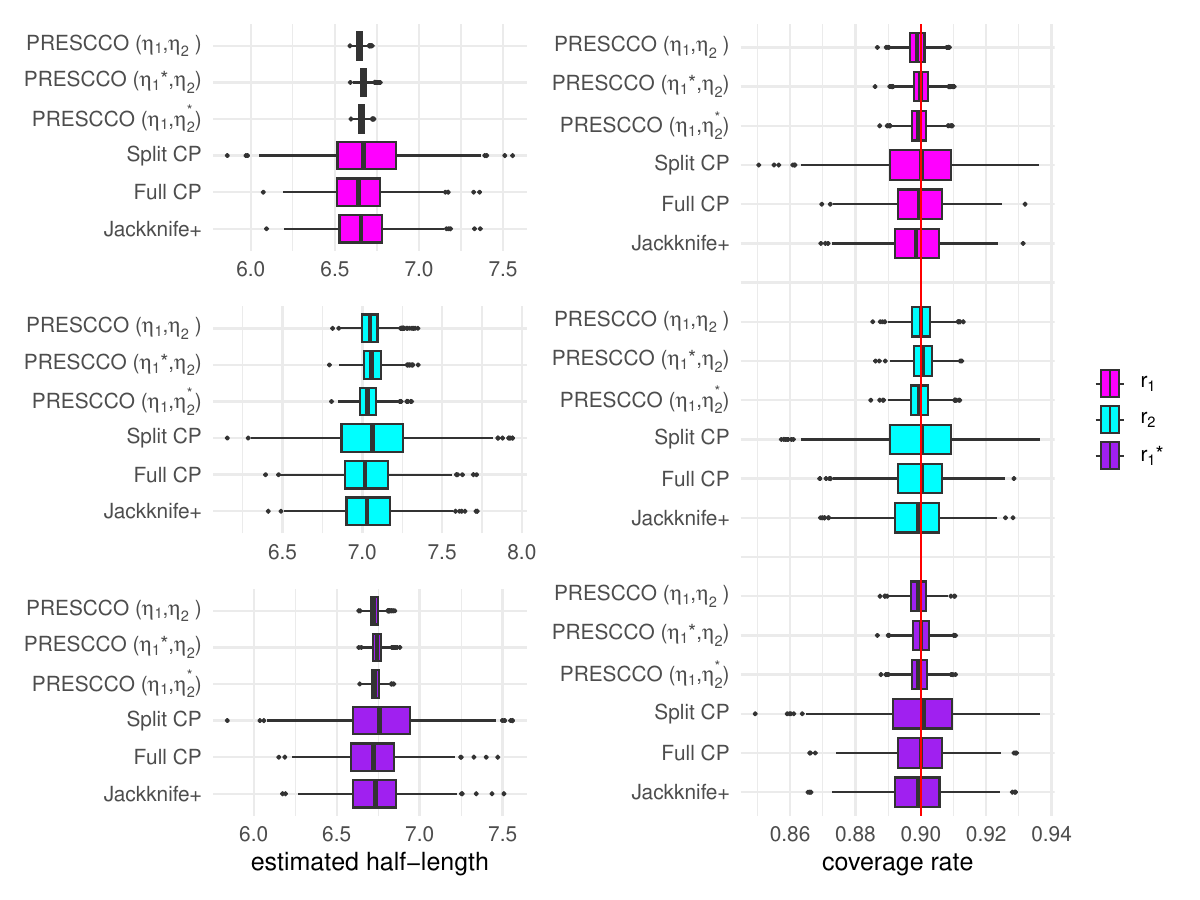}
  \caption{Boxplots of the estimated half-length and the coverage
    rate under moderate censoring (45--55\%) across 1,000 simulated
    observational studies. PRESCCO $(\eta_1,\eta_2)$, PRESCCO method
    with both working models correctly specified; $\eta_1^*$,
    misspecified time-to-event model; $\eta_2^\star$, misspecified
    censoring model; CP, conformal prediction.}
    \label{fig:f3}
\end{figure}

\begin{table}[!htbp]
\centering
\caption{Mean and standard deviation of the estimated half-length and
  the coverage rate under moderate censoring (45--55\%) across 1,000
  simulated observational studies. PRESCCO $(\eta_1,\eta_2)$, PRESCCO
  method with both working models correctly specified; $\eta_1^*$,
  misspecified time-to-event model; $\eta_2^\star$, misspecified
  censoring model; CP, conformal prediction; sd, standard deviation.}
\label{tab:t3}
\begin{tabular}{cc|cccc}
  \multicolumn{6}{c}{Moderate censoring}\\
 & method
 & \multicolumn{2}{c}{estimated half-length}
 & \multicolumn{2}{c}{coverage rate}\\
 & & mean & sd & mean & sd\\
\hline
\multirow{6}{*}{$r_1$}
  & PRESCCO $(\eta_1,\eta_2)$
  & 6.650 &0.020 & 0.899 &0.003 \\
  & PRESCCO $(\eta_1^*,\eta_2)$
  & 6.671 &0.025 & 0.900 &0.003 \\
  & PRESCCO $(\eta_1,\eta_2^\star)$
  & 6.659 &0.021 & 0.899 &0.003 \\
  & Split CP
  & 6.686 &0.260 & 0.900 &0.014 \\
  & Full CP
  & 6.649 &0.186 & 0.900 &0.010 \\
  & Jackknife+
  & 6.663 &0.186 & 0.899 &0.010 \\
\hline
\multirow{6}{*}{$r_2$}
  & PRESCCO $(\eta_1,\eta_2)$
  & 7.052 &0.079 & 0.900 &0.004 \\
  & PRESCCO $(\eta_1^*,\eta_2)$
  & 7.063 &0.082 & 0.901 &0.004 \\
  & PRESCCO $(\eta_1,\eta_2^\star)$
  & 7.038 &0.075 & 0.900 &0.004 \\
  & Split CP
  & 7.067 &0.289 & 0.900 &0.014 \\
  & Full CP
  & 7.028 &0.205 & 0.900 &0.010 \\
  & Jackknife+
  & 7.041 &0.205 & 0.899 &0.010 \\
\hline
\multirow{6}{*}{$r_1^*$}
  & PRESCCO $(\eta_1,\eta_2)$
  & 6.728 &0.032 & 0.899 &0.003 \\
  & PRESCCO $(\eta_1^*,\eta_2)$
  & 6.743 &0.034 & 0.900 &0.003 \\
  & PRESCCO $(\eta_1,\eta_2^\star)$
  & 6.733 &0.031 & 0.900 &0.003 \\
  & Split CP
  & 6.771 &0.265 & 0.900 &0.014 \\
  & Full CP
  & 6.721 &0.193 & 0.900 &0.010 \\
  & Jackknife+
  & 6.734 &0.193 & 0.899 &0.010 \\
\end{tabular}
\end{table}

\newpage
\clearpage

\begin{figure}[!htbp]
  \centering
\includegraphics[width=\linewidth]{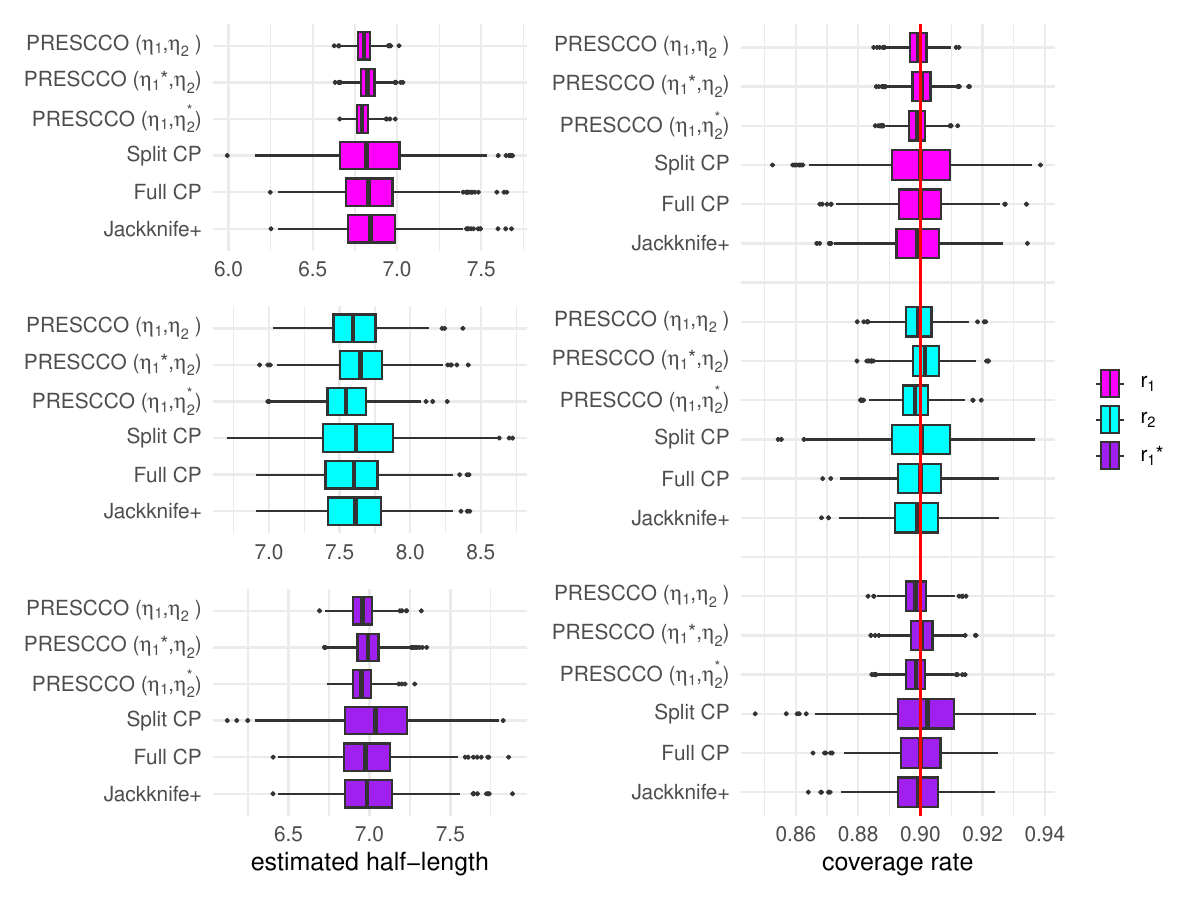}
   \caption{Boxplots of the estimated half-length and the coverage
    rate under high censoring (70--80\%) across 1,000 simulated
    observational studies. PRESCCO $(\eta_1,\eta_2)$, PRESCCO method
    with both working models correctly specified; $\eta_1^*$,
    misspecified time-to-event model; $\eta_2^\star$, misspecified
    censoring model; CP, conformal prediction.}
    \label{fig:f5}
\end{figure}

\begin{table}[!htbp]
\centering
\caption{Mean and standard deviation of the estimated half-length and
  the coverage rate under high censoring (70--80\%) across 1,000
  simulated observational studies. PRESCCO $(\eta_1,\eta_2)$, PRESCCO
  method with both working models correctly specified; $\eta_1^*$,
  misspecified time-to-event model; $\eta_2^\star$, misspecified
  censoring model; CP, conformal prediction; sd, standard deviation.}
\label{tab:t5}
\begin{tabular}{cc|cccc}
  \multicolumn{6}{c}{High censoring}\\
 & method
 & \multicolumn{2}{c}{estimated half-length}
 & \multicolumn{2}{c}{coverage rate}\\
 & & mean & sd & mean & sd\\
\hline
\multirow{6}{*}{$r_1$}
  & PRESCCO $(\eta_1,\eta_2)$
  & 6.805 &0.054 & 0.899 &0.004 \\
  & PRESCCO $(\eta_1^*,\eta_2)$
  & 6.827 &0.063 & 0.900 &0.004 \\
  & PRESCCO $(\eta_1,\eta_2^\star)$
  & 6.796 &0.049 & 0.899 &0.004 \\
  & Split CP
  & 6.838 &0.266 & 0.900 &0.014 \\
  & Full CP
  & 6.844 &0.208 & 0.900 &0.010 \\
  & Jackknife+
  & 6.857 &0.208 & 0.899 &0.010 \\
\hline
\multirow{6}{*}{$r_2$}
  & PRESCCO $(\eta_1,\eta_2)$
  & 7.602 &0.208 & 0.899 &0.006 \\
  & PRESCCO $(\eta_1^*,\eta_2)$
  & 7.649 &0.230 & 0.902 &0.006 \\
  & PRESCCO $(\eta_1,\eta_2^\star)$
  & 7.551 &0.197 & 0.898 &0.006 \\
  & Split CP
  & 7.635 &0.366 & 0.900 &0.014 \\
  & Full CP
  & 7.598 &0.268 & 0.900 &0.010 \\
  & Jackknife+
  & 7.609 &0.268 & 0.899 &0.010 \\
\hline
\multirow{6}{*}{$r_1^*$}
  & PRESCCO $(\eta_1,\eta_2)$
  & 6.959 &0.087 & 0.899 &0.005 \\
  & PRESCCO $(\eta_1^*,\eta_2)$
  & 6.994 &0.102 & 0.901 &0.005 \\
  & PRESCCO $(\eta_1,\eta_2^\star)$
  & 6.956 &0.081 & 0.899 &0.005 \\
  & Split CP
  & 7.040 &0.280 & 0.902 &0.013 \\
  & Full CP
  & 6.986 &0.217 & 0.900 &0.010 \\
  & Jackknife+
  & 6.998 &0.218 & 0.899 &0.010 \\
\end{tabular}
\end{table}

\clearpage 
\subsection{Additional results on Huntington disease data analysis}\label{sec:add-real-data}

Tables~\ref{tab:t6r1}--\ref{tab:t6r1star} report the estimated half-length
of the prediction interval, coverage rate, and band at $1-\nu=0.7,0.75,...,0.95$ under the residual functions $r_1$, $r_2$, and $r_1^*$ in the Enroll-HD data analysis.

\begin{table}
\caption{Estimated half-length of the prediction interval, coverage
rate, and band at $1-\nu=0.7,0.75,...,0.95$ under the residual function $r_1$,
for four scores recorded at study entry, based on the Enroll-HD dataset.
SC, Stroop Color; SW, Stroop Word; SI, Stroop Interference; cUHDRS,
composite Unified Huntington Disease Rating Scale. PRESCCO
$(\eta_1,\eta_2)$, PRESCCO method with both working models fitted to the
data; $\eta_1^*$, misspecified time-to-event model; $\eta_2^\star$,
misspecified censoring model; CP, conformal prediction. Coverage rates whose band does
not contain $0.9$ are marked in red.}
\label{tab:t6r1}
\centering
\scriptsize
\setlength{\tabcolsep}{3pt}
\begin{tabular}{c|cc|cc|cc|cc}
  \multirow{2}{*}{\shortstack{Value\\
  \&\\
  $1-\nu$}}
    & \multicolumn{2}{c|}{SC}
    & \multicolumn{2}{c|}{SW}
    & \multicolumn{2}{c|}{SI}
    & \multicolumn{2}{c}{cUHDRS} \\
    & \shortstack{estimated\\half-length} & \shortstack{coverage\\rate}
    & \shortstack{estimated\\half-length} & \shortstack{coverage\\rate}
    & \shortstack{estimated\\half-length} & \shortstack{coverage\\rate}
    & \shortstack{estimated\\half-length} & \shortstack{coverage\\rate} \\
\hline
\multicolumn{9}{l}{PRESCCO $(\eta_1,\eta_2)$}\\
Value
  & \multirow{7}{*}{21.904} & 0.902
  & \multirow{7}{*}{27.056} & 0.899
  & \multirow{7}{*}{16.684} & 0.901
  & \multirow{7}{*}{2.333}  & 0.896 \\
 0.7 &  & [0.890,0.913] &  & [0.887,0.911] &  & [0.889,0.912] &  & [0.884,0.908] \\
 0.75 &  & [0.889,0.915] &  & [0.886,0.912] &  & [0.888,0.914] &  & [0.883,0.909] \\
 0.8 &  & [0.887,0.916] &  & [0.884,0.913] &  & [0.886,0.915] &  & [0.881,0.911] \\
 0.85 &  & [0.885,0.918] &  & [0.882,0.915] &  & [0.884,0.917] &  & [0.879,0.913] \\
 0.9 &  & [0.883,0.920] &  & [0.880,0.917] &  & [0.882,0.919] &  & [0.877,0.915] \\
 0.95 &  & [0.880,0.924] &  & [0.876,0.921] &  & [0.879,0.923] &  & [0.873,0.919] \\
\hline
\multicolumn{9}{l}{PRESCCO $(\eta_1^*,\eta_2)$}\\
Value
  & \multirow{7}{*}{21.894} & 0.899
  & \multirow{7}{*}{27.085} & 0.902
  & \multirow{7}{*}{16.657} & 0.901
  & \multirow{7}{*}{2.342}  & 0.896 \\
 0.7 &  & [0.887,0.911] &  & [0.890,0.913] &  & [0.889,0.912] &  & [0.884,0.908] \\
 0.75 &  & [0.886,0.912] &  & [0.889,0.915] &  & [0.888,0.914] &  & [0.883,0.909] \\
 0.8 &  & [0.884,0.913] &  & [0.887,0.916] &  & [0.886,0.915] &  & [0.881,0.911] \\
 0.85 &  & [0.882,0.915] &  & [0.885,0.918] &  & [0.884,0.917] &  & [0.879,0.913] \\
 0.9 &  & [0.880,0.917] &  & [0.883,0.920] &  & [0.882,0.919] &  & [0.877,0.915] \\
 0.95 &  & [0.876,0.921] &  & [0.880,0.924] &  & [0.879,0.923] &  & [0.873,0.919] \\
\hline
\multicolumn{9}{l}{PRESCCO $(\eta_1,\eta_2^\star)$}\\
Value
  & \multirow{7}{*}{21.604} & 0.902
  & \multirow{7}{*}{26.813} & 0.900
  & \multirow{7}{*}{16.380} & 0.902
  & \multirow{7}{*}{2.266}  & 0.895 \\
 0.7 &  & [0.890,0.913] &  & [0.888,0.912] &  & [0.891,0.914] &  & [0.883,0.907] \\
 0.75 &  & [0.889,0.915] &  & [0.887,0.913] &  & [0.889,0.915] &  & [0.881,0.908] \\
 0.8 &  & [0.887,0.916] &  & [0.886,0.915] &  & [0.888,0.917] &  & [0.880,0.909] \\
 0.85 &  & [0.885,0.918] &  & [0.884,0.916] &  & [0.886,0.918] &  & [0.878,0.911] \\
 0.9 &  & [0.883,0.920] &  & [0.882,0.919] &  & [0.884,0.921] &  & [0.875,0.914] \\
 0.95 &  & [0.880,0.924] &  & [0.878,0.922] &  & [0.880,0.924] &  & [0.872,0.917] \\
\hline
\multicolumn{9}{l}{Split CP}\\
Value
  & \multirow{7}{*}{21.299} & 0.891
  & \multirow{7}{*}{25.539} & 0.880
  & \multirow{7}{*}{16.020} & 0.890
  & \multirow{7}{*}{2.307}  & 0.890 \\
 0.7 &  & [0.879,0.904] &  & {\red [0.867,0.893]} &  & [0.878,0.902] &  & [0.878,0.902] \\
 0.75 &  & [0.878,0.905] &  & {\red [0.866,0.894]} &  & [0.876,0.904] &  & [0.877,0.904] \\
 0.8 &  & [0.876,0.907] &  & {\red [0.864,0.896]} &  & [0.875,0.905] &  & [0.875,0.905] \\
 0.85 &  & [0.875,0.908] &  & {\red [0.862,0.898]} &  & [0.873,0.907] &  & [0.873,0.907] \\
 0.9 &  & [0.872,0.911] &  & [0.860,0.900] &  & [0.871,0.910] &  & [0.871,0.910] \\
 0.95 &  & [0.868,0.915] &  & [0.856,0.904] &  & [0.867,0.913] &  & [0.867,0.913] \\
\hline
\multicolumn{9}{l}{Full CP}\\
Value
  & \multirow{7}{*}{23.635} & 0.889
  & \multirow{7}{*}{27.770} & 0.886
  & \multirow{7}{*}{16.327} & 0.887
  & \multirow{7}{*}{3.128}  & 0.917 \\
 0.7 &  & [0.876,0.901] &  & {\red [0.873,0.898]} &  & {\red [0.875,0.899]} &  & {\red [0.906,0.927]} \\
 0.75 &  & [0.875,0.902] &  & {\red [0.872,0.900]} &  & [0.873,0.901] &  & {\red [0.905,0.929]} \\
 0.8 &  & [0.873,0.904] &  & [0.870,0.901] &  & [0.872,0.902] &  & {\red [0.903,0.930]} \\
 0.85 &  & [0.871,0.906] &  & [0.868,0.903] &  & [0.870,0.904] &  & {\red [0.901,0.932]} \\
 0.9 &  & [0.869,0.908] &  & [0.866,0.905] &  & [0.867,0.907] &  & [0.899,0.934] \\
 0.95 &  & [0.865,0.912] &  & [0.862,0.909] &  & [0.864,0.910] &  & [0.896,0.937] \\
\hline
\multicolumn{9}{l}{Jackknife+}\\
Value
  & \multirow{7}{*}{21.138} & 0.848
  & \multirow{7}{*}{26.915} & 0.873
  & \multirow{7}{*}{16.307} & 0.882
  & \multirow{7}{*}{2.303}  & 0.755 \\
 0.7 &  & {\red [0.834,0.862]} &  & {\red [0.860,0.886]} &  & {\red [0.870,0.895]} &  & {\red [0.739,0.772]} \\
 0.75 &  & {\red [0.832,0.864]} &  & {\red [0.858,0.887]} &  & {\red [0.868,0.896]} &  & {\red [0.737,0.774]} \\
 0.8 &  & {\red [0.831,0.865]} &  & {\red [0.857,0.889]} &  & {\red [0.867,0.898]} &  & {\red [0.735,0.776]} \\
 0.85 &  & {\red [0.829,0.868]} &  & {\red [0.855,0.891]} &  & {\red [0.865,0.900]} &  & {\red [0.732,0.779]} \\
 0.9 &  & {\red [0.826,0.870]} &  & {\red [0.852,0.893]} &  & [0.862,0.902] &  & {\red [0.729,0.782]} \\
 0.95 &  & {\red [0.821,0.875]} &  & {\red [0.848,0.897]} &  & [0.859,0.906] &  & {\red [0.724,0.787]} \\
\hline
\end{tabular}
\end{table}

\begin{table}
\caption{Estimated half-length of the prediction interval, coverage
rate, and band at $1-\nu=0.7,0.75,...,0.95$ under the residual function $r_2$,
for four scores recorded at study entry, based on the Enroll-HD dataset.
SC, Stroop Color; SW, Stroop Word; SI, Stroop Interference; cUHDRS,
composite Unified Huntington Disease Rating Scale. PRESCCO
$(\eta_1,\eta_2)$, PRESCCO method with both working models fitted to the
data; $\eta_1^*$, misspecified time-to-event model; $\eta_2^\star$,
misspecified censoring model; CP, conformal prediction. Coverage rates whose band does
not contain $0.9$ are marked in red.}
\label{tab:t6r2}
\centering
\scriptsize
\setlength{\tabcolsep}{3pt}
\begin{tabular}{c|cc|cc|cc|cc}
  \multirow{2}{*}{\shortstack{Value\\
  \&\\
  $1-\nu$}}
    & \multicolumn{2}{c|}{SC}
    & \multicolumn{2}{c|}{SW}
    & \multicolumn{2}{c|}{SI}
    & \multicolumn{2}{c}{cUHDRS} \\
    & \shortstack{estimated\\half-length} & \shortstack{coverage\\rate}
    & \shortstack{estimated\\half-length} & \shortstack{coverage\\rate}
    & \shortstack{estimated\\half-length} & \shortstack{coverage\\rate}
    & \shortstack{estimated\\half-length} & \shortstack{coverage\\rate} \\
\hline
\multicolumn{9}{l}{PRESCCO $(\eta_1,\eta_2)$}\\
Value
  & \multirow{7}{*}{22.500} & 0.903
  & \multirow{7}{*}{27.840} & 0.900
  & \multirow{7}{*}{17.274} & 0.905
  & \multirow{7}{*}{2.476}  & 0.909 \\
 0.7 &  & [0.891,0.915] &  & [0.888,0.912] &  & [0.894,0.917] &  & [0.898,0.920] \\
 0.75 &  & [0.890,0.916] &  & [0.887,0.913] &  & [0.893,0.918] &  & [0.897,0.922] \\
 0.8 &  & [0.889,0.917] &  & [0.886,0.915] &  & [0.891,0.920] &  & [0.895,0.923] \\
 0.85 &  & [0.887,0.919] &  & [0.884,0.916] &  & [0.889,0.921] &  & [0.894,0.925] \\
 0.9 &  & [0.885,0.921] &  & [0.882,0.919] &  & [0.887,0.924] &  & [0.891,0.927] \\
 0.95 &  & [0.881,0.925] &  & [0.878,0.922] &  & [0.884,0.927] &  & [0.888,0.931] \\
\hline
\multicolumn{9}{l}{PRESCCO $(\eta_1^*,\eta_2)$}\\
Value
  & \multirow{7}{*}{22.479} & 0.900
  & \multirow{7}{*}{27.884} & 0.896
  & \multirow{7}{*}{17.267} & 0.905
  & \multirow{7}{*}{2.487}  & 0.909 \\
 0.7 &  & [0.888,0.912] &  & [0.884,0.908] &  & [0.894,0.917] &  & [0.898,0.920] \\
 0.75 &  & [0.887,0.913] &  & [0.883,0.909] &  & [0.893,0.918] &  & [0.897,0.922] \\
 0.8 &  & [0.886,0.915] &  & [0.881,0.911] &  & [0.891,0.920] &  & [0.895,0.923] \\
 0.85 &  & [0.884,0.916] &  & [0.879,0.912] &  & [0.889,0.921] &  & [0.894,0.925] \\
 0.9 &  & [0.882,0.919] &  & [0.877,0.915] &  & [0.887,0.924] &  & [0.891,0.927] \\
 0.95 &  & [0.878,0.922] &  & [0.873,0.918] &  & [0.884,0.927] &  & [0.888,0.931] \\
\hline
\multicolumn{9}{l}{PRESCCO $(\eta_1,\eta_2^\star)$}\\
Value
  & \multirow{7}{*}{21.847} & 0.899
  & \multirow{7}{*}{27.111} & 0.897
  & \multirow{7}{*}{16.680} & 0.899
  & \multirow{7}{*}{2.312}  & 0.889 \\
 0.7 &  & [0.887,0.911] &  & [0.885,0.909] &  & [0.887,0.911] &  & [0.876,0.901] \\
 0.75 &  & [0.886,0.912] &  & [0.884,0.910] &  & [0.886,0.912] &  & [0.875,0.902] \\
 0.8 &  & [0.884,0.913] &  & [0.883,0.912] &  & [0.885,0.914] &  & [0.873,0.904] \\
 0.85 &  & [0.882,0.915] &  & [0.881,0.914] &  & [0.883,0.916] &  & [0.872,0.906] \\
 0.9 &  & [0.880,0.917] &  & [0.878,0.916] &  & [0.881,0.918] &  & [0.869,0.908] \\
 0.95 &  & [0.876,0.921] &  & [0.875,0.920] &  & [0.877,0.922] &  & [0.865,0.912] \\
\hline
\multicolumn{9}{l}{Split CP}\\
Value
  & \multirow{7}{*}{21.939} & 0.883
  & \multirow{7}{*}{26.356} & 0.880
  & \multirow{7}{*}{16.317} & 0.893
  & \multirow{7}{*}{2.386}  & 0.893 \\
 0.7 &  & {\red [0.870,0.895]} &  & {\red [0.867,0.893]} &  & [0.881,0.905] &  & [0.881,0.905] \\
 0.75 &  & {\red [0.869,0.897]} &  & {\red [0.866,0.894]} &  & [0.880,0.907] &  & [0.880,0.907] \\
 0.8 &  & {\red [0.867,0.898]} &  & {\red [0.864,0.896]} &  & [0.878,0.908] &  & [0.878,0.908] \\
 0.85 &  & [0.865,0.900] &  & {\red [0.862,0.898]} &  & [0.876,0.910] &  & [0.876,0.910] \\
 0.9 &  & [0.863,0.903] &  & [0.860,0.900] &  & [0.874,0.912] &  & [0.874,0.912] \\
 0.95 &  & [0.859,0.907] &  & [0.856,0.904] &  & [0.870,0.916] &  & [0.870,0.916] \\
\hline
\multicolumn{9}{l}{Full CP}\\
Value
  & \multirow{7}{*}{24.407} & 0.890
  & \multirow{7}{*}{28.896} & 0.890
  & \multirow{7}{*}{17.486} & 0.901
  & \multirow{7}{*}{3.191}  & 0.915 \\
 0.7 &  & [0.878,0.902] &  & [0.878,0.902] &  & [0.889,0.912] &  & {\red [0.904,0.926]} \\
 0.75 &  & [0.876,0.904] &  & [0.876,0.904] &  & [0.888,0.914] &  & {\red [0.903,0.927]} \\
 0.8 &  & [0.875,0.905] &  & [0.875,0.905] &  & [0.886,0.915] &  & {\red [0.902,0.929]} \\
 0.85 &  & [0.873,0.907] &  & [0.873,0.907] &  & [0.884,0.917] &  & [0.900,0.930] \\
 0.9 &  & [0.871,0.909] &  & [0.871,0.909] &  & [0.882,0.919] &  & [0.898,0.932] \\
 0.95 &  & [0.867,0.913] &  & [0.867,0.913] &  & [0.879,0.923] &  & [0.894,0.936] \\
\hline
\multicolumn{9}{l}{Jackknife+}\\
Value
  & \multirow{7}{*}{21.271} & 0.839
  & \multirow{7}{*}{27.579} & 0.877
  & \multirow{7}{*}{16.324} & 0.869
  & \multirow{7}{*}{2.376}  & 0.748 \\
 0.7 &  & {\red [0.825,0.854]} &  & {\red [0.864,0.890]} &  & {\red [0.855,0.882]} &  & {\red [0.731,0.765]} \\
 0.75 &  & {\red [0.823,0.855]} &  & {\red [0.863,0.891]} &  & {\red [0.854,0.883]} &  & {\red [0.729,0.767]} \\
 0.8 &  & {\red [0.822,0.857]} &  & {\red [0.861,0.893]} &  & {\red [0.852,0.885]} &  & {\red [0.727,0.769]} \\
 0.85 &  & {\red [0.819,0.859]} &  & {\red [0.859,0.895]} &  & {\red [0.850,0.887]} &  & {\red [0.725,0.772]} \\
 0.9 &  & {\red [0.817,0.862]} &  & {\red [0.857,0.897]} &  & {\red [0.848,0.890]} &  & {\red [0.721,0.775]} \\
 0.95 &  & {\red [0.812,0.867]} &  & [0.853,0.901] &  & {\red [0.844,0.894]} &  & {\red [0.716,0.780]} \\
\hline
\end{tabular}
\end{table}

\begin{table}
\caption{Estimated half-length of the prediction interval, coverage
rate, and band at $1-\nu=0.7,0.75,...,0.95$ under the residual function
$r_1^*$, for four scores recorded at study entry, based on the Enroll-HD
dataset. SC, Stroop Color; SW, Stroop Word; SI, Stroop Interference;
cUHDRS, composite Unified Huntington Disease Rating Scale. PRESCCO
$(\eta_1,\eta_2)$, PRESCCO method with both working models fitted to the
data; $\eta_1^*$, misspecified time-to-event model; $\eta_2^\star$,
misspecified censoring model; CP, conformal prediction. Coverage rates whose band does
not contain $0.9$ are marked in red.}
\label{tab:t6r1star}
\centering
\scriptsize
\setlength{\tabcolsep}{3pt}
\begin{tabular}{c|cc|cc|cc|cc}
  \multirow{2}{*}{\shortstack{Value\\
  \&\\
  $1-\nu$}}
    & \multicolumn{2}{c|}{SC}
    & \multicolumn{2}{c|}{SW}
    & \multicolumn{2}{c|}{SI}
    & \multicolumn{2}{c}{cUHDRS} \\
    & \shortstack{estimated\\half-length} & \shortstack{coverage\\rate}
    & \shortstack{estimated\\half-length} & \shortstack{coverage\\rate}
    & \shortstack{estimated\\half-length} & \shortstack{coverage\\rate}
    & \shortstack{estimated\\half-length} & \shortstack{coverage\\rate} \\
\hline
\multicolumn{9}{l}{PRESCCO $(\eta_1,\eta_2)$}\\
Value
  & \multirow{7}{*}{21.727} & 0.903
  & \multirow{7}{*}{26.881} & 0.899
  & \multirow{7}{*}{16.501} & 0.904
  & \multirow{7}{*}{2.291}  & 0.895 \\
 0.7 &  & [0.891,0.915] &  & [0.887,0.911] &  & [0.892,0.915] &  & [0.883,0.907] \\
 0.75 &  & [0.890,0.916] &  & [0.886,0.912] &  & [0.891,0.917] &  & [0.881,0.908] \\
 0.8 &  & [0.889,0.917] &  & [0.884,0.913] &  & [0.890,0.918] &  & [0.880,0.909] \\
 0.85 &  & [0.887,0.919] &  & [0.882,0.915] &  & [0.888,0.920] &  & [0.878,0.911] \\
 0.9 &  & [0.885,0.921] &  & [0.880,0.917] &  & [0.885,0.922] &  & [0.875,0.914] \\
 0.95 &  & [0.881,0.925] &  & [0.876,0.921] &  & [0.882,0.926] &  & [0.872,0.917] \\
\hline
\multicolumn{9}{l}{PRESCCO $(\eta_1^*,\eta_2)$}\\
Value
  & \multirow{7}{*}{21.756} & 0.903
  & \multirow{7}{*}{26.852} & 0.900
  & \multirow{7}{*}{16.493} & 0.902
  & \multirow{7}{*}{2.296}  & 0.892 \\
 0.7 &  & [0.891,0.915] &  & [0.888,0.912] &  & [0.891,0.914] &  & [0.879,0.904] \\
 0.75 &  & [0.890,0.916] &  & [0.887,0.913] &  & [0.889,0.915] &  & [0.878,0.905] \\
 0.8 &  & [0.889,0.917] &  & [0.886,0.915] &  & [0.888,0.917] &  & [0.877,0.907] \\
 0.85 &  & [0.887,0.919] &  & [0.884,0.916] &  & [0.886,0.918] &  & [0.875,0.909] \\
 0.9 &  & [0.885,0.921] &  & [0.882,0.919] &  & [0.884,0.921] &  & [0.872,0.911] \\
 0.95 &  & [0.881,0.925] &  & [0.878,0.922] &  & [0.880,0.924] &  & [0.869,0.915] \\
\hline
\multicolumn{9}{l}{PRESCCO $(\eta_1,\eta_2^\star)$}\\
Value
  & \multirow{7}{*}{21.606} & 0.903
  & \multirow{7}{*}{26.784} & 0.899
  & \multirow{7}{*}{16.367} & 0.902
  & \multirow{7}{*}{2.259}  & 0.895 \\
 0.7 &  & [0.891,0.915] &  & [0.887,0.911] &  & [0.891,0.914] &  & [0.883,0.907] \\
 0.75 &  & [0.890,0.916] &  & [0.886,0.912] &  & [0.889,0.915] &  & [0.881,0.908] \\
 0.8 &  & [0.889,0.917] &  & [0.884,0.913] &  & [0.888,0.917] &  & [0.880,0.909] \\
 0.85 &  & [0.887,0.919] &  & [0.882,0.915] &  & [0.886,0.918] &  & [0.878,0.911] \\
 0.9 &  & [0.885,0.921] &  & [0.880,0.917] &  & [0.884,0.921] &  & [0.875,0.914] \\
 0.95 &  & [0.881,0.925] &  & [0.876,0.921] &  & [0.880,0.924] &  & [0.872,0.917] \\
\hline
\multicolumn{9}{l}{Split CP}\\
Value
  & \multirow{7}{*}{21.109} & 0.891
  & \multirow{7}{*}{25.451} & 0.878
  & \multirow{7}{*}{16.055} & 0.898
  & \multirow{7}{*}{2.285}  & 0.892 \\
 0.7 &  & [0.879,0.904] &  & {\red [0.866,0.891]} &  & [0.886,0.910] &  & [0.879,0.904] \\
 0.75 &  & [0.878,0.905] &  & {\red [0.864,0.893]} &  & [0.885,0.911] &  & [0.878,0.905] \\
 0.8 &  & [0.876,0.907] &  & {\red [0.863,0.894]} &  & [0.883,0.912] &  & [0.877,0.907] \\
 0.85 &  & [0.875,0.908] &  & {\red [0.861,0.896]} &  & [0.881,0.914] &  & [0.875,0.909] \\
 0.9 &  & [0.872,0.911] &  & {\red [0.858,0.899]} &  & [0.879,0.917] &  & [0.872,0.911] \\
 0.95 &  & [0.868,0.915] &  & [0.854,0.903] &  & [0.875,0.920] &  & [0.869,0.915] \\
\hline
\multicolumn{9}{l}{Full CP}\\
Value
  & \multirow{7}{*}{23.270} & 0.887
  & \multirow{7}{*}{27.124} & 0.886
  & \multirow{7}{*}{16.124} & 0.881
  & \multirow{7}{*}{3.097}  & 0.917 \\
 0.7 &  & {\red [0.875,0.900]} &  & {\red [0.873,0.898]} &  & {\red [0.868,0.894]} &  & {\red [0.906,0.927]} \\
 0.75 &  & [0.873,0.901] &  & {\red [0.872,0.900]} &  & {\red [0.867,0.895]} &  & {\red [0.905,0.929]} \\
 0.8 &  & [0.872,0.902] &  & [0.870,0.901] &  & {\red [0.865,0.897]} &  & {\red [0.903,0.930]} \\
 0.85 &  & [0.870,0.904] &  & [0.868,0.903] &  & {\red [0.863,0.899]} &  & {\red [0.901,0.932]} \\
 0.9 &  & [0.867,0.907] &  & [0.866,0.905] &  & [0.861,0.901] &  & [0.899,0.934] \\
 0.95 &  & [0.864,0.911] &  & [0.862,0.909] &  & [0.857,0.905] &  & [0.896,0.937] \\
\hline
\multicolumn{9}{l}{Jackknife+}\\
Value
  & \multirow{7}{*}{20.875} & 0.854
  & \multirow{7}{*}{27.038} & 0.883
  & \multirow{7}{*}{16.368} & 0.881
  & \multirow{7}{*}{2.270}  & 0.761 \\
 0.7 &  & {\red [0.840,0.868]} &  & {\red [0.870,0.895]} &  & {\red [0.868,0.894]} &  & {\red [0.745,0.778]} \\
 0.75 &  & {\red [0.838,0.869]} &  & {\red [0.869,0.897]} &  & {\red [0.867,0.895]} &  & {\red [0.743,0.780]} \\
 0.8 &  & {\red [0.837,0.871]} &  & {\red [0.867,0.898]} &  & {\red [0.865,0.897]} &  & {\red [0.741,0.782]} \\
 0.85 &  & {\red [0.835,0.873]} &  & [0.865,0.900] &  & {\red [0.863,0.899]} &  & {\red [0.738,0.785]} \\
 0.9 &  & {\red [0.832,0.876]} &  & [0.863,0.903] &  & [0.861,0.901] &  & {\red [0.735,0.788]} \\
 0.95 &  & {\red [0.828,0.880]} &  & [0.859,0.907] &  & [0.857,0.905] &  & {\red [0.730,0.793]} \\
\hline
\end{tabular}
\end{table}

\clearpage

\end{document}